\documentclass[11pt]{article}

\usepackage[margin=1in]{geometry}
\usepackage{microtype}%if unwanted, comment out or use option "draft"

\usepackage{ bbold }
\usepackage{graphics}

\usepackage{pst-node}
\usepackage{tikz-cd}
\usepackage{enumitem}
\usepackage{wrapfig,color}

\usepackage{amsmath}
\usepackage{amssymb}
\usepackage{amsthm}
\usepackage{xspace}
\usepackage[normalem]{ulem}
\usepackage{graphicx}
\usepackage{epsfig}
\usepackage{ifpdf}
\usepackage{soul,hyperref}
\usepackage{latexsym}
\usepackage[mathscr]{euscript}
\usepackage{xspace}
\usepackage{color}
\usepackage{makeidx}
\usepackage{wrapfig,color}
\usepackage{stackrel}
\usepackage[]{algorithm}
\usepackage[noend]{algpseudocode}
 \usepackage[toc,page]{appendix}
\usepackage{subcaption}
\usepackage[all]{xy}

\usepackage{mathtools}

\usepackage{extarrows}
\usepackage{todonotes}
\newtheorem{problem}{Problem}

\newcommand{\simon}[1]    {\textcolor{red}{\emph{#1}}}
\newcommand{\Bi} {\mathrm{Bi}}
\newcommand{\eps}{\varepsilon}

\newcommand{\WS}{{\sf WS}}

\newtheorem{conjecture}{Conjecture}
\newtheorem{theorem}{Theorem}[section] % section

\newtheorem{corollary}[theorem]{Corollary}

\newtheorem{proposition}[theorem]{Proposition}
\newtheorem{definition}[theorem]{Definition}

\newtheorem{example}{Example}[section]
\newtheorem*{theorem*}{Theorem}
\newtheorem*{proposition*}{Proposition}
\newtheorem*{lemma*}{Lemma}
\newtheorem*{remark*}{Remark}
\newtheorem*{example*}{Example}

\newcommand*{\rom}[1]{\expandafter\@slowromancap\romannumeral #1@}

\definecolor{darkred}{rgb}{1, 0.1, 0.3}
\definecolor{darkblue}{rgb}{0.1, 0.1, 1}

\begin{document}

\newcommand{\footremember}[2]{%
   \footnote{#2}
    \newcounter{#1}
    \setcounter{#1}{\value{footnote}}%
}
\newcommand{\footrecall}[1]{%
    \footnotemark[\value{#1}]%
} 

\title{\Large Approximating 1-Wasserstein Distance between Persistence Diagrams by Graph Sparsification\footremember{support}{This work has been partially supported by NSF grants CCF 1839252 and 2049010}}
%\author{Tamal Dey \and Simon Zhang}
\author{\large Tamal K. Dey\footremember{school}{Purdue University, Department of Computer Science, USA}\and \large Simon Zhang\footrecall{school}}
\date{}

\maketitle

% Copyright Statement
% When submitting your final paper to a SIAM proceedings, it is requested that you include
% the appropriate copyright in the footer of the paper.  The copyright added should be
% consistent with the copyright selected on the copyright form submitted with the paper.
% Please note that "20XX" should be changed to the year of the meeting.

% Default Copyright Statement
%\fancyfoot[R]{\scriptsize{Copyright \textcopyright\ 2022 by SIAM\\ Unauthorized reproduction of this article is prohibited}}

% Depending on which copyright you agree to when you sign the copyright form, the copyright
% can be changed to one of the following after commenting out the default copyright statement
% above.

%\fancyfoot[R]{\scriptsize{Copyright \textcopyright\ 20XX\\
%Copyright for this paper is retained by authors}}

%\fancyfoot[R]{\scriptsize{Copyright \textcopyright\ 20XX\\
%Copyright retained by principal author's organization}}

%\pagenumbering{arabic}
%\setcounter{page}{1}%Leave this line commented out.

\begin{abstract} Persistence diagrams (PD)s play a central role in topological data analysis. This analysis requires
computing distances among such diagrams such as the $1$-Wasserstein distance. Accurate computation of these PD distances for large data sets that render large
diagrams may not scale appropriately with the existing methods. The main source of difficulty
ensues from the size of the bipartite graph on which a matching needs to be computed for
determining these PD distances. We address this problem by making several algorithmic and computational observations. By exploiting the  metric on the plane, we can obtain\textbf{, in theory, a near-linear fully polynomial-time approximation scheme.} This is \textit{theoretically optimal} assuming the $(1+\epsilon)$-approximate EMD conjecture in constant dimension, which is that the EMD problem on the plane cannot be approximated by a PTAS in time $O(\frac{1}{\epsilon^2}n)$ up to polylog factors. %$1+O(\eps)$.
In our implementation, first, taking advantage of the distribution of PD points, we \emph{condense} them thereby decreasing the number of nodes in the graph for computation. The increase in point multiplicities is addressed by reducing the matching problem to a min-cost flow problem on a transshipment network. Second, we use Well Separated Pair Decomposition to sparsify the graph to a size that is linear in the number of points. Both node and arc sparsifications contribute to
the approximation factor where we leverage a lower bound given by the Relaxed Word Mover's distance. Third, we eliminate bottlenecks during the sparsification procedure by introducing parallelism. %and furthermore address stalling of network simplex, the solver of min-cost flow, by introducing a stopping heuristic. 
Fourth, we develop an open source software called 
%PDoptFlow
\footnote{\url{https://github.com/simonzhang00/pdoptflow}}{{\sc PDoptFlow}} 
based on our algorithm, exploiting parallelism by GPU and multicore. We perform extensive experiments and show that the actual empirical error is very low.
We also show that we can achieve high performance at low guaranteed relative errors, improving upon the state of the arts. 
%Obtaining such a speed at a low error is critical for applications such as exact nearest neighbor search on PDs.  
%Furthermore we use PDoptFlow for the $1$ nearest neighbor ($1$-NN) search problem and show a $O(1)$ approximate nearest neighbor approximation bound. We also perform extensive experiments that show high performance in execution time over the state-of-the-arts, such as Hera, GPU-sinkhorn, and dense network simplex implementations at low guaranteed relative errors. On the $1$-NN PD search problem, upon replacing exact $W_1$ computation, PDoptFlow on low sparsity graphs can break its computational bottleneck. %We also demonstrate an application of PDoptFlow to finding 1 nearest neighbor ($1$-NN) in the space of PDs. For this application, PDoptFlow can be used to replace exact $W_1$ computation using graphs of low sparsity. Such exact computation forms the bottleneck to computing $1$-NN.
\end{abstract}

\section{Introduction}
\label{sec: introduction}
A standard processing pipeline in topological data analysis (TDA) converts data, such as a point cloud or a function on it, to
 a topological descriptor called the persistence diagram (PD) by a persistence algorithm 
\cite{edelsbrunner2000topological}. See books~\cite{DW22,edelsbrunner2010computational} for
a general introduction to TDA. 
Two PDs are compared by
 computing a distance between them. By the stability theorem of PDs \cite{cohen2007stability,skraba2020wasserstein}, close distances between shapes
 or functions on them imply close distances between their PDs; thus, computing
 diagram distances efficiently becomes important. It can help an increasing list of
 applications such as clustering \cite{davies2020fuzzy, lacombe2018large, marchese2017k}, classification \cite{carriere2017sliced,le2021entropy,som2018perturbation} and deep learning \cite{ wangtopogan} that have found the use of topological persistence for analyzing data.
 The 1-Wasserstein ($W_1$) distance is a common distance to compare persistence diagrams; {\sc Hera} \cite{kerber2017geometry} is a widely used open source software for this. Others include \cite{morozov2012dionysus,scikittda2019}. In this paper, we develop a new approach and its efficient software implementation for computing
 the $1$-Wasserstein distance called here the $W_1$-distance that improves the state-of-the-art.

 \section{Background}
  Here are the notations we will use in this paper. The terminology in the meaning will be discussed when the notations are introduced.
 \begin{table}[h] {
	%\setlength{\extrarowheight}{3pt}
	%\fontsize{20}{20}\selectfont
	
	%\footnotesize
	\centering
\begin{tabular}{
|p{4cm}|p{10.0cm}|}
 \hline
 \multicolumn{2}{|c|}{\small Notations} \\
 \hline
 
 Symbol & Meaning \\ %&\scriptsize theor. error $= 0.2$ {\sc hera}  & \scriptsize N.S.\\
 \hline
 $\mathbb{R}, \mathbb{N}, \mathbb{Z}$ & Denote the real, natural and integer numbers \\
 $\mathbb{R}^+=\{x \in \mathbb{R}: x>0\}$ & Denote the positive reals\\
 $A,B$ & input PDs\\
 
 $\tilde{A}, \tilde{B}$ & multiset of points on $\mathbb{R}^2$ (nondiagonal points of $A$ and $B$)\\
 
 $\Delta$ & set of diagonal points \\
 
 $\tilde{A}_{proj}, \tilde{B}_{proj}$ & multisets of projections of $\tilde{A}, \tilde{B}$ to $\Delta$ \\
 
 $\hat{A}, \hat{B}$ & sets of points corresponding to 
 $\tilde{A}, \tilde{B}$ \\
 
 $\bar{a}, \bar{b}$ & virtual points that represent $\tilde{A}_{proj}$ and 
 $\tilde{B}_{proj}$\\
 
 $\hat{A}^{\delta}$, $\hat{B}^{\delta}$ & $\delta$ condensation of $\hat{A}$ and $\hat{B}$ \\
 
 $\sigma, c, f$ & supply, cost and flow functions of a transshipment network \\
 
 $L, \delta$ & a lower bound to the $W_1$-distance, additive error \\
 
 {\sc WCD}, {\sc RWMD} & word centroid distance, relaxed word movers distance \\
 
 $s, \eps, n$ & sparsification factor, theoretical relative error, and $|\tilde{A} \cup \tilde{B}|$\\
 
 $G(A,B)$ & bipartite transportation network on $\hat{A} \cup \{\bar{b}\}$ and $\hat{B} \cup \{\bar{a}\}$\\
 
 $G_{\delta}$ & $G(\hat{A}^{\delta} \cup \{\bar{b}\}, \hat{B}^{\delta} \cup \{\bar{a}\})$ \\
 
 $\WS_s(\hat{A}^{\delta}\cup\hat{B}^{\delta})$ & s-WSPD on $(\hat{A}^{\delta}\cup\hat{B}^{\delta})$ \\
 
 $\WS_s^{PD}(A^{\delta},B^{\delta})$ & sparsified transshipment network induced by $\WS_s(\hat{A}^{\delta}\cup\hat{B}^{\delta})$ \\
 
 $W_1(A,B)$ & ground truth $W_1$-distance
 \\
 \hline
 
 %\\
 
 \end{tabular}
 \caption{Notations used in this paper.}
 \label{table: notation}
 }
 \end{table}
 
 We discuss here some of the basic background concepts from complexity theory. We will first  discuss asymptotic analysis. This is used to measure complexity at scale. We also discuss approximation algorithms in the context of asymptotic complexity. \subsection{Asymptotic Analysis}
 We define a \textbf{(multivariate)  monomial} in $d$ variables of degree $(k_1,...,k_d), k_i \in \mathbb{R}^+, \forall i= 1,...,d$ is a function $f: \mathbb{R}^d \rightarrow \mathbb{R}$ of the form:
 \begin{equation}
     f(x_1,...,x_d)=C_{k_1,...,k_d}\Pi_{i=1}^d x_i^{k_i}, C_{k_1,...,k_d} \in \mathbb{R}
 \end{equation}
 Traditionally a monomial is defined with integer powers on the variables $x_i, i=1,...,d$. 
 For asymptotic analysis, however, we only care about the large scale behavior of such functions. This is why the integer power assumption is not needed. 
 \begin{definition}
     We define a \textbf{(multivariate) polynomial} as a sum of finitely many multivariate monomials. 
 \end{definition}
 A (multivariate) polynomial composed with a logarithm, we denoted this (multivariate) function as a "\textbf{polylog}."

 When doing algorithmic analysis, we are interested in the behavior of the algorithmic complexity at scale. If we have a function $g: \mathbb{N} \rightarrow \mathbb{N}$ that computes the complexity of an algorithm in terms of "input size" $n \in \mathbb{N}$, we can describe its complexity by an asymptotic bound by a simpler function. This simpler function is usually a polynomial on the input size. We will also discuss when this polynomial has a polylog multiplicative factor.

 The conventional upper bound for complexity is given by big-O notation:
 \begin{definition}
 \textbf{Big-O Asymptotics}
 
 For two functions $g: \mathbb{N} \rightarrow \mathbb{N}$ and $f: \mathbb{N} \rightarrow \mathbb{N}$ we have that:
 \begin{equation}
     g(n)=O(f(n)) \text{ iff }
 \end{equation}
 \begin{equation} \exists C \in \mathbb{R}: C>0, \exists N \in \mathbb{N}, \forall n \geq N, g(n) \leq C f(n)
 \end{equation}
  \end{definition}
    When the simpler function $f$ upper bounds $g$ without leaving any positive constant in the large $n$ limit, we have little-o notation:
  \begin{definition}
 \textbf{Little-o Asymptotics}
 
 For two functions $g:  \mathbb{N} \rightarrow \mathbb{N}$ and $f: \mathbb{N} \rightarrow \mathbb{N}$ we have that:
 \begin{equation}
     g(n)=o(f(n)) \text{ iff }
 \end{equation}
 \begin{equation} \frac{g(n)}{f(n)} \rightarrow 0, n \rightarrow \infty
 \end{equation}
  \end{definition}
  For lower bounds, we often use Big-Omega notation:
  \begin{definition}
 \textbf{Big-Omega Asymptotics}
 
 For two functions $g:  \mathbb{N} \rightarrow \mathbb{N}$ and $f: \mathbb{N} \rightarrow \mathbb{N}$ we have that:
 \begin{equation}
     g(n)=\Omega(f(n)) \text{ iff }
 \end{equation}
 \begin{equation} \exists N \in \mathbb{N}, \forall n \geq N, g(n) \geq C f(n)
 \end{equation}
  \end{definition}
  We will use the following notation for parameterized algorithms. The complexity of such algorithms are determined by a parameter $\varepsilon \in \mathbb{R}$ along with the usual input size $n$:
  \begin{definition}
 \textbf{Big-O Asymptotics up to PolyLog Factors in the input size and Polynomial Factors in a Parameter}
 
 For two functions $g: \mathbb{R} \times \mathbb{N} \rightarrow \mathbb{N}$ and $f: \mathbb{R} \times \mathbb{N} \rightarrow \mathbb{N}$ we have that:
 \begin{equation}
     g(\epsilon,n)=\tilde{O}(f(\frac{1}{\epsilon},n)) \text{ iff }
 \end{equation}
 $\exists k \in \mathbb{N}, \exists C \in \mathbb{R}: C>0, \exists h \text{ a polynomial of order $k$ respectively},$
 \begin{equation} \exists N \in \mathbb{N}, \forall n \geq N, g(\epsilon,n) \leq C h(\log_2(n))f(\frac{1}{\epsilon},n)
 \end{equation}
  \end{definition}
\subsection{Approximation Scheme}
We formally define what an optimization problem is and optimization algorithms that compute optimization problems.

\begin{definition}
    A \textbf{(minimization/maximization)  Optimization Problem} is defined by a triple $(I,S,c)$ of data instances, a solution space and a cost function $c: S \rightarrow \mathbb{R}$. For every data instance $x \in I$ there is some corresponding solution $s(x) \in S$ where $c(s(x)) \in \mathbb{R}$ is minimized/maximized.

\end{definition}
In order to solve an optimization problem, we define optimization algorithms that can provide solutions to data instances of an optimization problem.
\begin{definition}
An \textbf{Optimization Algorithm} $\mathcal{A}$ is a function that on data instances $x \in I$ that computes a corresponding solution $s(x) \in S$ of an optimization problem $\mathcal{P}=(I,S,c)$. An optimization algorithm is \textbf{value-returning} if it only computes the number $c(s(x)) \in \mathbb{R}$
\end{definition}
We measure the (time/space) \textbf{complexity} of an algorithm by the number of operations the algorithm must perform for a given input size. The \textbf{asymptotic complexity} of an algorithm is the asymptotic behavior of the (time/space) complexity of an algorithm.

We can also approximate solutions of an optimization problem up to a distortion on the cost. We define such algorithms here:
\begin{definition}
An \textbf{Approximation Scheme} is an algorithm $\mathcal{A}$ for a problem $\mathcal{P}$ that is a function from parameters in $\mathbb{R}$ and data instances so that:
\begin{equation}
\forall x \in I, 
    \forall \epsilon \in \mathbb{R},  (1-\epsilon)\min_{s(x) \in S}c(s) \leq c(\mathcal{A}(x))\leq (1+\epsilon)\min_{s(x) \in S}c(s)
\end{equation}
\end{definition}
We give approximation schemes with certain polynomial asymptotic complexity bounds a name: 
\begin{definition}(\cite{garey1978strong})
    A \textbf{Polynomial-Time Approximation Scheme} (PTAS) is an approximation scheme where the asymptotic complexity of the approximation scheme is $\tilde{O}(n^k)$ for input size $n \in \mathbb{N}$ and $k \in \mathbb{R}$ a constant.

    A PTAS is called a \textbf{Fully Polynomial-Time Approximation Scheme} (FPTAS)   if the asymptotic complexity of the approximation scheme is $\tilde{O}(f(\frac{1}{\epsilon},n))$ for $f$ some multivariate polynomial with constant degree.
\end{definition}
We took the liberty of allowing for polylog factors in the PTAS and FPTAS definitions, although this is not traditional. 
  
 \section{Existing Algorithms and Our Approach}
 As defined in Section \ref{sec: w1formulation}, the $W_1$-distance between PDs is the assignment problem on a bipartite graph~\cite{burkard2012assignment}. This is the problem of minimizing the cost of a perfect matching on it. Thus, any algorithm that solves this problem~\cite{bernstein2021deterministic,bertsekas1988auction, brand2021minimum, jonker1986improving,lee2013path, manne2014new} can solve the exact $W_1$-distance between PDs problem.% The best known theoretical complexity for solving the assignment problem on a bipartite graph is $\tilde{O}(n^2)$~\cite{brand2021minimum} where $n$ is the total number of points constituting the PDs to be compared.
 
 Different algorithms computing $W_1$-distance between PDs have been implemented for open usage which we briefly survey here. For $\eps>0$, the software {\sc hera}~\cite{kerber2017geometry} gives a (1+$\eps$) approximation to the $W_1$-distance by solving a bipartite matching problem using the auction algorithm in $\tilde{O}(\frac{n^{2.5}}{\eps})$ time.
 %It was shown empirically~\cite{kerber2017geometry} that the sequential Gauss-Seidel version of the auction algorithm performs better than a parallel Jacobi implementation. On the other hand, parallelism is vital to our sparsification algorithm which significantly reduces computational bottlenecks. 
 In software GUDHI \cite{maria2014gudhi}, the problem is solved exactly by leveraging a dense min-cost flow implementation from the POT library \cite{divol2019understanding,flamary2017pot} to solve the assignment problem. %The sinkhorn algorithm can also be used to replace the dense min-cost flow algorithm. 
 The sinkhorn algorithm for optimal transport has a time complexity of $\tilde{O}(\frac{n^2}{\eps^2})$ \cite{chakrabarty2020better} but requires $O(n^2)$ memory and incurs numerical errors for small $\eps$. %the additive error of the approximation. 
 The $O(n^2)$ memory requirement is demanding for large $n$, especially on GPU. It was shown in \cite{chen2021approximation} that the {\sc quadtree}~\cite{indyk2003fast} and {\sc flowtree}~\cite{backurs2020scalable} algorithms can be adapted to achieve a $O(\log A)$ approximation in $O(n  \log A)$ memory and time where $A$ is the aspect ratio, which is the ratio of the largest pairwise distance between PD points divided by their closest pairwise distance. One has no control over the error with this approach and in practice the approximation factor is large. The sliced Wasserstein Distance achieves an upper bound on the error with a factor of $2\sqrt{2}$ in $O(n^2 \log n)$ time \cite{carriere2017sliced}.
 Table \ref{table: complexities} shows the complexities and approximation factors of {\sc PDoptFlow} and other algorithms. 
 
\begin{table}[h] 
	%\setlength{\extrarowheight}{3pt}
	%\fonts}ize{20}{20}\selectfont
	
	%\footnotesize
	%\centering
\begin{tabular}{
|p{5cm}|p{4cm}|p{3cm}|p{3cm}|}
 \hline
 \multicolumn{4}{|c|}{\small Algorithm Complexities and Approximation Factors} \\
 \hline
 
 Algorithm & Time (Sequential) & Memory & Approx. Bound\\ %&\scriptsize theor. error $= 0.2$ {\sc hera}  & \scriptsize N.S.\\
 \hline
 {\sc hera} & $\tilde{O}(\frac{n^{2.5}}{\eps})$ & $O(n)$ & $(1+\eps)$ \\
 
 dense MCF & $\tilde{O}(n^3)$ & $O(n^2)$ & exact\\
 
 sinkhorn & $\tilde{O}(\frac{n^2}{\eps^2})$ & $O(n^2)$ &$\eps$ abs. err\\
 
 flowtree, quadtree & $O(n\log A)$ & $O(n\log A)$ & $O(\log A)$ \\
 
 WCD, RWMD & $O(n)$ and $O(n\sqrt{n})$  & $O(n)$ & none\\ 
 
 sliced Wasserstein & $O(n^2\log n)$ & $O(n^2)$ & $2\sqrt{2}$
 \\
 
 PDoptFlow & $\tilde{O}(\frac{\hat{n}^2}{\eps^2})$ & $O(\max(\frac{\hat{n}}{\eps^2}, n))$ & $1+O(\eps)$\\
 \hline
 
 %\\
 
 \end{tabular}
 \caption{$\hat{n}\leq n$ depends on $n$, the total number of points. A better bound of $\tilde{O}(\hat{n}^2/\eps)$ for {\sc PDoptFlow} is possible with a tighter spanner, see Section \ref{sec: WSPD} for the reasoning behind our spanner choice.}
 \label{table: complexities}
 \end{table}

\subsection{Our Approach}
We design an algorithm that achieves a $(1+O(\eps))$ approximation to $W_1$-distance. The input to our algorithm is two PDs and a sparsity parameter $s$ with $\eps= O(1/s)$. 

Our approach is centered around the following theoretical result for the complexity of computing the $W_1$-distance between PDs.
    \begin{theorem}\label{thm: W1-sparse-complexity}
    (Main Theorem for the Complexity of Computing the $W_1$-distance)
    
	Let $\varepsilon>0$ and $A=\tilde{A}\cup \Delta,B=\tilde{B}\cup \Delta$ two PDs of atmost $n$ points, 
    
The $W_1$-distance can be reduced to computing a min-cost flow on a sparse network. This can theoretically be computed in time $O(\frac{1}{\epsilon^2}n\log(n))$.
\end{theorem}
See Section \ref{sec: bounds} for a proof. 

This means that there exists a near linear value-returning FPTAS for the $W_1$ distance between persistence diagrams, which is our main claim.

The computational complexity in Theorem \ref{thm: W1-sparse-complexity} \textbf{is optimal} assuming the EMD problem on $\mathbb{R}^2$ cannot be approximated with a PTAS in $\tilde{O}(\frac{1}{\epsilon^2}n^{1+o(1)-\delta})$ time for any $\delta>0$. See Theorem \ref{thm: 1+epsilon-EMDlowerbound}.

\textbf{Implementation: }
In order to achieve this,  
the problem is reduced to a min-cost flow problem on a sparsified transshipment network with sparsification determined by $s$. We use two geometric ideas to sparsify the nodes and arcs of the transshipment network. We are able to construct networks of linear complexity while availing high parallelism. The min-cost flow problem is implemented with the network simplex algorithm. This lowers the inherent complexity of the network simplex routine, and enables us to gain speedup using the GPU and multicore executions over existing implementations. %for a $(1+O(\eps))$-approximation. 
%Surprisingly, we can also obtain speedups when the multiplicities are low on average though the speedup is less striking (see Section \ref{sec: experiments-PDoptFlow}).

%if a flow based algorithm is used, we must keep track of the flow on each edge and thus deal with the data movement of $R$-subgraph construction. We address this computational bottleneck with two approaches. 
We apply a simple geometric idea called $\delta$-condensation (see Figure \ref{fig: delta-condensation}) to reduce the number of nodes in the transshipment network. This approach is synonymous to "grid snapping" \cite{fasy2018approximate,mumey2018indexing} or "binning" \cite{lacombe2018large} to a $\delta$-grid where $\delta$ depends on $s$. In order to maintain a $(1+O(\eps))$-approximation, we use a lower bound given by the Relaxed Word Mover's distance \cite{kusner2015word}. Its naive sequential computation can be a bottleneck for large PDs. We parallelize its computation with parallel nearest neighbor queries to a kd-tree data structure.%, obtaining a fast $O(\sqrt{n})$ depth and $O(n \cdot \sqrt{n})$ work complexity.

%parallel algorithm for the lower bound.%Furthermore, since the worst case max flow instance is when all capacities are 1, we use

 In existing flow-based approaches ~\cite{cuturi2013sinkhorn,flamary2017pot} that compute the $W_1$-distance, the cost matrix is stored and processed incurring a quadratic memory complexity. We address this issue by reducing the number of arcs to $O(s^2n)$ using an $s$-well separated pair decomposition ($s$-WSPD) ($s$ is the algorithm's sparsity parameter) where $n$ is the number of nodes. This requires $O(s^2n)$ memory. Moreover, we parallelize WSPD construction in the pre-min-cost flow computation since it is a computational bottleneck. This can run in time $O(polylog(s^2n))$ according to \cite{wang2021fast}. Thus, the pre-min-cost flow computation of our algorithm incurs $O(s^2n)$ cost. %where $n$ is the number of nodes in the transshipment network and $s$ is the WSPD parameter.
 We focus on the $W_1$-distance instead of the general $q$-Wasserstein distance since we can use the triangle inequality for a guaranteed (1+$\eps$)-spanner \cite{cabello2005matching}.
%We implement the min-cost flow problem on the sparse network with the network simplex algorithm. %For an arbitrary pivot strategy, its worst case complexity is super quadratic in the number of nodes. % and can also \emph{stall}, or not improve its objective function for exponentially many steps.
%However, it is known to converge with accurate feasible solutions early on in practice. %In fact, the execution of network simplex, where the most time is spent for our approach, can be $O(s^2n^{1.5})$ empirically for small enough $s$, see Figure \ref{fig: arcsXtime}. %Furthermore we propose a parallel block search pivot algorithm to lower the complexity of pivot searching, the bottleneck to our entire algorithm.     
%We have conducted intensive experiments to test and evaluate our algorithms. 
%We have also experimentally compared our algorithms with existing and commonly used implementations including GPU-sinkhorn, Hera, and POT's emd2. 

	%The rest of the paper explains our approach, implementation, and experiments. 

 \begin{table}[h] {
	%\setlength{\extrarowheight}{3pt}
	%\fontsize{20}{20}\selectfont
	
	%\footnotesize
	\centering
\begin{tabular}{
|p{5.0cm}||p{1.1cm}|p{1.0cm}|p{1.0cm}|p{1.0cm}|}
 \hline
 \multicolumn{5}{|c|}{\small $W_1$ Comput. Times (sec.) for Relative Error Bound} \\
 \hline
 
 &  {\sf bh} & {\sf AB} &  {\sf mri} & {\sf rips}\\ %&\scriptsize theor. error $= 0.2$ {\sc hera}  & \scriptsize N.S.\\
 
 \hline
 %\\
 
 \scriptsize Ours (th. error $=0.5$  &\small 8.058s   &\small 0.67s&\small 18.0s & \small 48.4s\\ \scriptsize {\sc hera} (th. error $= 0.5$)   & \small 405.02s & \small 10.46s &\small 1010.7s& \small 207.38s \\
 \scriptsize Ours (th. error $= 0.2$)  & \small 29.15s & \small 1.52s& \small 51.5s & \small 154s\\
 \scriptsize {\sc hera} (th. error $=0.2$)  &\small 405.02s& \small 14.56s & \small 1256.4s & \small 342.1s\\
 \scriptsize {\sc S.H.} (emp. err. $\leq 0.5$)  &\small $>$32GB& \small 3.80s & \small $>$32GB & \small $>$32GB\\
 \scriptsize dense {\sc NtSmplx} &\small $>$.3TB& \small 5.934s & \small $>$.3TB & \small 354s\\
 \scriptsize Ours, Sq. th. err. $= 0.5$   & \small 9.13s & \small 0.88s &\small 29.3s& \small 80.16s \\
 \scriptsize Ours, Sq. th. err. $= 0.2$   & \small 35.69s & \small 3.03s &\small 88.85s& \small 266.48s \\
 \hline
 \end{tabular}
 \caption{Running times of {\sc PDoptFlow}, parallel (Ours) and sequential (Ours, Sq.), against {\sc hera}, GPU-sinkhorn ({\sc S.H.}), and Network-Simplex ({\sc NtSmplx}) for $W_1$-distance; $>32$ GB or $>.3$ TB means out of memory for GPU or CPU respectively.}
 \label{table: 1-wasserstein-results}
 }
 \end{table}

\subsection{Experimental Results}
Table~\ref{table: 1-wasserstein-results} and Table~\ref{table: 1-NN-accuracy-time} summarize
the results obtained by our approach. First, we detail these results and explain the algorithms later.
 For our experimental setup and datasets, the reader may refer to Section \ref{sec: experiments}. Experiments show that our methods accelerated by GPU and multicore, or even serialized, can outperform state-of-the-art algorithms and software packages. These existing approaches include GPU-sinkhorn~\cite{cuturi2013sinkhorn}, {\sc Hera}~\cite{kerber2017geometry}, and dense network simplex~\cite{flamary2017pot} ({\sc NtSmplx}). We outperform them by an order of magnitude in total execution time on large PDs and for a given low guaranteed relative error. Our approach is implemented in the software {\sc PDoptFlow}, published at \url{https://github.com/simonzhang00/pdoptflow}.
 
 %We do not implement a sequential version of our algorithm nor do we compare a sequential version with other algorithms since that is not our proposal algorithm.
 %POT is a library for the general optimal transport problem. {\sc hera} uses the auction algorithm to approximate the $W_q$ distance.  
 
 We also perform experiments for the nearest neighbor(NN) search problem on PDs, see Problem \ref{prob: NN} in Section \ref{sec: NNproblemandbound}. This means finding the nearest PD from a set of PDs for a given query PD with respect to the $W_1$ metric. Following
 ~\cite{backurs2020scalable, chen2021approximation}, define recall@1 for a given algorithm as the percentage of nearest neighbor queries that are correct when using that algorithm for distance computation. We also use the phrase "prediction accuracy" synonymously with recall@1. Our experiments are conducted with the {\sf reddit} dataset; we allocate $100$ query PDs and search for their NN amongst the remaining $100$ PDs. We find that {\sc PDoptFlow} at $s=1$ and $s=18$ achieve very high NN recall@1 while still being fast, see Table \ref{table: 1-NN-accuracy-time}. Although at $s=1$ there are no approximation guarantees, {\sc PDoptFlow} still obtains high recall@1; see Figure \ref{fig: sxrelerror-Wasserstein} and Table \ref{table: 1-wasserstein-stats} for a demonstration of the low empirical error from our experiments. Other approximation algorithms~\cite{backurs2020scalable,chen2021approximation,kusner2015word} are incomparable in prediction accuracy though they run much faster.
 
 	 \begin{figure*}[t]
\includegraphics[width=1.0\columnwidth]{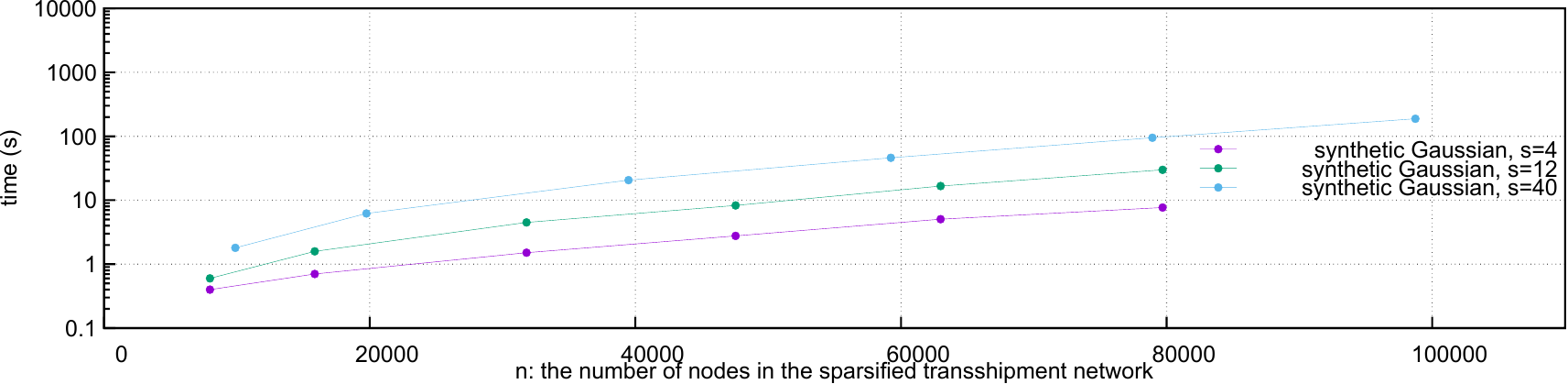}
\centering
		\caption{Plot of the empirical time (log scale) against the number of nodes $n$.} %The algorithm is empirically linear in the nusphere of arcs $m =s^2 \cdot n$. }
		\label{fig: nXtime}
		\centering
	\end{figure*}
	 In Table \ref{table: 1-wasserstein-results}, we present the total execution times for comparing four pairs of persistence diagrams: {\sf bh}, {\sf AB}, {\sf mri}, {\sf rips} from Table \ref{table: datasets} in Section \ref{sec: experiments}. %(also see Appendix~\ref{sec: datasetdesc}). 
	 The guaranteed relative error bound is given for each column. We achieve 50x, 15.6x, 56x and 4.3x speedup over {\sc hera} with the {\sf bh}, {\sf AB}, {\sf mri} and {\sf rips} datasets at a guaranteed relative error of 0.5. When this error is 0.2, we achieve a speedup of 13.9x, 9.6x, 24.4x, and 2.2x respectively on the same datasets. We achieve a speedup of up to 3.90x and 5.67x on the AB dataset over the GPU-sinkhorn and the {\sc NtSmplx} algorithm of POT respectively. Execution on {\sf rips} is aborted early by POT. We also run {\sc PDoptFlow} sequentially, doing the same total work as our parallel approach does. A slowdown of 1.1x-2.0x is obtained on {\sf bh} at $\eps=0.5$ and {\sf AB} at $\eps=0.2$ respectively compared to the parallel execution of {\sc PDoptFlow}. This suggests most of {\sc PDoptFlow}'s speedup comes from the approximation algorithm design irrespective of the parallelism. The Software from \cite{chen2021approximation} is not in Table \ref{table: 1-wasserstein-results} since its theoretical relative error ($2 \times$ (height of its quadtree)-1) \cite{backurs2020scalable, chen2021approximation} is not comparable to values ($0.5$ and $0.2$) from Table \ref{table: 1-wasserstein-results}. In fact, it has theoretical relative errors of $75, 41, 61, 41$ for {\sf bh}, {\sf AB}, {\sf mri}, {\sf rips} respectively. {\sc flowtree} \cite{chen2021approximation} is much faster than {\sc PDoptFlow} and is less accurate empirically. On these datasets, there is a 10.1x, 2.6x, 18.8x and 90.3x speedup against {\sc PDoptFlow($s=18$)} at $\eps=1.3$, for example.
	 %(see Section \ref{sec: experiments} for comparison on the $1$-NN problem). % network simplex by its implemented anti-stalling heuristic. 

\begin{table}[h] {
	%\setlength{\extrarowheight}{3pt}
	%\fontsize{20}{20}\selectfont
	%\footnotesize
	\centering
\begin{tabular}{
|p{5cm}||p{4.0cm}|p{4.0cm}|}
 \hline
 \multicolumn{3}{|c|}{\small NN PD Search for $W_1$ Time and Prediction Accuracy } \\
 \hline
 
  & \scriptsize avg. time $\pm$ std. dev. & \scriptsize avg. recall@1 $\pm$ std. dev.\\
 \hline
 %\\
 
 {{\sc quadtree}}: ($\eps=37.8 \pm 0.5$) &\small 0.46s $\pm$ 0.05s  &\small 2.2$\%$ $\pm$ 0.75$\%$\\
 {{\sc flowtree}}: ($\eps=37.8 \pm 0.5$) & \small 4.88s $\pm$ 0.2s  & \small 44$\%$ $\pm$ 4.05$\%$\\
 {{\sc WCD}} & \small 8.14s $\pm$ 2.0s & \small 39.8$\%$ $\pm$ 2.71$\%$\\
 {{\sc RWMD}}  &\small 17.16s $\pm$ 0.97s&\small 29.8$\%$ $\pm$  5.74$\%$\\
 {{\sc PDoptFlow}(s=1)}  &\small 62.6s $\pm$ 3.38s&\small 81$\%$ $\pm$ 5.2$\%$\\
 {PDoptFow(s=18)}: ($\eps=1.4$)  &\small 371.2s $\pm$ 85s &\small 95.4$\%$ $\pm$ 1.62$\%$\\
 {{\sc hera}: ($\eps=0.01$)}  &\small 2014s $\pm$ 12.6s&\small 100$\%$\\
 
 \hline
 \end{tabular}
 \caption{Total time for all $100$ NN queries and overall prediction accuracy over $5$ dataset splits of $50$/$50$ queries/search PDs; $\eps$ is the theoretical relative error. }
 \label{table: 1-NN-accuracy-time}
 }
 \end{table}
	
	Table \ref{table: 1-NN-accuracy-time} shows the total time for $100$ NN queries amongst 100 PDs in the {\sf reddit} dataset. The overall prediction accuracies using each of the algorithms are listed. See Section \ref{sec: 1-NN} for details on each of the approximation algorithms. %and Section \ref{sec: datasetdesc} in Appendix for dataset information. %\textbf{Prediction time and accuracy}:
	Table \ref{table: 1-NN-accuracy-time} shows that the algorithms 
	ordered from the fastest to the slowest on average on the {\sf reddit} dataset are  {\sc quadtree}~\cite{backurs2020scalable,chen2021approximation}, {\sc flowtree}~\cite{backurs2020scalable,chen2021approximation}, {\sc WCD}~\cite{kusner2015word}, {\sc RWMD}~\cite{kusner2015word}, {\sc PDoptFlow}($s=1$), {\sc PDoptFlow}($s=18$), and {\sc hera}~\cite{kerber2017geometry}. 
    
    Table \ref{table: 1-NN-accuracy-time} also ranks the algorithms from the most accurate to the least accurate on average as {\sc hera}, {\sc PDoptFlow}($s=18$), {\sc PDoptFlow}($s=1$), {\sc flowtree}, {\sc WCD}, {\sc RWMD}, and {\sc quadtree}. The average accuracy is obtained by $5$ runs of querying the {\sf reddit} dataset $100$ times. {\sc PDoptFlow}($s=18$) provides a guaranteed $2.3$-approximation which can even be used for ground truth distance since it computes 95\% of the NNs accurately. Furthermore, it takes only one-fourth the time that {\sc Hera} takes. %Only hera and PDoptFlow($s=18$) can achieve $\geq 90\%$ $1$-NN accuracy.
    Figure \ref{fig: pareto-boundary} shows the time-accuracy tradeoff of the seven algorithms in Table \ref{table: 1-NN-accuracy-time} on the {\sf reddit} dataset.
	%PDoptFlow at $s=18$ achieves a 4x speedup against hera while only dropping the accuracy by $5\%$. 
	%In section \ref{sec: 1-NN} a pipelining framework is discussed where these algorithms are composed together to maintain an recall@1 of $\geq 90\%$ while minimizing the time. 
	%See Figure \ref{fig: pareto-boundary} for a visualization of the Pareto frontier for the 6 algorithms. 
	%In any such pipeline, the computational bottleneck is the final stage of computation coming from algorithms with $\geq 90\%$ recall@1, such as hera and PDoptFlow($s=18$).  
	
	\begin{figure}[h]
    \includegraphics[width=0.8 \columnwidth]{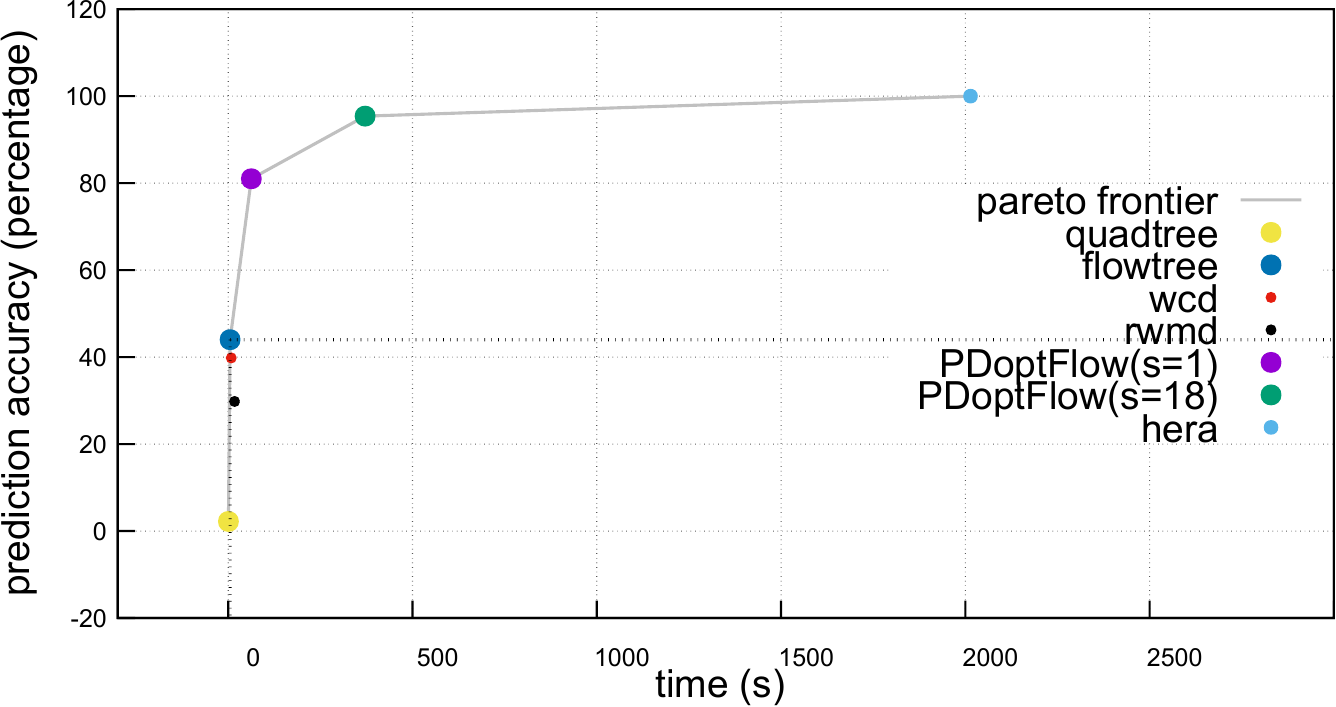}
\centering
		\caption{Pareto frontier of 7 algorithms showing the time and prediction accuracy tradeoff amongst the algorithms from Table \ref{table: 1-NN-accuracy-time} on the {\sf reddit} dataset. 
		%wcd and rwmd can be eliminated since they are Pareto dominated by {\sc flowtree} on {\sf reddit}. 
		}
		\label{fig: pareto-boundary}
		\centering
	\end{figure}
	
 Figure \ref{fig: nXtime} shows that our overall approach runs empirically in $O(s^2n^{1.5})$ time for small $s$ ($\leq 40)$. The empirical complexity improves with a smaller $s$. The datasets are given by synthetic 2D Gaussian point distributions on the plane acting as PDs. There are total of 10K, 20K, 40K, ... 100K points in the synthetic PDs. We achieve up to 20\% reduction in the total number of PD points by $\delta$-condensation. Section \ref{sec: empirical complexity} in the Appendix further explains the trend. 
 %how it was computed, and considers when $s$ is the variable and $n$ stays fixed. 
 This partly explains the speedups that Table \ref{table: 1-wasserstein-results} exhibits. For empirical relative errors, see Table \ref{table: 1-wasserstein-stats} Section \ref{sec: experiments}.
\iffalse
	\begin{figure}[h]
\includegraphics[width=1.3\columnwidth]{figures/plots/arcsXtime-crop.pdf}
		\caption{Plot of the empirical time (log  scale) depending on the number of arcs of the sparsified transshipment network for each dataset. $n$ is the number of nodes.} %The algorithm is empirically linear in the nusphere of arcs $m =s^2 \cdot n$. }
		\label{fig: arcsXtime}
		\centering 
	\end{figure}
	
	The pre-min-cost flow computation of our algorithm is guaranteed $O(m)$ in the number of arcs, as determined by the $s$-WSPD construction. Figure \ref{fig: arcsXtime} shows, in fact, that our overall approach runs empirically in $O(\sqrt{n}m)$ time, where $m=s^2 n$ with $s$ the WSPD parameter and $n$ the number of nodes in the sparsified transshipment network. The Appendix further shows the trend for varying $n$. This partially explains the speedups of Table \ref{table: 1-wasserstein-results} exhibits. For empirical relative errors see Table \ref{table: 1-wasserstein-stats} Section \ref{sec: experiments}.
\fi

	The rest of the paper explains our approach, implementation, and further experiments.

\section{1-Wasserstein Distance Problem }
\label{sec: w1formulation}
%Persistence diagrams encode the topological description of data of persistent homology over multiple scales. 
A persistence diagram is a multiset of points in the plane along with the points of infinite multiplicity on the diagonal line $\Delta$ (line with slope 1). The pairwise distances between diagonal points are assumed to be $0$. Each point $(b,d)$, $b \neq d$ in the multiset represents the birth and death time of a topological feature as computed by a persistence algorithm \cite{edelsbrunner2010computational, edelsbrunner2000topological}. Diagonal points are introduced to ascertain a stability \cite{chazal2014persistence,cohen2007stability,edelsbrunner2010computational} of PDs. 
%of PDs, which is an important property of persistent homology.
\subsection{Topological Origins of Persistence Diagrams}\label{sec: topo-origin}
Another way to define a persistence diagram is by the M\"obius inversion~\cite{mobius1832besondere} of the rank of the induced homomorphisms of the homology functor: 
\begin{equation}
    \text{dim}(\text{im}(H_{\bullet}(K_{t_i}\hookrightarrow K_{t_{j}})))
\end{equation}
over a simplicial filtration, or nested sequence of simplicial complexes, viewed as a subcategory of the category of simplicial complexes:
\begin{equation}
    (\{K_{t_i}\}_{i=1}^n, \text{inc}: K_{t_i} \hookrightarrow K_{t_{i+1}}), t_i \in \mathbb{R}, i=1,...,n
\end{equation}
For proof details and a formal explanation of persistence over a functor, see \cite{zhang2025computing}.

Notice that for a persistence diagram $P= \tilde{P}\cup \Delta$, the size complexity of $\tilde{P}$ is 
\begin{equation}
    n= \sum_{p\geq 0}^D \beta_p, \beta_p=\text{dim}(H_p(K_{t_n}))
\end{equation}
We notice that for most persistence diagrams we have that: $\beta_0>\!\!> \beta_p, p\geq 1$. We presume that this is because simplices beyond points easily close up, e.g. for a filtration of Vietoris Rips complexes \cite{zhang2020gpu}.
\subsection{Problem Formulation}
Given two PDs $A= \tilde{A} \cup \Delta$ and $B= \tilde{B} \cup \Delta$,  
$\tilde{A},\tilde{B}\subset \mathbb{R}^2\setminus \Delta$ let

\label{eq: PD-W1}
$$W_1(A,B)= \adjustlimits \inf_{\Pi: A \rightarrow B} \sum_{x_1 \in A} (\|x_1-\Pi(x_1)\|_{2}),$$
where $\Pi$ is a bijection from $A$ to $B$. Notice that this formulation is slightly different from the ones in \cite{DW22,edelsbrunner2010computational} which takes the $l_1$ and $l_{\infty}$-norms respectively instead of the $l_2$-norm considered here. It is easy to check that this is equivalent to the following formulation:

%\begin{equation}
\label{eq: PD-W1-decomposed}
$$\inf_{M \subset \tilde{A} \times \tilde{B}}{( \sum_{(x_1,x_2) \in M} (\|x_1-x_2\|_{2}) + \sum_{x_1 \notin \pi_1(M)}d_{\Delta}(x_1) + \sum_{x_2 \notin \pi_2(M)} d_{\Delta}(x_2))}
$$
%\end{equation}
where $M$ is a partial one-to-one matching between 
$\tilde{A}$ and $\tilde{B}$; $\pi_1$, $\pi_2$ are the projections of the matching $M$ onto the first and second factors, respectively; $d_{\Delta}(x)$ is the $l_{2}$-distance of $x$ to its nearest point on the diagonal $\Delta$. The triangle inequality does not hold among the points on $\Delta$. In that sense, this $W_1$-distance differs from the classical Earth Mover's Distance ({\sc EMD}) \cite{rubner2000earth} between point sets with the $l_{2}$ ground metric. Computing $W_1(A,B)$ (Problem \ref{prob: wasserstein}) reduces to the problem of finding a minimizing partial matching $M \subset \tilde{A} \times \tilde{B}$.
%and thus we cannot directly apply algorithms from geometric optimal transport e.g. \cite{agarwal2004near, agarwal2019faster, fox2019near,khesin2019preconditioning} or metric embedding methods such as in \cite{dong2019scalable, indyk2003fast}. We must deal with the non-geometric diagonal points as part of the algorithm. 
%With this definition in place, the problem we wish to solve is:
\begin{problem}\label{prob: wasserstein} %\cite{cohen2010lipschitz}
Given two PDs $A$ and $B$, Compute $W_1(A,B)$.%= (\adjustlimits \inf_{\sigma:A \leftrightarrow B} \Sigma_{(x,\sigma(x)) \in A \times B} |x-\sigma(x)|^{q}_{\infty})^{1/q}$
\end{problem}
%The $l_{2}$ norm in Problem \ref{prob: wasserstein} can be replaced by any $p$-norm. Furthermore, there can be a power on the distances such as a square, however this will fundamentally change the problem. %For Problem \ref{prob: wasserstein}, we focus on $q$=1, or the $W_1$ distance.% and assume all persistence diagrams have finitely many nondiagonal points. 

\subsection{Matching to Min-Cost Flow}\label{sec: reduction}
Let $\tilde{A}_{
proj}$, $\tilde{B}_{proj}$ be the sets of points in $\Delta$ nearest (in $l_{2}$-distance) to $\tilde{A}$, $\tilde{B}$, respectively.
Define the bipartite graph $\Bi(A,B)=(U_1 \dot{\cup} U_2, E)$ where $U_1:=\tilde{A} \cup \tilde{B}_{proj}$
and $U_2:=\tilde{B} \cup \tilde{A}_{proj}$. Define the point $p_{proj}$ to be the nearest point in $l_2$-distance to $p$ in $\Delta$ and let 
\begin{eqnarray*}
E&=&(\tilde{A}\times \tilde{B})\cup\{ (p,p_{proj})\}_{p \in \tilde{A}} \cup \{(q_{proj},q)\}_{q \in \tilde{B}} \cup (\tilde{A}_{proj}\times \tilde{B}_{proj}).
\end{eqnarray*}
The edge $e=(p,q)\in E$ has weight (i) $0$ if $e \in \tilde{A}_{proj}\times \tilde{B}_{proj}$, (ii) weight $\|p-q\|_{2}$ if $p\in \tilde{A}$, $q \in \tilde{B}$, (iii) weight $d_{\Delta}(p)$ if $q=p_{proj}$, and (iv) weight $d_{\Delta}(q)$ if $p=q_{proj}$. Because of the edges with cost 0, minimizing the total weight of a perfect matching on $Bi(A,B)$ is equivalent to finding a minimizing partial matching $M \subset \tilde{A}\times \tilde{B}$ and thus computing $W_1(A,B)$ in turn.
%Equation \ref{eq: PD-W1}.
%Our goal is to compute a perfect matching that minimizes the total weight. 
%\begin{proposition}
%There is a perfect matching of minimum total weight on $Bi(A,B)$ iff there is a partial matching solution to Equation \ref{eq: PD-W1} (Proof in Appendix)
%\end{proposition}
%\simon{Let $p_{proj}$ denote the projection of point $p \in U_1 \dot \cup U_2$ to $\Delta$. Notice that a perfect matching $P$ that minimizes the total weight cannot involve picking a pair $(p,q) \in U_1 \times U_2$ with $q\neq p_{proj}$ or $p\neq q_{proj}$ since the matching would be suboptimal. Since this is the only possible discrepency between $P$ and $M \cup \{p,p_{proj}\}_{p \in \pi_1(M)} \cup \{q_{proj},q\}_{q \in \pi_2(M)}$ we have shown that $P$ is a perfect matching of minimum total weight iff $M \cup \{p,p_{proj}\}_{p \in \pi_1(M)} \cup \{q_{proj},q\}_{q \in \pi_2(M)}$ is a solution to Equation \ref{eq: PD-W1}. Therefore Problem \ref{prob: wasserstein} can be solved by finding this perfect matching on $Bi(A,B)$.}
\begin{definition}\label{def: transhipment}
Let $G=(V,E,c,\sigma)$ be a \textbf{transshipment network}. This consists of nodes and directed edges called arcs where we have:
\begin{itemize}
    \item A supply function $\sigma: V(G) \rightarrow \mathbb{Z}$
    \item A cost function $c: V(G) \times V(G) \rightarrow \mathbb{R}^+$, and 
    \item An uncapacitated flow function
$f:V(G)\times V(G) \rightarrow \mathbb{R}$, which is defined by the following properites:
\begin{itemize}\label{eq: flow}
\item Nonnegativity: $f(u,v)\geq 0, \forall u,v \in V(G)$
\item Flow conservation out: $\sum_w f ( u , w ) = \lvert \sigma(u)\rvert$  for all  $u \in V(G)$, 
\item Flow conservation in:  $\sum_w f ( w , u ) = \lvert \sigma(u)\rvert$  for all  $u \in V(G)$ 
\end{itemize}
\end{itemize}
\end{definition}

Let $G=(V,E,c,\sigma)$ be a transshipment network made up of nodes and directed edges called arcs where we have a supply function $\sigma: V(G) \rightarrow \mathbb{Z}$, a cost function $c: V(G) \times V(G) \rightarrow \mathbb{R}^+$, and a flow function
$f:V(G)\times V(G) \rightarrow \mathbb{Z}$. Define the uncapacitated min-cost flow on $G$ as:
%\begin{definition}\label{def: mincostflow}
%\[
%\begin{equation*}
%\begin{equation}
\label{eq: mincostflow}
%\left\{
        %\begin{array}{ll}
        \iffalse
        $$
            \min_{f} \Sigma_{(u,v)\in E}\; c(u,v) \cdot f(u,v) \text{, subject to:}$$
            
            $$\Sigma_{(u,v)\in E} f(u,v)-\Sigma_{(v,w)\in E} f(v,w)=\sigma(v), \forall 
            v \in V $$
            
            $$f(u,v) \geq 0 \quad \forall (u,v) \in E $$
        \fi    
            %\end{array}
    %\right.
%\end{equation}
%\end{equation*}
%\]
%\end{definition}
$$
\label{eq: general-optimal-transport}
\displaystyle \min_{\Sigma_u f(u,v) =|\sigma(v)|, \Sigma_v f(u,v)= |\sigma(u)|, f(u,v)\geq 0} c(u,v) \cdot f(u,v), \text{where } (u,v) \in E(G).
$$
%u \in \hat{A} \cup \{\bar{b}\}, v \in \hat{B} \cup \{\bar{a}\}.
%\tamal{where are $\bar{a}$ defined? I see them much later...so it is used before definition.}
%Let us consider converting this multiset of elements into a set of unique nodes or point equivalence classes for purposes of shrinking $Bi(A,B)$'s representation. We call an element (ignoring geometric equivalence) from a persistent diagram multiset an element and a point equivalence class a node. 

Now we describe a construction of the bipartite transshipment network $G(A,B)$
for two PDs $A$ and $B$. Intuitively, $G(A,B)$ is $Bi(A,B)$ with a set instead of multiset representation for the nodes. Let $\pi_A$ and $\pi_B$ denote the mapping
of the points in $\tilde{A} \cup \tilde{B}_{proj}$ and $\tilde{B} \cup \tilde{A}_{proj}$ respectively to the nodes in the graph $G(A,B)$. %These two maps are injective on all points except the points with distance $0$.
All points with distance $0$ are mapped to the same node by $\pi_A$ and $\pi_B$. Since the diagonal points $\tilde{A}_{proj}$ and $\tilde{B}_{proj}$ are assumed to have distance zero, 
all points in $\tilde{A}_{proj}$ map to a single node, say $\bar a=\pi_A( \tilde{A}_{proj})$. Similarly, all points in $ \tilde{B}_{proj}$ map to a single node, say $\bar b=\pi_B( \tilde{B}_{proj})$ (See Figure~\ref{fig: matchingiffflow}).
We call this \emph{$0$-condensation} because it does not perturb the non-diagonal PD points. All arcs to or from $\bar{a}$ or $\bar{b}$ form \emph{diagonal arcs}, which are used in our main algorithm.% and thus incurs zero error in representing these points with the graph. We will generalize this to $\delta$-condensation later for $\delta>0$.

\begin{figure}[h]
\centering
		\includegraphics[width=0.9\columnwidth]{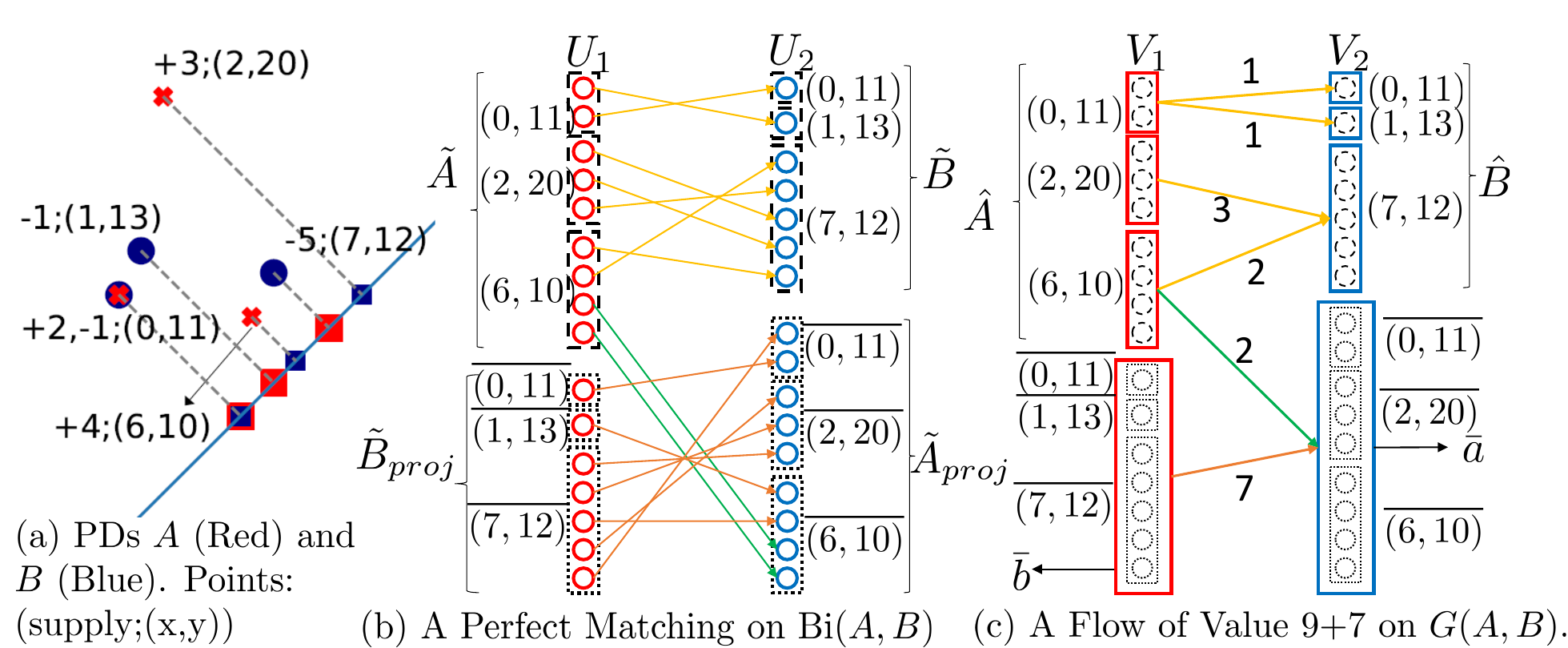}
\caption{(a) $-5;(7,12)$ means a supply of $-5$ units at point $(7,12)$. (b) $\Bi(A,B)$ with the nodes denoted by solid circles. (c) $G(A,B)$, nodes are the solid outer boxes. Supplies in $G(A,B)$ are set by the number of circles inside each box. In (b) and (c), barred-points e.g. $\overline{(7,12)}$ are projections to the diagonal.}
\label{fig: matchingiffflow}
\end{figure}
Let $\hat{A}$ and $\hat{B}$ denote the set of nodes corresponding to the non-diagonal points, that is, $\hat{A}=\pi_A(\tilde{A})$, $\hat{B}=\pi_B(\tilde{B})$. 
The nodes of $G(A,B)$ are $(\hat{A} \cup \{\bar{b}\}) \dot{\cup} (\hat{B} \cup \{\bar{a}\})$. The vertices of the transshipment network $G(A,B)$ are assigned supplies $\sigma(u)= |\pi_A^{-1}(u)|$ for $u \in U_1$ and $\sigma(v)= -|\pi_B^{-1}(v)|$ for $v \in U_2$. Intuitively, negative supply at a node means that there is a demand for a net flow at that node, which corresponds to a point in $B$. The intuition for positive supply is analogous. 
%Positive supply at a node means that the net flow must leave that node, which corresponds to a point in $A$.} Define the cost on arc $(u,v)$ as the weight of $(p,q)$ in $\Bi(A,B)$ with $p \in \pi_A^{-1}(u)$ and $q \in \pi_B^{-1}(v)$. Define the value of a flow to be the sum of all flow on arcs from $\hat{A} \cup \{\bar{b}\}$ to $\hat{B} \cup \{\bar{a}\}$.

\begin{proposition} \label{lemma: matching-iff-flow}
There is a perfect matching on $\mathrm{Bi}(A,B)$ with $|\tilde{A}|=n_1$ and $|\tilde{B}|=n_2$ iff there is a feasible flow of value $n_1+n_2$ in $G(A,B)$. %(Proof in Appendix).
\end{proposition}

\begin{proof} %Proof of Proposition \ref{lemma: matching-iff-flow}

%we can significantly shrink the size of this by just defining the construction then spend more words on \Leftarrow
$\Rightarrow$ Any perfect matching $\mu$ on $\Bi(A,B)$ can be converted to a feasible flow on $G(A,B)$ by assigning a flow between $u \in U_1$ and $v \in U_2$ equal to the number of pairs $(p,\mu(p))$ with $p \in \pi^{-1}(u)$ and $\mu(p) \in \pi^{-1}(v)$. The supplies on $G(A,B)$ are met because of the way $G(A,B)$ is constructed.
%Given a perfect matching $\mu$ on $\Bi(A,B)_d$, we can do a similar conversion of $\mu$ to a perfect flow on $N(A,B)_d$ since if there is an arc between points $p \in U_1$ and $q \in U_2$ with distance $\leq d$ there must also be an arc with infinite capacity between $\pi_A(p)$ and $\pi_B(q)$. Thus, flow between any two nodes $u\in V_1$ and $v\in V_2$ can be assigned as the number of pairs $(p,\mu(p))$ with $p \in \pi^{-1}(u)$ and $\mu(p) \in \pi^{-1}(v)$. %Since we have a bipartite network flow graph,  flow on a non-crossing arc connecting to/from a vertex $v$ from/to source/sink node is equal to the sum of all flows to/from $v$, respectively, on cross arcs (arcs between $V_1$ and $V_2$). 
%The capacity on other arcs (involving sink and source) incident to each node $u$ is also satisfied.
%This is because $\mu$ is a matching (a bijection) and by construction of $N(A,B)_d$ the capacity of the arc $(s,u)$ or $(u,t)$ equals
%$|\pi^{-1}_A(u)|$ or $|\pi^{-1}_B(u)|$ respectively. Assume, furthermore, the flow conservation property. This is all that is needed to have a feasible flow.
The value of the flow for the conversion is $n_1$+$n_2$ since there were that many pairings in the perfect matching.
%Furthermore, by the skew edge lemma, there are no "skew arcs." By construction of $N(A,B)$ and $G(A,B)$ only the projection arcs exist and thus we can also assign integer flow for the corresponding projection pairs in the perfect matching. All remaining flow on the diagonal gets sent across the single edge from $\bar{B}$ to $\bar{A}$. For $N(A,B)$ we thus have the capacity constraints satisfied and for $G(A,B)$ we have the supply constraints satisfied. The flow is thus feasible and with size $n_1+n_2$ since that is the number of pairings in the perfect matching.

$\Leftarrow$ Given a feasible flow of value $n_1+n_2$ on $G(A,B)$, we obtain a matching on $\Bi(A,B)$ by observing that we can decompose any feasible flow on arc $(u,v) \in \hat{A} \cup \{\bar{b}\} \times \hat{B} \cup \{\bar{a}\}$, into unit flows from $\pi^{-1}_A(u)$ to $\pi^{-1}_B(v)$ with no pair repeating any point from  other pairs. Each unit flow corresponds to a pair in the matching. Since the flow has $n_1+n_2$ flow value, there must be the $n_1+n_2$ pairings in the matching, making it perfect. 
\end{proof}

Problem \ref{prob: wasserstein} reduces to a min-cost flow problem on $G(A,B)$ by Proposition \ref{lemma: matching-iff-flow}. A proof based on linear algebra can be found in \cite{lacombe2018large}. %One can check that the following equation on $G(A,B)$ is also equivalent to the min-cost flow problem: %; in fact, all three of our formulations are equivalent:

\section{Approximating 1-Wasserstein Distance}
\label{sec: W1ApproxAlg}

In this section we design a $(1+O(\eps))$-approximation algorithm for Problem \ref{prob: wasserstein} that first sparsifies the bipartite graph $G(A,B)$ with an algorithm incurring a cost of $\tilde{O}_{\eps}(n)$, where $\tilde{O}_{\eps}$ hides a polylog dependence on $n$ and a polynomial dependence on $\frac{1}{\eps}$. Due to the node and edge sparsification, we must then use the min-cost flow formulation of Section \ref{sec: w1formulation} instead of a bi-partite matching for computing an approximation to the $W_1$-distance. We use the network simplex algorithm to solve the min-cost flow problem because it suits our purpose
aptly though theoretically speaking any min-cost flow algorithm can be used. 

\begin{figure}[h]
    \centering
		\includegraphics[width=0.5\columnwidth]{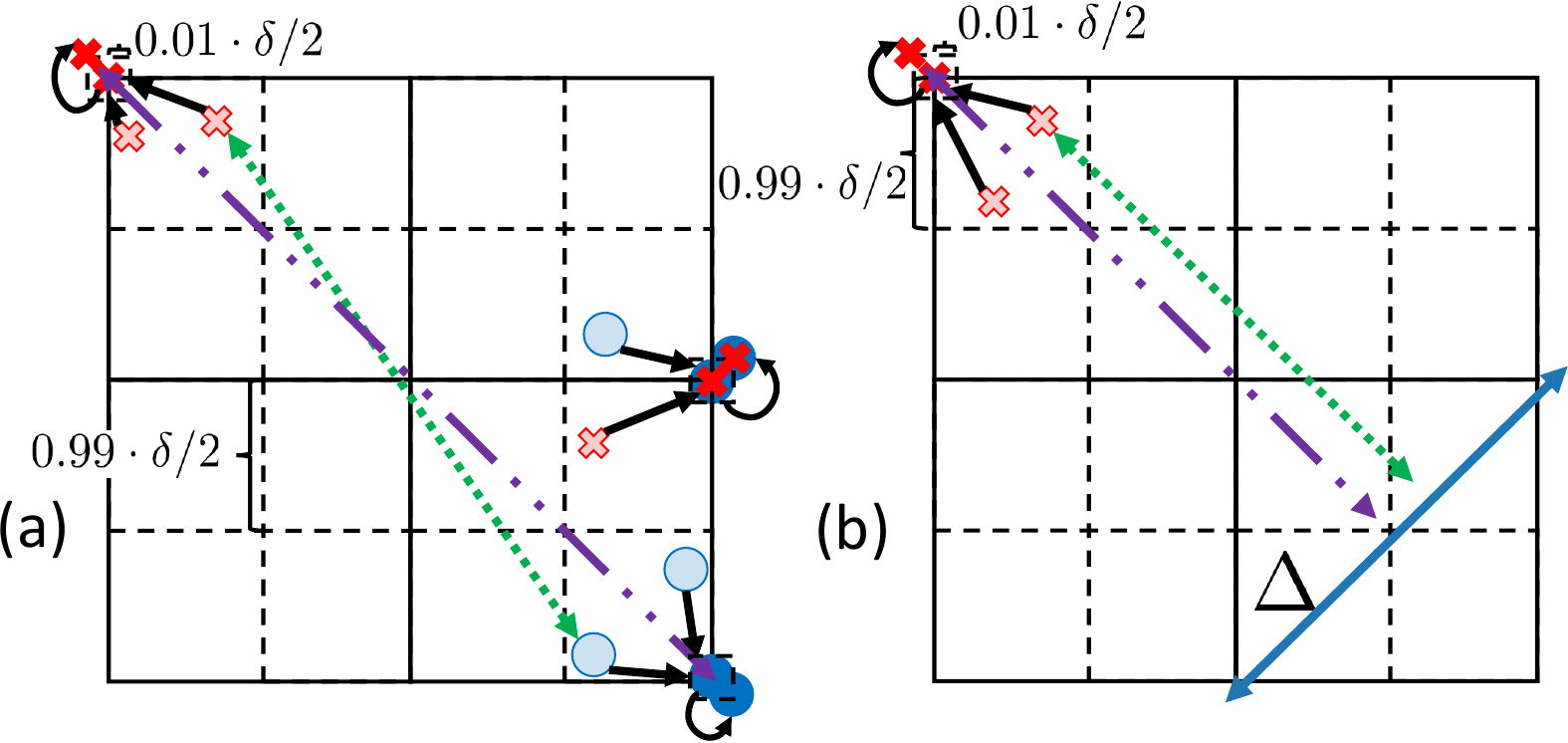}
		\caption{ $\delta$-condensation for (a) matched and (b) unmatched points. Points snapped to  their nearest $0.99 \delta$-grid point. %(corners of the 4 boxes).
		 Points are then perturbed in a $0.01 (\frac{\delta}{2})$ neighborhood. Green dotted pairwise distances change to new purple dotted and dashed pairwise distances. }
		\label{fig: delta-condensation}
		\centering
	\end{figure}
	
	\subsection{Condensation (Node Sparsification)}
\label{sec: condensation}

%In particular, we expect a small closest pair distance in the PDs for a large number of points. 
Figure \ref{fig: PDs} shows such evidences for voxelized data. %This results in many unnecessary edges and nodes for a matching-based computation of the distances, causing a high-complexity execution, even when the edges are represented implicitly as in Hera.
We draw upon a common technique for rasterizing the plane by snapping points to an evenly spaced grid to decrease the number of points. 
%We thus draw upon the common technique of rounding points to a grid to decrease the number of points. 
As discussed in Section \ref{sec: N.S. behavior}, it is known that the network simplex algorithm performs better on a transshipment network with many different arc lengths than the one with many arcs having the same length. %The points on a lattice result in symmetric distances and thus many arcs of the same cost. 
To avoid the symmetry induced by the lattice, we perturb randomly the combined points. For a $\delta>0$ and a fraction $k\geq 0.5$ (say $k=0.99$), we snap nondiagonal points to a $k\delta \cdot \mathbb{Z}\times k\delta \cdot \mathbb{Z}$ lattice (grid). Let
$\pi_{\delta}:\hat{A}\cup\hat{B} \rightarrow (k\delta \cdot \mathbb{Z})\times (k\delta \cdot \mathbb{Z}$) define this snapping of a point to its nearest $\delta$-lattice point where $\pi_\delta((x,y))=( k\delta\cdot round(\frac{x}{k\delta}), k\delta \cdot round(\frac{y}{k\delta}) )$. We follow the snapping by $\pi_{\delta}$ with a random shift of each condensed point by at most $\frac{1-k}{2} \cdot \delta$ in any of the $\pm x$ or $\pm y$ directions; see Figure \ref{fig: delta-condensation}. We call the entire procedure as "$\delta$-condensation" or "$\delta$-snapping". 
%that first projects
%the points to the grid by $\pi_{\delta}$ and then perturbs them randomly. 
The aggregate of
the points snapped to a grid point is accounted for by a supply value assigned to it; see Algorithm \ref{alg: deltacond}. %For ease of exposition we use $k=0.5$ but any $k$ within allowable machine precision can be used.

\begin{proposition} \label{prop: delta-condensation}
Let $A$ and $B$ be two PDs and $\eps > 0$. For $\delta:= \frac{2\eps L }{\sqrt{2} (|\tilde{A}|+|\tilde{B}|)}$ where $L \leq  W_1(A,B)$, let the snapping by $\pi_{\delta}$ followed by a $\delta\cdot (1-k)/2$ random shift on $A$ and $B$ produce $A^\delta$ and $B^\delta$ respectively. 
Then, 
 $(1-\eps)W_1(A,B)\leq W_1(A^\delta,B^\delta) \leq (1+\eps) W_1(A,B)$.
\end{proposition}
\begin{proof} 
After applying $\pi_{\delta}$, each point moves in a $\frac{\sqrt{2}  \delta}{2}$ neighborhood. Thus for any pair of nondiagonal points $p \in \tilde{A}$ and $q \in \tilde{B}$, the $l_2$-distance between the two points shrinks/grows at most by $\frac{2 \sqrt{2} k \delta}{2}$ units. A $\frac{(1-k) \delta}{2}$-perturbation contributes to an error of $\frac{2 \sqrt{2} (1-k)\delta}{2}$ units for the $l_2$-distance between $p$ and $q$. Thus, for a pair of nondiagonal points the additive error incurred is $\sqrt{2} \delta$ units.  Furthermore, for any nondiagonal point in either diagram, its distance to $\Delta$ can shrink/grow by at most $\frac{\sqrt{2} k \delta }{2} +\frac{\sqrt{2} (1-k) \delta}{2}=\frac{\sqrt{2} \delta}{2}$ units.    

Let $m_1$ be the number of pairs of matched nondiagonal points and $m_2$ be the number of unmatched nondiagonal points. Let the $\delta$-condensation of A and B be $A^{\delta}, B^{\delta}$ and let $\delta'= \sqrt{2} (m_1 + \frac{m_2}{2})\delta$. To reach the conclusion
of the proposition, we want $\delta$ to induce a relative error of $\eps$ for $W_1(A^\delta, B^\delta)$ with respect to $W_1(A, B)$ satisfying the following inequalities:
\begin{equation}
\label{eq: delta-to-eps-relerror}
(1-\eps) W_1(A,B) \leq W_1(A,B) - \delta' \leq  W_1(A^{\delta},B^{\delta}) \leq W_1(A,B) + \delta' \leq (1+\eps) W_1(A,B).
\end{equation}
%Furthermore, we have that $0\leq m_1 \leq \min(|\tilde{A}|,|\tilde{B}|)$ and $\max(|\tilde{A}|,|\tilde{B}|)-\min(|\tilde{A}|,|\tilde{B}|) \leq m_2 \leq |\tilde{A}|+|\tilde{B}|$. 
Observe that $m_1+\frac{m_2}{2}=\frac{ (|\tilde{A}|+|\tilde{B}|)}{2}$. Also, we have that $L \leq W_1(A,B)$. These
together constrain $\delta$ to satisfy 
$\sqrt{2} (m_1+\frac{m_2}{2}) \delta =\sqrt{2}\frac{ (|\tilde{A}|+|\tilde{B}|)}{2} \delta\leq \eps  L \leq \eps W_1(A,B)$, which gives the desired value of $\delta$
as stated.
%\simon{Observe that $m_1+\frac{m_2}{2}$ satisfies the following inequality: $\min(|\tilde{A}|,|\tilde{B}|)+\frac{(\max(|\tilde{A}|,|\tilde{B}|)-\min(|\tilde{A}|,|\tilde{B}|))}{2} \leq  m_1 +\frac{m_2}{2}\leq \frac{ (|\tilde{A}|+|\tilde{B}|)}{2}$ and that the leftmost and rightmost quantities are equal. Also we have that $L \leq W_1(A,B)$, }. Hence, solving for $\delta$ to satisfy $ \sqrt{2} (m_1+\frac{m_2}{2}) \delta \leq \eps  L \leq \eps W_1(A,B)$ %$\delta \leq \frac{\eps \cdot L}{\sqrt{2}\cdot ((|\tilde{A}|+|\tilde{B}|)/2)}$. 
%and setting $\delta$ to its upper bound in the left inequality, \tamal{what is this left inequality?} we obtain a $1 \pm \eps$ approximation.
\end{proof}

A lower bound $L$ from Proposition \ref{prop: delta-condensation} is needed in order to convert the additive error of $\delta$ to a multiplicative error of $1\pm\eps$. To find the lower bound $L$, we use the Relaxed Word Mover's distance ({\sc RWMD}) \cite{kusner2015word} that gives a lower bound for the min-cost flow of $G(A,B)$, hence for $W_1(A,B)$. There are many lower bounds that can be used such as those from \cite{atasu2019linear, kusner2015word}. However, we find {\sc RWMD} to be the most effective in terms of computational time and approximation in general. %, however we can obtain sublinear parallel complexity for the RWMD using geometry and it is a pretty good enough lower bound in practice. 
%Our fast algorithm for RWMD may be of independent interest as an approximation to Problem \ref{prob: wasserstein}, see Section \ref{sec: RWMD} in Appendix for its usage as a lower bound in computing exact k-nearest neighbors of a PD to a set of PDs with the $W_1$ metric.

Recall that {\sc RWMD} is a relaxation of one of the two constraints of the min-cost flow problem. If we "relax" or remove the constraint $\sum_{v}f(u,v)= |\sigma(u)|, u\in \hat{B} \cup \{\bar{a}\}$ from the min-cost flow formulation, we obtain the following feasible flow to the min-cost flow with one of its constraints removed 

\begin{equation*}
\label{eq: f-low}
f^{low,A}(u,v)= \left\{
        \begin{array}{ll}
            |\sigma(u)| & \quad 
            \text{if }  v= \mathrm{argmin}_{v'} \,c(u,v') \\
            0 & \quad \text{otherwise}
        \end{array}
    \right.
\end{equation*} and evaluate $L_A:= \sum_{u,v} c(u,v) \cdot f^{low,A}(u,v) $. Since $W_1(A,B)$ is a feasible solution to the relaxed min-cost flow problem, $L_A \leq W_1(A,B)$. Relaxing the constraint $\sum_u f(u,v)= |\sigma(v)|, v \in \hat{A} \cup \{\bar{b}\}$, we can define $f^{low,B}(u,v)$ and $L_B$ similarly.

Our simple parallel algorithm involves computing $L:=\max(L_A,L_B)$, the {\sc RWMD}, by exploiting the geometry of the plane via a kd-tree to perform fast parallel nearest neighbor queries. For $L_A$, we first construct a kd-tree for $\hat{B}$ viewed as points in the plane, then proceed to search in parallel for every $u \in \hat{A}$, its nearest $l_2$-neighbor $v^*$ in $\hat{B}$ while writing the quantity $c(u,v^*) \cdot f^{low,A}(u,v^*)$ to separate memory addresses. Noticing that the closest point to $\bar{b}$, is $\bar{a}$ at cost $0$, it suffices to consider the points $\hat{A}$ to compute $L_A$. %In parallel, for $u= \bar{b}$ apply a min-element sort to find the nearest point in $\hat{B}$ to $\Delta$, also while writing to a separate memory address. Letting $n= \max(|\hat{A}|,|\hat{B}|)$, the min-element sort takes $O(\log n)$ depth \cite{cole1988parallel} if done on a separate device such as GPU and 
We then apply a sum-reduction to the array of products, taking $O(\log n)$ depth \cite{blelloch2010parallel}. We apply a similar procedure for $L_B$. See Algorithm \ref{alg: RWMD}. %in Section \ref{sec: RWMD}. 
\begin{algorithm}[h]

\begin{algorithmic}[1]
%\Require{Point sets $\hat{A}, \hat{B}$}
\State build kd-tree on $\hat{B}$ using Euclidean distance on $\mathbb{R}^2$
\State compute $v^*= \mathrm{argmin}_{v \in \hat{B}} c(u,v)$  by NN search on $\hat{B}$ and store $f^{low,A}(u,v^*)$ for each $u \in \hat{A}$ in parallel

\State $L_A \gets$ compute sum-reduction of line $2$
\State $L_B \gets$ compute lines 1-3 with $\hat{A}$ and $\hat{B}$ swapped
\State
\Return $\max(L_A, L_B)$
\end{algorithmic}
\caption{{\sc RWMD}($\hat{A}, \hat{B}$, $c$)}
\label{alg: RWMD}
\end{algorithm}
Since the kd-tree queries each takes $O(\sqrt{n})$ sequential time, we obtain an algorithm with $O(n)$ processors requiring $O(\sqrt{n}+\log n)=O(\sqrt n)$ depth and $O(n\sqrt{n})$ work. 
\begin{algorithm}[H]

%\SetAlgoLined
\begin{algorithmic}[1]
\Require{PDs $A,B$, $s>0$}
%\KwOut{$\pmb{P},\sigma_{\pmb{P}}$}
\State $(\hat{A},\bar{b},\sigma_{\hat{A}}, \hat{B},\bar{a},\sigma_{\hat{B}}) \gets$ 0-condense$(A,B)$

\State $L\gets RWMD(\hat{A},\hat{B}, c)$ \Comment{$c(\cdot, \cdot)$ from Section \ref{sec: w1formulation}}

\State $\eps \gets \frac{8}{s-4}$ if $s\geq12$ else $\eps\gets 1$;
$\delta \gets \frac{2\eps L}{\sqrt{2} (|\tilde{A}|+|\tilde{B}|)}$

\State $(\hat{A}^{\delta}$, $\hat{B}^{\delta})$ $\gets$ $(\pi_{\delta}(\hat{A}), \pi_{\delta}(\hat{B}))$ \Comment{snap points of $\hat{A}$}, $\hat{B}$ to a common $ 0.99\delta$-lattice

\State $\sigma_{\hat{A}^{\delta}\cup\hat{B}^{\delta} \cup \{\bar{a}\} \cup \{\bar{b}
\}}$ $\gets$$\begin{cases} 
      \sum_{u \in \pi^{-1}_{\delta}(v)}  \sigma(u)& v \in \hat{A}^{\delta} \cup \hat{B}^{\delta} \\
      \sigma(v) & v= \{\bar{a}\} \cup \{\bar{b}\} 
   \end{cases}$

%$\sigma_{\hat{A}^{\delta}\cup\hat{B}^{\delta}}$ $\gets$ aggregate $\sigma_{\hat{A}}$ and $\sigma_{\hat{B}}$;  $\sigma_{\bar{a}}, \sigma_{\bar{b}}$;

\State perturb $\hat{A}^{\delta}\cup \hat{B}^{\delta}$ in a $\frac{0.01}{2}\delta$-radius square
\State \Return $(\hat{A}^{\delta}\cup\hat{B}^{\delta}, \sigma_{\hat{A}^{\delta}\cup\hat{B}^{\delta} \cup \{\bar{a}\} \cup \{\bar{b}
\}})$
\end{algorithmic}
 \caption{$\delta$-condensation}
\label{alg: deltacond}
\end{algorithm}
\iffalse
\begin{algorithm}[H]
\label{alg: deltacond}
\SetAlgoLined
\KwIn{PDs $A,B$, $s>0$}
%\KwOut{$\pmb{P},\sigma_{\pmb{P}}$}
$\hat{A},\bar{b},\sigma_{\hat{A}}, \hat{B},\bar{a},\sigma_{\hat{B}} \gets$ 0-condense$(A,B)$ 

$L\gets {\sc RWMD}(\hat{A}\cup
\{\bar{b}\},\hat{B}\cup\{\bar{a}\}, c)$ \Comment{$c$ is as in Section \ref{sec: w1formulation}}

$\eps \gets \frac{8}{s-4}$
$\delta \gets \frac{2\eps \cdot L}{\sqrt{2}\cdot (|\tilde{A}|+|\tilde{B}|)}$

$\hat{A}^{\delta}$, $\hat{B}^{\delta}$ $\gets$ $\pi_{\delta}(\hat{A}), \pi_{\delta}(\hat{B})$ by snapping all points of $\hat{A}$, $\hat{B}$ to a common $ 0.99\delta$-lattice

$\sigma_{\hat{A}^{\delta}\cup\hat{B}^{\delta} \cup \{\bar{a}\} \cup \{\bar{b}
\}}$ $\gets$$\begin{cases} 
      \Sigma_{u \in \pi^{-1}_{\delta}(v)}  \sigma(u)& v \in \hat{A}^{\delta} \cup \hat{B}^{\delta} \\
      \sigma(v) & v= \{\bar{a}\} \cup \{\bar{b}\} 
   \end{cases}$

%$\sigma_{\hat{A}^{\delta}\cup\hat{B}^{\delta}}$ $\gets$ aggregate $\sigma_{\hat{A}}$ and $\sigma_{\hat{B}}$;  $\sigma_{\bar{a}}, \sigma_{\bar{b}}$;

perturb $\hat{A}^{\delta}\cup \hat{B}^{\delta}$ in a $\frac{0.01}{2}\delta$-radius square neighborhood

\Return $\hat{A}^{\delta}\cup\hat{B}^{\delta}, \sigma_{\hat{A}^{\delta}\cup\hat{B}^{\delta} \cup \{\bar{a}\} \cup \{\bar{b}
\}}$
 \caption{$\delta$-condensation}
\end{algorithm}
\fi

The algorithm for $\delta$-condensation is given in Algorithm \ref{alg: deltacond}. We first gather all the points based on their $x$ and $y$ coordinates called a $0$-condensation; see Section \ref{sec: w1formulation}. Then, we compute the {\sc RWMD} in order to compute $\delta$. This $\delta$ depends on an intermediate relative error of $\eps$ for $\delta$-condensation, which depends on the input $s$. The quantity $\eps$ is chosen to be less than $1$. %It decreases as $s$ increases. 
In particular, we set $\eps \gets \frac{8}{s-4}$ if $s\geq12$ and $\eps\gets 1$ otherwise; see line 3 in Algorithm \ref{alg: deltacond}. Finally, we snap the points of $\hat{A}$ and $\hat{B}$ to the $\delta$-grid and then perturb the condensed points in a small neighborhood. The resulting sets of points are denoted $\hat{A}^{\delta}$ and $\hat{B}^{\delta}$. For each condensed point, we aggregate the supplies of points that are snapped to it. The aggregated supply function is denoted $\sigma_{\hat{A}^{\delta}\cup\hat{B}^{\delta} \cup \{\bar{a}\} \cup \{\bar{b}
\}}$. The bipartite transshipment network that could be constructed by placing arcs between all nodes from $A^{\delta}:=\hat{A}^{\delta}\cup\{\bar{b}\}$ to $B^{\delta}:=\hat{B}^{\delta}\cup\{\bar{a}\}$ is denoted as $G_{\delta}$:=$G(A^{\delta},B^{\delta})$. The cost $c_{\delta}$ is defined on arcs of  $G(A^{\delta},B^{\delta})$ as $c_{\delta}(u,v)= \|u-v\|_2$ for $u \in \hat{A}^{\delta}$ and $v \in \hat{B}^{\delta}$. The costs $c_{\delta}(u,\bar{a})$ and $c_{\delta}(\bar{b},v)$ are defined by the $l_2$-distances of $u$ and $v$ to $\Delta$ as in Section \ref{sec: reduction}. Furthermore, the supply on all points is defined by %$\sigma_{\hat{A}^{\delta}\cup\hat{B}^{\delta} \cup \{\bar{a}\} \cup \{\bar{b}\}}$.
$\sigma_{A^{\delta}\cup B^{\delta}}$. Only the nodes and supplies of this network are constructed.
\subsection{For large $n$, $\delta$-condensation collects the heavy-hitter filtration values}
\label{sec: deltacond-assumption}
Proposition \ref{alg: deltacond} has $\delta=O(\frac{1}{n})$ where $n$ is the total number of points of both PDs. In particular,
assuming $W_1(A,B)$ is bounded, we have
%\sout{\tamal{keeping $W_1(A,B)$ fixed, we get}}
$\delta \rightarrow 0$ as $n \rightarrow \infty$. In order for $\delta$-condensation to scale with $n$, we need to make an appropriate assumption about the empirical distribution of points for our PDs. Define the density for a point set $A \subseteq \mathbb{R}^2$ on a $\delta$-square grid as $\frac{\|A\|}{\|\Gamma_{\delta}\|}$ where $\|\Gamma_{\delta}\|$ is the number of nonempty grid cells with points from $A$.
\begin{proposition}\label{prop: density-ongrid-condensation}
For a PD $A$, the fraction of nodes eliminated from $A$ by $\delta$-condensation increases
if the density of a PD $A$ on a $\delta$-grid increases. 
 %\tamal{This proposition statement is troublesome. What does this increasing percentage refer to? Increasing wrt what? Is it true only for $|A|$ sufficiently large?}
\label{prop: densityfordeltacond}
\end{proposition}

\begin{proof}
For each grid point $p\in \Gamma_{\delta}$, all points in a $\delta$-square neighborhood centered at $p$ snap to $p$. These new cells partition the plane just like the original grid cells and are a translation of the original grid cells. We consider this translated grid as $\Gamma_{\delta}$, which can only affect the number of nonempty cells by at most a factor of $4$. Say a $\delta$-cell $i \in \Gamma_{\delta}$, $\delta$ depending on $\|A\|$, has $c_i$ points. %Shifting the points of the grid to the center of each square cell,
We get that exactly $c_i$ points collapse into one point. Thus $c_i-1$ points are eliminated. Adding this up over all nonempty cells $i$, we get that the fraction of nodes eliminated from $A$ is: 
\begin{equation}
\frac{\sum_{i \in \Gamma_{\delta}} (c_i-1)}{\|A\|} = \frac{\|A\|-\|\Gamma_{\delta}\|}{\|A\|}
\end{equation}
It follows that if the density $\frac{\|A\|}{\|\Gamma_{\delta}\|}$ increases, we eliminate a larger fraction of nodes as claimed. 
\end{proof}
We can directly translate the sufficient condition for Proposition \ref{prop: density-ongrid-condensation} as saying that the grid size is sublinear in the number of points $\|\Gamma_{\delta}\|=o(\|A\|)$. 

We give some usages of Proposition \ref{prop: densityfordeltacond}. As discussed in Section \ref{sec: topo-origin}, we consider the case of filtration values (times) coming from $0$-dimensional simplices:
\begin{example}
In particular, for lower star filtrations on voxel based data, we have that there are only $2^8$ possible number of filtration values to fill up, up to infinitesimal perturbations from the data. 
%\sout{In practice, these filtration values are usually perturbed infinitesimally to make the filtration injective so that, for example, gradients of the filtration values are well defined for filtration learning~\cite{pmlr-v119-hofer20b} or natural noise can be added to the data. Either way,} \simon{(I think this observation is important since we can cover the case of perturbed graph filtrations)}. 
We thus have, $\|\Gamma_{\delta}\|\leq 2^{16}$ for all $\delta$%\tamal{Why $2^{16}$ and not $2^8$?} \simon{(response: because the plane is made up of pairs of values. There are $2^8$ values per entry in a pair)}. 
, where $2^{16}$ is a counting bound on the number of pairs of filtration values that lie in $\mathbb{R}^2$.
Then, by Proposition \ref{prop: densityfordeltacond}, $\delta$-condensation scales well when $n$ is sufficiently large.

This also means that the PD stays under a constant size.
\end{example}
\begin{example}
For lower star filtrations defined on degree valued nodes of scale free networks, the degree distribution is given by the power law: $P(k) \sim k^{-\gamma}$, $2<\gamma<3$ a constant and $k$ the degree of any node. Thus, as $n \rightarrow \infty$, we sample at most $n$ times independently from this distribution. We show that the degrees sampled won't depend on the number of samples. 

Using the CDF of the power law, we get that:
\begin{equation}
    P(k<N(\gamma)) = 1-k^{-\gamma+1} \geq 0.99 \Rightarrow N(\gamma)= O(1-0.99)^{\frac{1}{-\gamma+1}}
\end{equation}
We have shown that with probability 0.99, each sample is bounded by some constant threshold $N(\gamma)= O(1-0.99)^{\frac{1}{-\gamma+1}}$ independent of $n$. Hence, $\|\Gamma_{\delta}\|$ is bounded w.h.p. and by Proposition \ref{prop: densityfordeltacond}, we have that $\delta$-condensation eliminates an eventually increasing proportion of nodes w.h.p. as $n \rightarrow \infty$.

This means that the size of the PD converges to a constant size with high probability.
\end{example}
\subsection{Well Separated Pair Decomposition(Arc Sparsification)}
\label{sec: WSPD}
The node sparsification of $G(A,B)$ gives $G_{\delta}$ whose arcs are further sparsified. Using Theorem $1$ in~\cite{cabello2005matching}, we bring the quadratic number of arcs down to a linear number by constructing a geometric $(1+\eps)$-spanner on the point set $\hat{A}^{\delta} \cup \hat{B}^{\delta}$. %that using the $(1+\eps)$-spanner on the union of the nondiagonal points of two PDs, 
For a point set $P\subset\mathbb{R}^2$, let its complete distance graph be
defined with the points in $P$ as nodes where every pair $p,q\in P$, $p\not=q$, is joined by an edge with weight equal to $\|p-q\|_2$. Define a geometric $t$-spanner $S(P)$ as a subgraph of the complete distance graph of $P$ where for any $p,q \in P, p \neq q$,
the shortest path distance $d_{SP}(p,q)$ between $p$ and $q$ in $S(P)$ satisfies the condition $d_{SP}(p,q)\leq t\cdot \|p-q\|_2$.

We compute a spanner using the well separated decomposition $s$-WSPD~\cite{chan2008well,har2011geometric}. Notice that there are many other possible spanner constructions such as $\theta$-graphs~\cite{clarkson1987approximation,keil1988approximating} and others, e.g.~\cite{gudmundsson2002fast,le2020light}. However, experimentally we find that WSPD is effective in practice, and becomes especially effective when $s$ is small. %very efficient in its construction, performs well as a pairwise distance approximator, and has a wider range of control through the $s$ value. 
The $\theta$-graphs, for example, can be an order of magnitude slower to compute as implemented in the CGAL software~\cite{fabri2009cgal}. This is theoretically justified by the $O(\log n)$ factor in the $O(n \log n)$ construction time of $\theta$-graphs when $n>1024$.
%, the set of unique points from the union of the two persistence diagrams.
An $s$-WSPD is a well known geometric construction that approximates the pairwise distances between points by pairs of "$s$-well-separated" point subsets. Two point subsets $U$ and $V$ are $s$-well separated in $l_2$-norm if there exist two $l_{2}$ balls of radius $d$ containing $U$ and $V$ that have distance at least $d \cdot s$. 
An $s$-WSPD of a point set $P\subset \mathbb{R}^2$ is a collection of pairs of $s$-well separated subsets of
$P$ so that for
every pair of points $p,q\in P$, $p \neq q$,
there exists a unique pair of subsets $U,V$ in the $s$-WSPD with $U\ni p$ and $V\ni q$. 
%The parameter $s$ controls the granularity of the well separated subsets. The larger the $s$ is, the finer (closer to a list of all point pairs) the $s$-WSPD gets. 
Each subset in an $s$-WSPD is represented by an arbitrary but fixed point in the subset. We can construct a digraph $\WS_s(P)$ 
from the $s$-WSPD on $P$ by taking the point representatives as nodes and
placing biarcs between any two nodes $u,v$, that is, creating both arcs $(u,v)$ and $(v,u)$.
It is known \cite{chan2008well, har2011geometric} that $\WS_s(P)$, viewed as an undirected graph,
is a geometric t-spanner for $t= (s+4)/(s-4)$. Putting $t=(1+\eps)$, this gives 
$s=4+\frac{8}{\eps}$. It was recently shown in \cite{de2021local} that by taking leftmost points as representatives in the well separated subsets, one can improve $t$ to $1+\frac{4}{s}+\frac{4}{s-2}$. Furthermore, it is also known that $\WS_s(P)$ has $O(s^2 n)$ number of arcs where $n= |P|$.

\begin{figure*}[h]
		\includegraphics[width=1.1\columnwidth]{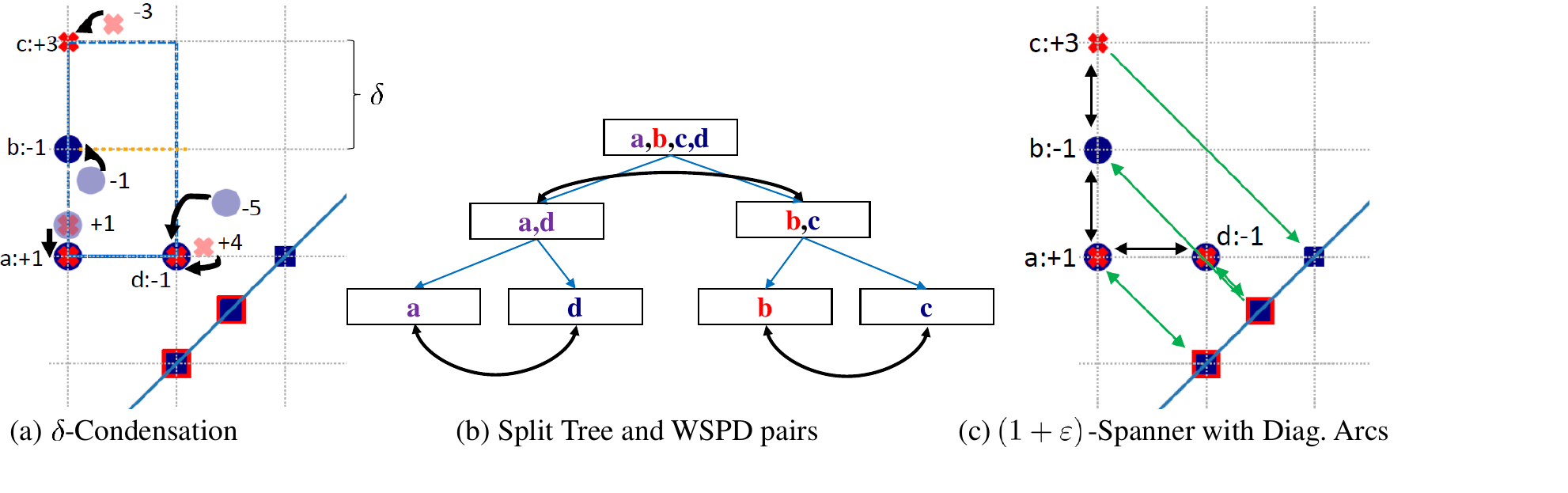}
		\centering
		\caption{Illustration of Algorithm \ref{alg: wasserstein-PDoptFlow}: (a) $\delta$-condensation for the example in Figure \ref{fig: matchingiffflow} with the split tree construction on $\hat{A}^{\delta} \cup \hat{B}^{\delta}$; (b) WSPD pairs (black biarcs) on the split tree from (a); and (c) the induced transshipment network from the WSPD with the green diagonal arcs included.}%s is 1 in the above "spanner"  %with the red diagonal nodes having supply +7 and the blue diagonal nodes having supply -9.}
		\label{fig: construction-1-Wasserstein}
		\centering
	\end{figure*}

Now we describe how we compute an arc sparsification of $G_{\delta}$. 
%where $G_{\delta}$ is the node sparsification of $G$. 
To save notations, we assume the points of $\hat{A}^{\delta}$ and $\hat{B}^{\delta}$, the $\delta$-condensation of $\hat{A}$ and $\hat{B}$ respectively, 
to be nodes also.
%First, for every non-diagonal node $u\in\hat{A}\subseteq V_1$ we take any point in $\pi_A^{-1}(u)\subseteq \tilde A$. Similarly, we take a point corresponding to each nondiagonal node $v\in \hat{B}\subseteq V_2$. To save notations, let $\hat{A}$ and $\hat{B}$ denote these points as well.
%First, we view all the nodes in $\hat{A}\cup \hat{B}$ as a set of nodes as determined by the points they represent. For every node $u\in\hat{A}\cup \hat{B}$ if there exists a $p \in \pi^{-1}_A(u)\subseteq \tilde A$ set the supply of $u$ as $s(u)$ from $G(A,B)$. If there is point $q \in \pi_B(u)$ and q coincides with  $p$, then add $s(\pi_B(q))$ from $G(A,B)$ to $s(u)$ to get the supply of $u$, otherwise $u \in \hat{B}$ and set supply of $v$ to $s(v)$ from $G(A,B)$.\sout{To save notations, let $\hat{A}$ and $\hat{B}$ denote these weighted points as well.}
We compute a $(1+\eps)$-spanner $\WS_s(\hat{A}^{\delta}\cup\hat{B}^{\delta})$ via an $s$-WSPD on the points $\hat{A}^{\delta}\cup\hat{B}^{\delta}$.
Notice that this digraph has all nodes of $G_{\delta}$ except the two diagonal nodes $\bar{a}$ and 
$\bar{b}$ which we add to it with all the original arcs from $\bar{a}$ and to $\bar{b}$ having the cost same as in $G_{\delta}$. Now we assign supplies to nodes in $\WS_s(\hat{A}^{\delta}\cup\hat{B}^{\delta})$ as in $G_{\delta}$. 
There is a caveat here. 
It may happen that points from $\hat{A}^{\delta}$ and $\hat{B}^{\delta}$ overlap.
Two such overlapped points from two sets are represented with a single point having the supply equal to
the supplies of the overlapped points added together. 
Let $\WS_s^{PD}(A^{\delta},B^{\delta})$ denote this sparsified transshipment network. Adapting an argument
in~\cite{cabello2005matching} to our case, we have: %(proof in Appendix):

%\simon{you have to aggregate the supply if a point from $\hat{A}$ coincides with a point from $\hat{B}$}

\begin{theorem}\label{theorem: Wasserstein-approx-thm}
Let $f^*$ and $\bar f^*$ be the min-cost flow values in $G_{\delta}$ and $\WS_s^{PD}(A^{\delta},B^{\delta})$ respectively where
$s$ satisfies $\eps=\frac{4}{s}+\frac{4}{s-2}$ for some $\eps>0$.
Then $f^*$ and $\bar f^*$ satisfy $f^* \leq \bar{f}^* \leq (1+\eps)f^*$.
\end{theorem}

\begin{proof}
%Let $U_1= \hat{A}^{\delta}\cup\{ \bar{b} \}$ and $U_2= \hat{B}^{\delta}\cup \{\bar{a}\}$. 
First, notice that the nodes of $\WS_s^{PD}(A^{\delta},B^{\delta})$ are exactly the same as in %$G_{\delta}= ((\hat{A}^{\delta}\cup\{ \bar{b} \}) \dot\cup (\hat{B}^{\delta}\cup \{ \bar{a} \}),(\hat{A}^{\delta}\cup\bar{b}) \times (\hat{B}^{\delta}\cup \{ \bar{a} \}), c_{\delta}, \sigma_{\hat{A}^{\delta}\cup \{ \bar{b} \} \cup \hat{B}^{\delta}\cup\bar{a}})$ 
$G(A^{\delta},B^{\delta})= (A^{\delta} \dot\cup B^{\delta}, A^{\delta} \times B^{\delta}, c_{\delta}, \sigma_{A^{\delta} \cup B^{\delta}})$
except the
overlapped nodes. We can decompose the overlapped nodes back to their original versions in $\hat{A}^{\delta}$ and $\hat{B}^{\delta}$ with biarcs of $0$-distance between them. This will also restore the supplies at each node. This does not affect $\bar f^*$.
%Consider the graph $C(V)$ for $V:= \hat{A}\dot \cup \hat{B}\cup\{\bar{a}\}\cup\{\bar{b}\}$. Define $d(u,v)$ for $u, v \in G(A,B)$ as the cost from $G(A,B)$ and the $d(u,v)$ as the $l_{\infty}$    
%Let us put a weight $w(u,v)$ to every arc $(u,v)$ in $G(A,B)$ which equals the cost $c(u,v)$ and $d_{SP}(u,v)$ be the shortest path distance between $u$ and $v$ with these weights in $\WS_s^{PD}(A,B)$, a subgraph of $G(A,B)$}.
Let the cost $c_{\delta}(u,v)$ in $\WS_s^{PD}(A^{\delta},B^{\delta})$ be the $l_{2}$-distance between corresponding points of $u$ and $v$ for $u,v \in \hat{A}^{\delta} \sqcup \hat{B}^{\delta}$ (all non-diagonal points pairs). %, including overlaps of $\hat{A}$ with $\hat{B}$). 
Furthermore, let $c_{\delta}(\bar{b},v), v \in \hat{B}^{\delta}$ and $c_{\delta}(u,\bar{a}), u \in \hat{A}^{\delta}$ have cost exactly as in $G(A^{\delta},B^{\delta})$. Recall that in $\WS_s^{PD}(A^{\delta},B^{\delta})$ there is no arc between $\hat{A}^{\delta}$ and $\bar{b}$ nor between $\hat{B}^{\delta}$ and $\bar{a}$. Treating the costs on the arcs as weights, let the shortest path distance between $u$ and $v$ be $d_{SP}^{\WS}(u,v)$ on $\WS_s^{PD}(A^{\delta}, B^{\delta})$.
%be determined by the aforementioned costs restricted to $\WS_s^{PD}(A,B)$. 
We already have a $(1+\eps)$-spanner $\WS_s(\hat{A}^{\delta} \cup \hat{B}^{\delta})$, and adding the nodes $\bar{a}$ and $\bar{b}$ with the diagonal arcs to form $\WS_s^{PD}(A^{\delta},B^{\delta})$ still preserves the spanner property, namely 

$$d_{SP}(\bar{b},v) = c_{\delta}(\bar{b},v)\leq (1+\eps) c_{\delta}(\bar{b},v) \text{ for } v \in B^{\delta}$$ and 

$$d_{SP}(u,\bar{a}) = c_{\delta}(u,\bar{a})\leq (1+\eps) c_{\delta}(u,\bar{a}) \text{ for } u \in A^{\delta}.$$ Let $f$ and $\bar{f}$ denote the respective flows for $f^*$ and $\bar{f}^*$. We can now prove the conclusion of the theorem. \\ % the proof of Theorem \ref{thm: cabello-spanner} will still hold for $\WS_s^{PD}(A,B)$ replacing $\WS_s(\hat{A}\cup \hat{B})$. %The desired inequalities thus follow.
%\simon{note: I cannot just define P and Q and use Theorem \ref{thm: cabello-spanner} since the diagonal points are not actually "points in the plane"}
%We can thus prove 

$f^* \leq \bar{f}^*$:
$\bar{f}$ can be decomposed into flows along paths from nodes in $A^{\delta}$ to nodes in $B^{\delta}$. 
%These paths are shortest ones in the uncapacitated case.
One can get a flow $\hat{f}$ on $G(A^{\delta}, B^{\delta})$ from $\bar f$ by considering a flow
on every bipartite arc $(u,v)$ in $G(A^{\delta},B^{\delta})$ which equals the path decomposition flow from $u$ to $v$ in $\WS^{PD}_s(A^{\delta},B^{\delta})$. We have 
$$f^*=\sum_{(u,v)\in G(A^{\delta},B^{\delta})} c_{\delta}(u,v) \cdot f_{uv} \leq \sum_{(u,v) \in G(A^{\delta},B^{\delta})} c_{\delta}(u,v) \cdot \hat{f}_{uv} \leq \sum_{(u,v) \in A^{\delta} \times B^{\delta}} d_{P}^{\WS}(u,v) \cdot \hat{f}_{uv}  =\bar{f}^*,$$
where $d_{P}^{\WS}(u,v)$ is the path distance on $\WS_s^{PD}(A^{\delta},B^{\delta})$ as determined by the flow decomposition. %\tamal{I dont understand why we cannot have a decomposition only along
%shortest paths...I think it is not only possible but is required (for uncapacited case) and that is why I put shortest paths before.} \simon{answer: see http://www.columbia.edu/~cs2035/courses/ieor6614.S16/flow-decomp.pdf flow decomposition, we cannot choose the path decomposition, we only know that there exists some paths that decompose the flow}
%\tamal{Again, I dont believe you are correct for uncapaciated case...dont just read and blindly follow it. In general, I understand that it cannot be the case if capacities are allowed, but for uncapaciated case it has to be shortest paths...to me the logic is clear.}
%\simon{see: https://apps.dtic.mil/dtic/tr/fulltext/u2/a594171.pdf theorem 2.1 page 33. The proof does not depend on capacities at all.}
The leftmost inequality follows since $\hat{f}$ is a feasible flow on $G(A^{\delta},B^{\delta})$ and the rightmost inequality follows since any path length between two nodes $u$ and $v$ is bounded from below by the direct distance $c_{\delta}(u,v)$ between the points they represent. The last equality follows by the flow decomposition.\\%\tamal{What are those inequality 1, 2, last etc.? You must have taken these from previous version without correcting the numbering.}

$\bar{f}^* \leq (1+\eps)f^*$: 
$$\bar{f}^* \leq \sum_{(u,v)\in A^{\delta} \times B^{\delta}} d_{SP}^{\WS}(u,v) \cdot f_{uv} \leq \sum_{(u,v) \in G(A^{\delta},B^{\delta})}(1+\eps)c_{\delta}(u,v)\cdot f_{uv} = (1+\eps)f^*.$$
The leftmost inequality follows since the flow $f$ of $G(A^{\delta},B^{\delta})$ sent across shortest paths forms a feasible flow on $\WS_s^{PD}(A^{\delta},B^{\delta})$. To check this, notice that the supplies are all satisfied for every node in $\WS_s^{PD}(A^{\delta},B^{\delta})$. Any intermediate node of a shortest path between $u \in A^{\delta}$ and $v \in B^{\delta}$ gets a net change of 0 supply. The rightmost inequality follows because $\WS_s^{PD}(A^{\delta},B^{\delta})$ still satisfies the $(1+\eps)$-spanner property as mentioned above.
\end{proof}

\textbf{$s$-WSPD Construction:}
In order to construct an $s$-WSPD, a hierarchical decomposition such as a split tree or quad tree is constructed. We build a split tree due to its simplicity and high efficiency. A split tree can be computed sequentially with any of the standard algorithms in \cite{callahan1995decomposition,chan2008well, har2011geometric} that runs in $O(n \log n)$
time. It is not a bottleneck in practice. This is because there is only $O(n)$ writing to memory for constructing the tree. A simple construction of the split tree $T$ starts with a bounding box containing the input point set followed by a recursive division that splits a box into two halves by dividing the longest edge of the box in the middle. 
The split tree construction for a given box stops its recursion when it has one point.	
	
   Sequential construction of a WSPD involves collecting all well separated pairs of nodes which represent point subsets from the split tree $T$. This is done by searching for descendant node pairs from each interior node $w$ in $T$. For each pair of descendant nodes $u$ and $v$ reached from $w$, the procedure recursively continues the search on both children of the node amongst $u$ and $v$ that has the larger diameter for its bounding box. When the points corresponding to a pair of nodes $u,v$ become well separated, we collect $(u,v)$ in the WSPD and stop recursion. 
   
   The construction of WSPD is the primary bottleneck before the min-cost flow computation. %Although in sequential mode it is a linear time $O(s^2 \cdot n$) \cite{chan2008well} algorithm, its constant factor and high data movement is formidable in practice. 
   The sequential computation incurs high data movement and also a large hidden constant
   factor in the complexity.
   To overcome these difficulties, we compute the WSPD in parallel while still preserving locality of reference, only using $n-1$ threads, and $O(n)$ auxiliary memory. We propose a simple approach on multicore that avoids linked lists or arbitrary pointers as in \cite{callahan1993optimal, callahan1995decomposition}. A unique thread is assigned to each internal node $w$ in the split tree $T$. Then, we write a prefix sum \cite{ladner1980parallel} of the counts of well separated pairs found by each thread. Following this, each thread on $w \in T$ re-searches for well separated pairs and independently writes out its well separated descendant nodes in its memory range as determined by the prefix sum. Recursive calls on split tree node pairs can also be run in parallel as in~\cite{wang2021fast}; doing so requires an unbounded data structure to store the pairs found by each thread such as a 2-layer tree with blocks at its leaves. Such a parallel algorithm can have worst-case depth of $O(\text{polylog}(s^2 n))$ and work complexity of $O(s^2 n)$. In practice, we can gain speedup in our simplified implementation, which does not issue recursive calls at interior nodes and thus has $O(s^2 n)$ depth. This is because significant work can arise at internal nodes near the leaves. For an illustration of the implementation, see Appendix.

\subsection{Min-cost Flow by Network Simplex}
\label{sec: network simplex}
Having constructed a sparsified transshipment network, we solve the min-cost flow problem on this network with an efficient implementation of the network simplex ({\sc NtSmplx}) algorithm.  

The {\sc NtSmplx} algorithm is a graph theoretic version of the simplex algorithm used for linear programming. It involves the search for basic feasible min-cost flow solutions. This is done by successively applying \emph{pivoting} operations to improve the objective function. A pivot involves an interchange of arcs for a spanning tree on the transshipment network. %for a feasible basis. At every iteration, until convergence, it maintains a potential value $p(u) \in \mathbb{R}$ for every node $u \in WS_s^{PD}(A,B)$ and a spanning tree of arcs $(u,v)$ with \emph{reduced cost}, $\tilde{c}(u,v)= c(u,v)-p(u)-p(v)=0$. An iteration, or pivot, consists of first performing a pivot search to determine an incoming arc. This forms a unique undirected cycle $C$. The flow $f$ is then updated by $f$ + $\delta_C$, where $\delta_C$ is the minimum of the absolute value of all flows in $C$. We then find a leaving arc to maintain the spanning tree. All data structures are then updated.
%There are three primary issues associated with the network simplex algorithm. The first is that it has no guarantee to terminate. The second is that it has super quadratic complexity for arbitrary pivoting strategies. The third is the potential existence of degenerate pivots, or pivots that do not make progress in the objective function. It is known that up to 90\% of pivots are degenerate~\cite{bradley1977design}.
%We utilize two optimizations on the network simplex algorithm. First, we maintain a strongly feasible spanning tree at every iteration by selecting the exiting arc by the strongly feasible basis technique~\cite{cunningham1976network}, which guarantees that no basic feasible tree is revisited, and thus that the network simplex will terminate. Second, we use the block search pivot strategy of~\cite{grigoriadis1986efficient} for its good performance in practice and perform extensive experimental analysis of its behavior, identifying two very different computational natures. For its worst case performance, namely when it \emph{stalls} or repeatedly performs degenerate pivots for exponentially many iterations, we introduce an early stopping heuristic based on our experiments. 
% We introduce an early stopping heuristic based on our experiments. 
As observed in~\cite{kiraly2012efficient}, we also find that the pivot searching phase for the incoming arc during pivoting dominates the runtime of {\sc NtSmplx}. In particular, it is vital to have an efficient pivot searching algorithm: to quickly find a high quality entering arc that lessens the number of subsequent pivots. %There are two diametrically different simple pivot search rules. One is Dantzig's greedy pivot rule where the negative arc of lowest reduced cost value is returned. The entering arc found by this algorithm is guaranteed to eliminate degenerate pivots. The other is the first eligible pivot rule where we return the first possible negative reduce cost arc amongst all arcs in a sequential scan. This results in low quality pivots but low search times.
 Authors in~\cite{grigoriadis1986efficient} propose an interpolation between Dantzig's greedy pivot rule~\cite{dantzig2006linear} and Bland's pivot rule~\cite{bland1977new} by the block search pivot ({\sc BSP}) algorithm. This implementation for {\sc NtSmplx} is adopted in~\cite{dezsHo2011lemon}. %The list of arcs is partitioned into blocks of size $\sqrt{m}$, a good size in practice. Define the reduced cost on an arc as a cost determined by the current spanning tree on the transshipment network. Scanning over each block, we find the minimum negative reduced cost arc and return it if it exists, otherwise we move to the next block until all blocks are searched or a block is found to have a negative reduced cost arc. 
 It is found empirically in~\cite{kiraly2012efficient} that the {\sc BSP} algorithm is very efficient, simple, and results in a low number of degenerate pivots in practice. %We confirm empirically that {\sc BSP} algorithm is most efficient amongst many other pivot searching algorithms available in lemon \cite{dezsHo2011lemon}. 
 We use the {\sc BSP} algorithm in our implementation because of these reasons.
 
 Notice that if dynamic trees~\cite{tarjan1997dynamic} are used, the complexity of a pivot search can be brought down to $O(\log n)$ and thus {\sc NtSmplx} can run in time $\tilde{O}(s^2n^2)$~\cite{aggarwal1996faster, goldberg1988efficiency} on our WSPD spanner.
 %We use it since it has very low computation time, exhibits fast convergence at the start of computation in practice. 
%For approximation purposes, it is more important to converge near to the solution quickly than to assure properties for rare pessimistic cases.  

%The block size of the block search pivoting strategy is an important parameter and is tuned to $\sqrt{m}$, where $m$ is the number of arcs, by~\cite{kiraly2012efficient}, which we also agree to be a good value. When lowering the block size, each individual pivot goes by quickly but the total number of pivots also grows. Similarly, when increasing the block size, the opposite occurs. We set the block size to $\sqrt{m}$ for all experiments shown. 

{\sc BSP} sacrifices theoretical guarantees for simplicity and efficiency in practice. During computation, degenerate pivots, or pivots that do not make progress in the objective function
may appear. %It is observed that up to 90\% of pivots are degenerate~\cite{bradley1977design}. Even worse, 
There is the possibility of \emph{stalling} or repeatedly performing degenerate pivots for exponentially many iterations. As Section \ref{sec: N.S. behavior} in Appendix illustrates, stalling drives the execution to a point where no progress is made. However, our experiments suggest that, before stalling, {\sc BSP} usually arrives at a very reasonable feasible solution.

We observe that performance of {\sc NtSmplx} depends heavily on the sparsity of our network. Since a pivot involves forming a cycle with an entering arc and a spanning tree in the network, if the graph is sparse there are few possibilities for this entering arc. 
%Notice that when $\WS_s(A^{\delta},B^{\delta})$, the network on the nondiagonal points, forms a tree when , it has $O(n)$ possible entering arcs for each pivot by construction. 
%The {\sc BSP} algorithm is known to empirically lower the $O(s^2n)$ number of entering arcs that needed to be searched for each pivot.

%\simon{Theoretically speaking, \sout{it is known that} {\sc NtSmplx} on our WSPD spanner can run in time $\tilde{O}(s^2n^2)$~\cite{tarjan1997dynamic}.}

\subsection{Approximation Algorithm}
The approximation algorithm is given in Algorithm \ref{alg: wasserstein-PDoptFlow}, which proceeds as follows. Given input PDs $A$ and $B$ and the parameter $s> 2$, first we set $\eps= \frac{8}{s-4}$. We compute a $\delta$ according to Proposition \ref{prop: delta-condensation} using this $\eps$ for $s\geq 12$ and setting $\eps=1$ for $2<s<12$. Then, we perform a $\delta$-condensation and compute an $s$-WSPD 
via a split tree construction on $\hat{A}^{\delta} \cup \hat{B}^{\delta}$. 

We then compute $\WS_s(\hat{A}^{\delta}\cup\hat{B}^{\delta})$ from the $s$-WSPD. It is a $(1+\eps')$-spanner for $s>2$ where $\eps'=\frac{4}{s}+\frac{4}{s-2}$. %Since $\eps' \leq \eps$, $\WS_s(\hat{A}^{\delta}\cup\hat{B}^{\delta})$ is also a $(1+\eps)$-spanner. 
Diagonal nodes along with their arcs are added to this graph
as determined by $G_{\delta}$. This means that we add the nodes $\bar{a}$ and $\bar{b}$ and all arcs from $\hat{A}^{\delta}$ to $\bar{b}$ and $\bar{a}$ to $\hat{B}^{\delta}$. This produces $\WS_s^{PD}(A^{\delta},B^{\delta})$. Figure \ref{fig: construction-1-Wasserstein} shows our construction. The network simplex algorithm is applied to the sparse network $\WS_s^{PD}(A^{\delta},B^{\delta})$ to get a distance that approximates the min-cost flow value on $G_{\delta}$ within a factor of $(1+\eps')$ between inputs $A^{\delta}$ and $B^{\delta}$.  The algorithm still runs for $s>0$ instead of $s>2$ since we can still construct a valid transshipment network for optimization. However, there are no guarantees if $s\leq2$. Nonetheless, empirical error is found to be low and the computation turns out very efficient; see Section \ref{sec: experiments}.

\begin{algorithm}[H]

\begin{algorithmic}[1]
\Require{PDs: $A,B$, $s> 2$ the sparsity parameter, $\eps = \frac{8}{s-4}$ for $s\geq 12$ and $\eps= 1+\frac{8}{s}+\frac{8}{s-2}$ for $2<s<12$}
\Ensure{a $(1+O(\eps))$-approximation to $W_1$-distance}
    
    \State $(\pmb{P}$,$\sigma_{\pmb{P}}) \gets$ $\delta$-condensation$(A,B,s)$ \Comment{$\pmb{P}= \hat{A}^{\delta}\cup \hat{B}^{\delta}$} 
    %results in a $1\pm \eps$ approx. factor;
			
    \State $\pmb{T}\gets$ form-splittree($\pmb{P}$)
			
	\State nondiag-arcs $\gets$ form-WSPD$(\pmb{T},s)$ \Comment{$1+\eps$-spanner}
			
	\State diag-arcs $\gets$ form-diag-arcs$(\pmb{P})$ \Comment{diagonal arcs constructed as in Section \ref{sec: reduction}}
			
	\State $\pmb{G} \gets (\pmb{P}$, nondiag-arcs $\cup$ diag-arcs,~$\sigma_{\pmb{P}}$, $c:=$dists(nondiag-arcs $\cup$ diag-arcs))
	\Comment{
	%$\pmb{G}=(V,E,\sigma_{\pmb{P}},c)$ is 
	Defn. \ref{eq: mincostflow}}
			
	\Return min-cost flow$(\pmb{G})$%+inftydist$
	
\end{algorithmic}
\caption{Approximate $W_1$-Distance Algorithm}
\label{alg: wasserstein-PDoptFlow}
\end{algorithm}

\iffalse
\begin{algorithm}[H]
\label{alg: wasserstein-PDoptFlow}
\SetAlgoLined
\KwIn{PDs $A,B$ $s\geq 12$, the WSPD parameter with $\eps = \frac{8}{s-4}$}
\KwOut{a $(1+O(\eps))$ approximate $W_1$-distance}
    
    $(\pmb{P}$,$\sigma_{\pmb{P}}) \gets$ $\delta$-condensation$(A,B,s)$ \Comment{$\pmb{P}= \hat{A}^{\delta}\cup \hat{B}^{\delta}$}; results in a $1\pm \eps$ approx. factor
			
    $\pmb{T}\gets$ form-splittree($\pmb{P}$)
			
	nondiag-arcs $\gets$ form-WSPD$_{parallel}(\pmb{T},s)$ \Comment{gives a $1+\eps$-spanner if $s>4$}
			
	diag-arcs $\gets$ form-diag-arcs$_{parallel}(\pmb{P})$  \Comment{diagonal arcs constructed as in Section \ref{sec: reduction}}
			
	$\pmb{G} \gets (\pmb{P}$, nondiag-arcs $\cup$ diag-arcs,~$\sigma_{\pmb{P}}$, $c:=$dists(nondiag-arcs $\cup$ diag-arcs))
			
	\Return network-simplex$(\pmb{G})$%+inftydist$
	\Comment{$\pmb{G}=(V,E,\sigma_{\pmb{P}},c)$ is transship. network as in Defn. \ref{eq: mincostflow}}

\caption{Approximate 1-Wasserstein Distance Algorithm}
\end{algorithm}
\fi

	The time complexity of the algorithm is dominated by the computation of the min-cost flow routine. Thus, all the steps of our algorithm are designed to improve the efficiency of the {\sc NtSmplx} algorithm. Replacing {\sc NtSmplx} with the algorithm in \cite{brand2021minimum}, a complexity of $\tilde{O}(n s^2+n^{1.5})$ 
	%\cite{bernstein2021deterministic} a complexity of $\tilde{O}((ns^2)^{1+o(1)})$ for a $1+\eps$ min-cost flow on a sparse graph 
	can be achieved. %, assuming a WSPD is used to form a linear number of arcs to approximate a full bipartite graph and a sparse $O(n^2\cdot s^2)$ incidence matrix is constructed in parallel. %The min-cost flow algorithm used in Algorithm \ref{alg: wasserstein-PDoptFlow} is not limited to network simplex.
	However, {\sc NtSmplx} is simpler, more memory efficient, has a reasonable complexity of $\tilde{O}(s^2n^2)$ \cite{tarjan1997dynamic}, and is very efficient in practice; see Figure \ref{fig: nXtime} and Figure \ref{fig: arcsXtime}  in Appendix.
	%\simon{Furthermore, it is known that {\sc NtSmplx} on our WSPD spanner can run in time $\tilde{O}(s^2n^2)$~\cite{tarjan1997dynamic}}.%For small $s$ {\sc NtSmplx} acts faster than $O(s^2n^2)$%As in Figure \ref{fig: arcsXtime}, Algorithm \ref{alg: wasserstein-PDoptFlow} can actually achieve a similar complexity experimentally.%In our implementation, however, we use the practically network simplex algorithm from the Lemon library \cite{dezsHo2011lemon} for min-cost flow.

 \section{Theoretical Bounds}
 \label{sec: bounds}
 By Theorem \ref{theorem: Wasserstein-approx-thm}, the spanner achieves a $(1+\frac{4}{s}+\frac{4}{s-2})$-approximation to the min-cost flow value on the $\delta$-condensed graph. A $\delta$-condensation results in an approximation of the
$W_1$-distance with a factor of $(1\pm (\frac{8}{s-4}))$ for $s\geq 12$ and $2$ for $2<s<12$. The factor $2$ for the range $2<s<12$ is obtained by putting $s=12$ in $\frac{8}{s-4}$ because $s\leq 12$ and we need $\frac{8}{s-4}>0$. The node and arc sparsifications together guarantee an approximation factor of  $((1+\frac{4}{s}+\frac{4}{s-2})(1\pm (\frac{8}{s-4}))) \leq 
(1+\eps)^2$
%1+q(\eps)$, where $q$ is a quadratic polynomial of $\eps$ 
where $\eps= \frac{8}{s-4}$ and $s\geq 12$. For the range $2<s<12$, we have $2 (1+\frac{4}{s}+\frac{4}{s-2}) = 1+\eps$ where $\eps= 1+\frac{8}{s}+\frac{8}{s-2}$. We are thus guaranteed a $(1+O(\eps))$-approximation to the $W_1$-distance if $s>2$ as claimed
in Algorithm~\ref{alg: wasserstein-PDoptFlow}.
We state this as the following Corollary to Theorem \ref{theorem: Wasserstein-approx-thm}.
\subsection{Main Result}

\begin{corollary}
\label{corollary: approximationbound}
Let $\eps>0$ and define $s= 4+\frac{8}{\eps}$ for $s\geq 12$. %or $s$ \tamal{what is this `or $s$'? what does it mean?} such that $\eps= 1+\frac{8}{s}+\frac{8}{s-2}$ for $2<s<12$. 
Define $\delta$ in terms of $\eps$ as in Proposition \ref{prop: delta-condensation}.
%be the solution to $(1+\frac{8}{s-4})\cdot(1\pm (\frac{8}{s-4}))\leq (1+\eps)^2$ and 
Then, $\bar f^*$, the min-cost flow value of $\WS_s^{PD}(A^{\delta},B^{\delta})$, is a $(1+O(\eps))$-approximation of $W_1(A,B)$.
\end{corollary}
This allows us to now prove the main theorem upon which our approach is centered. 
\begin{theorem}\label{thm: W1-sparse-complexity-insec}
    (Main Theorem for the Complexity of Computing the $W_1$-distance)
    
	Let $\varepsilon>0$ and $A=\tilde{A}\cup \Delta,B=\tilde{B}\cup \Delta$ two PDs of atmost $n$ points, 
    
The $W_1$-distance can be reduced to computing a min-cost flow on a sparse network. This can theoretically be computed in time $O(\frac{1}{\epsilon^2}n\log(n))$.
\end{theorem}
\textbf{Proof of Main Theorem: }
\begin{proof}
    According to Corollary \ref{corollary: approximationbound}, the sparse transhipment network $\WS_s^{PD}(A^{\delta},B^{\delta})$ has a min-cost flow that is a $(1+\varepsilon)$-approximation of $W_1$. We know that computing the min-cost flow can be computed in near linear time~\cite{brand2023deterministicalmostlineartimealgorithm} and constructing a hierarchical decomposition tree such as a kd-tree, quadtree or split tree takes $O(n\log(n))$ time. Constructing the WSPD geometric spanner takes time complexity of $O(\frac{1}{\varepsilon^2}n\log(n))$ Composing then gives the complexity as stated in the Theorem. 
\end{proof}
\subsection{Conditional Lower Bound for $W_1(A,B)$}
The $W_1$ distance between persistence diagrams can be viewed as a variation of the Earth mover's distance (EMD) problem from computational geometry.
We state the EMD problem here:
\begin{problem}(EMD)
    Let $V_1,V_2 \subseteq \mathbb{R}^d$ be two point sets of $d$-dimensional Euclidean space. The EMD problem seeks for the minimum value of the following optimization problem:
    \begin{equation}
        EMD(V_1,V_2) \triangleq \min_{\sigma: V_1 \rightarrow V_2}\sum_{u \in V_1} \|u-\sigma(u)\|_2
    \end{equation}
    where $\sigma: V_1\rightarrow V_2$ is a matching (injective map) between $V_1$ and $V_2$
\end{problem}
\begin{conjecture}(Constant Dimension EMD Conjecture)

For a constant $d\geq 2$,
    there is no $\delta>0$ where there is a deterministic algorithm that given
two lists of $n$ points from $\mathbb{R}^d$ can compute in $O(n^{2-\delta})$ time the EMD between these two lists.
\end{conjecture}
This conjecture appears reasonable since a perfect matching over a bipartite graph with an arbitrary cost matrix takes time $\Omega(n^2)$, the size of the input. %even a simple search for the nearest conflict-free neighbor of each point would require $\Omega(n^2)$ time assuming a $O(\log(n)+k)$ optimal lower bound complexity \cite{10.5555/1496770.1496791} for searching for $k$-nearest neighbors in the plane.

Assuming the constant dimension EMD conjecture, we show through the technique of fine-grained reduction~\cite{vassilevskawilliams:LIPIcs.IPEC.2015.17} that the exact $W_1$ distance between persistence diagrams also cannot be subquadratic unless EMD can be solved in subquadratic time.
\begin{theorem}\label{thm: EMDfgr2W1}
    Let $n>0$ be an integer and let $\epsilon>0$

If the exact EMD on $\mathbb{R}^2$ and two point sets of size $n$ cannot be computed in time $O(n^{2+o(1)-\delta})$ for any $\delta>0$, then the computation of $W_1$ between two persistence diagrams of total size $n$ cannot be computed in time $O(n^{2+o(1)-\delta'})$ for some $\delta'>0$
\end{theorem}
\begin{proof}
We do a $(n^{2},n^{2})$-fine grained reduction from the exact EMD for $d=2$ to the $W_1$ problem between persistence diagrams.
    
    \textbf{The Reduction: }
    
    Given an input $A,B \subseteq \mathbb{R}^2$, 
    
    1. Compute $\text{diam}(A\cup B)= \max_{x,y \in A \cup B} \|x-y
    \|_2$. This takes time $O(n)$ time.

    2. Compute the displacement vectors 
    \begin{equation}
        \mathcal{D}= \{(\|p-p_{proj}(p)\|_2, (p-p_{proj}(p))_x)\}_{p \in A \cup B}
    \end{equation}
    consisting of (magnitude, direction) pairs where $p_{proj}: \mathbb{R}^2 \rightarrow \Delta$ is the projection map to the diagonal $\Delta= \{(x,x): x \in \mathbb{R}\}$ as given in Section \ref{sec: reduction} and $(p)_x$ is the $x$-coordinate of point $p$. This takes time $O(n)$.
    
    3. Amongst all the vectors $(m,r) \in \mathcal{D}$ with $r=(p-p_{proj}(p))_x \leq 0$ find the vector with the largest magnitude $m$, call this maximizer $m^*$. This takes time $O(n)$.
    
    4. Translate all points of point sets $A$ and $B$ by 
    \begin{equation}
        dd\triangleq (-(m^*+n(\epsilon+1) \text{diam}(A\cup B)),+(m^*+n(\epsilon+1) \text{diam}(A\cup B)))
    \end{equation}
    Call these translated point sets $A_t, B_t$. This takes time $O(n)$.

 \textbf{The Reduction Maintains Correctness: }

    \textbf{An optimal EMD matching iff an optimal $W_1$ matching}
    
    We claim that $EMD(A,B)=W_1(A_t,B_t)$, namely that the $EMD(A,B)$ and $W_1$ distances don't change under translation by $dd$. In fact, the witnesses to both problems are exactly equal:  

    Let $\sigma^*: A\rightarrow B$ be the optimal EMD matching and let $\sigma_t^*: A_t \rightarrow B_t$ be the partial matching witnessing $W_1(A_t,B_t)$. We claim that $\sigma^*=\sigma_t^*$. 
    
    This follows since if we introduce any matching $(p,p_{proj}(p)), p \in A_t \cup B_t$ to the diagonal into ${\sigma}_t^*$ by replacing a match $(p_t,q_t) \in A_t \times B_t$ by the two matches $(p_t,p_{proj}(p_t)), (q_t,p_{proj}(q_t))$ the new $W_1$ cost results in the following inequality:
    \begin{subequations}
    \begin{equation}
        \sum_{u_t \in A_t} \|u_t-{\sigma}^*_t(u_t)\|_2-\|p_t-q_t\|_2+\|p_t-p_{proj}(p_t)\|_2+\|q_t-p_{proj}(q_t)\|_2 
    \end{equation}
    \begin{equation}
        \geq \sum_{u_t \in A_t} \|u_t-{\sigma}^*_t(u_t)\|_2= W_1(A_t,B_t)
    \end{equation}
    \end{subequations}
    This follows since
    \begin{subequations}
    \begin{equation}
        \|p_t-p_{proj}(p_t)\|_2+\|q_t-p_{proj}(q_t)\|_2 \geq 2(n(\epsilon+1)\text{diam}(A\cup B)) 
        \end{equation}
    \begin{equation}
        \geq \text{diam}(A\cup B) \geq \|p_t-q_t\|_2,  \forall p_t,q_t \in A_t \times B_t
    \end{equation}
    \end{subequations}
    Thus $\sigma_t^*$ cannot involve any matchings to the diagonal and thus $W_1(A_t,B_t)$ reduces to the $EMD(A_t,B_t)$. Thus $\sigma^*=\sigma^*_t$.
\end{proof}
\subsection{Conditional Lower Bound for the $(1+\epsilon)$ case}
For an approximate EMD problem, namely a problem where the desired solution is near the original EMD, we can define the following.
A $(1+\epsilon)$-approximate EMD solution is defined by a transshipment network $G=(V_1\cup V_2, V_1\cup V_2 \times V_1\cup V_2, f,c,\mu)$ with uncapacitated flow function ${f}: V_1\cup V_2 \times V_1\cup V_2 \rightarrow \mathbb{R}^+$ as defined in Equation \ref{eq: flow}, for some cost function $c: V_1\cup V_2 \times V_1\cup V_2 \rightarrow \mathbb{R}^+$, and some supply function $\mu: V_1 \cup V_2 \rightarrow \mathbb{R}$ with $\mu(v)=1, \mu(w)=-1, \forall v,w \in V_1\times V_2$ where:
\begin{equation}\label{eq: approx-EMD}
    EMD(V_1,V_2) \leq \min_{f: V_1 \times V_2 \rightarrow \mathbb{R}^+, \text{ a flow }}\sum_{(u,v) \in V_1 \times V_2} c(u,v) f(u,v) \leq (1+\epsilon) EMD(V_1,V_2)
\end{equation}
This allows us to define the $(1+\epsilon)$-approximate EMD problem:
\begin{problem}
    The $(1+\epsilon)$-approximate EMD problem computes the value 
    
     $\min_{f: V_1 \times V_2 \rightarrow \mathbb{R}^+, \text{ a flow }}\sum_{(u,v) \in V_1 \times V_2} c(u,v) f(u,v)$ from Equation \ref{eq: approx-EMD}.
\end{problem} 
\iffalse
This gives a subsequent reduction to the $W_1$ distance problem. Thus, assuming the hypothesis that the near linear time algorithm of \cite{brand2023deterministicalmostlineartimealgorithm} achieves the optimal lower bound complexity for solving the min-cost flow problem, 
\begin{conjecture}\cite{abboud2016approximation}(Hitting Sets Conjecture)
    There is no $\epsilon>0$ such that for all $c \geq 1$, there is an algorithm that given
two lists of $n$ subsets of a universe $U$ of size $\omega( \log n)$, can decide in $O(n^{2-\epsilon})$ time if there is a set in the first
list that intersects every set in the second list, i.e. a “hitting set”.
\end{conjecture}
\fi
Within the fine-grained complexity framework \cite{williams2018some} we show that the $W_1$ distance between PDs and the EMD problem in the plane are reducible to each other in both the exact and $(1+\epsilon)$-approximate cases.  

It is known that a $(1+\epsilon)$-approximate EMD in $d$ dimensions can be computed by a randomized algorithm in time $O(n \text{poly}(\frac{1}{\epsilon},\log(n)))$ \cite{raghvendra2020near} as well as in deterministic time $\tilde{O}(n (\frac{1}{\epsilon}
\log(n))^{O(d)})$ \cite{agarwal2022deterministic}. Certainly the EMD can be solved in $d$ dimensions through a spanner followed by the near linear time min-cost flow algorithm of \cite{brand2023deterministicalmostlineartimealgorithm}, making a $(1+\epsilon)$-approximate EMD computable in $O(n^{1+o(1)}\frac{\log(n)}{\epsilon^2})$ time according to \cite{cabello2005matching}. This is, in fact, faster than the $O(n\frac{\log^2(n)}{\epsilon^2})$ time algorithm of \cite{agarwal2022deterministic} due to $n^{o(1)}=O(\log(n))$. We hypothesize that for any $\delta>0$ and any $\epsilon>0$, a $(1+\epsilon)$-approximate EMD in constant $d$ dimensions cannot be solved in time $O(n^{1+o(1)-\delta})$.

This is stated in the following conjecture:
\begin{conjecture}(Constant Dimension $(1+\epsilon)$-approximate EMD Conjecture)

For a constant $d\geq 2$,
    there is no $\delta>0$ such that for all $\epsilon>0$, there is a deterministic algorithm that given
two lists of $n$ points from $\mathbb{R}^d$ can compute in $\tilde{O}(\frac{1}{\epsilon^d}n^{1-\delta})$ time the $(1+\epsilon)$-approximate EMD.
\end{conjecture}
It is known through fine-grained reduction \cite{rohatgi2019conditional} that for any $\delta>0$ if a $(1 + \frac{1}{n^{\delta}})$-approximate EMD in dimensions $\omega(\log(n))$ of Euclidean space cannot be solved in $O(n^{2-\delta})$ time, then the Hitting Sets Conjecture \cite{abboud2016approximation} would be false. However this is separate from the finite dimensional case due to the dependency of $d$ on $n$. 

It is presumed that the smaller the dimension $d\geq 2$ that the $(1+\epsilon)$-approximate EMD problem on $\mathbb{R}^d$ would be easier to solve.
So there would be no contradiction that the constant dimension version of the problem is solvable in subquadratic time. 
We show below that assuming the hypothesis that the $(1+\epsilon)$-approximate EMD has an optimal near linear time lower bound complexity, then the $W_1$ distance between PDs has optimal lower bound complexity of near linear time. 
\begin{theorem}\label{thm: 1+epsilon-EMDlowerbound}
Let $n>0$ be an integer and let $\epsilon>0$

    If the $(1+\epsilon)$-approximate EMD on $\mathbb{R}^2$ and two point sets of size $n$ cannot be computed in time $\tilde{O}(\frac{1}{\epsilon^2}n^{1-\delta})$ for any $\delta>0$, then a $(1+\epsilon)$-approximate computation of $W_1$ between two persistence diagrams of total size $n$ cannot be computed in time $\tilde{O}(\frac{1}{\epsilon^2}n^{1-\delta'})$ for some $\delta'>0$
\end{theorem}
\begin{proof}
    We do a $(\tilde{O}(\frac{1}{\epsilon^2}n),\tilde{O}(\frac{1}{\epsilon^2}n))$-fine grained reduction from the $(1+\epsilon)$-approximate EMD for $d=2$ to the $W_1$ problem between persistence diagrams.
    
    \textbf{The Reduction (same as in the reduction of Theorem \ref{thm: EMDfgr2W1}): }
    
    Given an input $A,B \subseteq \mathbb{R}^2$, 
    
    1. Compute $\text{diam}(A\cup B)= \max_{x,y \in A \cup B} \|x-y
    \|_2$. This takes time $O(n)$ time.

    2. Compute the displacement vectors 
    \begin{equation}
        \mathcal{D}= \{(\|p-p_{proj}(p)\|_2, (p-p_{proj}(p))_x)\}_{p \in A \cup B}
    \end{equation}
    consisting of (magnitude, direction) pairs where $p_{proj}: \mathbb{R}^2 \rightarrow \Delta$ is the projection map to the diagonal $\Delta= \{(x,x): x \in \mathbb{R}\}$ as given in Section \ref{sec: reduction} and $(p)_x$ is the $x$-coordinate of point $p$. This takes time $O(n)$.
    
    3. Amongst all the vectors $(m,r) \in \mathcal{D}$ with $r=(p-p_{proj}(p))_x \leq {0}$ find the vector with the largest magnitude $m$, call this maximizer $m^*$. This takes time $O(n)$.
    
    4. Translate all points of point sets $A$ and $B$ by 
    \begin{equation}
        dd\triangleq (-(m^*+n(\epsilon+1) \text{diam}(A\cup B)),(m^*+n(\epsilon+1) \text{diam}(A\cup B)))
    \end{equation}
    Call these translated point sets $A_t, B_t$. This takes time $O(n)$.

    \textbf{The Reduction Maintains Correctness: }
    
    \textbf{$(1+\epsilon)$-approximate $EMD(A,B)$ iff $(1+\epsilon)$-approximate $W_1(A_t,B_t)$}

    Let \begin{equation}
    \begin{split}
    G_{t,proj}\triangleq (A_t \cup B_t\cup p_{proj}(A_t)\cup p_{proj}(B_t),\\ (A_t \cup B_t \cup p_{proj}(A_t)\cup p_{proj}(B_t))  \times (A_t \cup B_t \cup p_{proj}(A_t)\cup p_{proj}(B_t)), \hat{f}_t,c_t,\mu_t ) 
    \end{split}
    \end{equation} 
    and denote 
    \begin{equation}
        G\triangleq ((A\cup B),(A\cup B) \times (A\cup B), f,c,\mu)
    \end{equation} as transhipment networks for the $(1+\epsilon)$-$W_1(A_t,B_t)$ and $(1+\epsilon)$-$EMD(A,B)$ problems, respectively.

    $(\Rightarrow): $
    
    For a witness flow $f^*: G \rightarrow \mathbb{R}^+$ that pushes all $n$ units of flow from $A$ to $B$ that solves the $(1+\epsilon)$-approximate $EMD(A,B)$ problem, define $f_t:  G_{t,proj} \rightarrow \mathbb{R}^+$ as 
    \begin{equation}
        f_t((u+dd,v+dd))\triangleq f^*(u,v)
    \end{equation} 
    We claim that this flow $f_t$ obtains a $(1+\epsilon)$-approximate $W_1(A_t,B_t)$ distance. This means that it is the minimizer of the following distance:
    \begin{gather}
    \begin{split}
        \hat{W}_1(A_t,B_t)\triangleq \min_{\hat{f}_t: G_{t,proj} \rightarrow \mathbb{R}^+\text{ is a flow}}F_{c_t}(\hat{f}_t)
        \end{split}
    \end{gather}
    where: 
    \begin{equation}
        F_{c_t}(\hat{f}_t)\triangleq\sum_{(u_t,v_t) \in G_{t,proj}} c_t(u_t,v_t) \hat{f}_t(u_t,v_t) 
    \end{equation}
    satisfying:
    \begin{equation}
        W_1(A_t,B_t)\leq \hat{W}_1(A_t,B_t)\leq (1+\epsilon)W_1(A_t,B_t)
    \end{equation}
    We first notice that any flow $\hat{f}_t$ %for the $(1+\epsilon)$-approximate $W_1(A_t,B_t)$ distance between persistence diagrams 
    cannot involve any flow to the diagonal. Similar to the proof above, we have that pushing flow towards the diagonal will increase the cost.  This can be expressed as:
    \begin{subequations}
    \begin{gather}
    \begin{split}
        F_{c_t}(\hat{f}_t)- c_t(p,q)\hat{f}_t(p,q)\rho \\+ \rho(\|p-p_{proj}(p)\|_2 \hat{f}_t(p,p_{proj}(p))+ \|q-p_{proj}(q)\|_2 \hat{f}_t(p_{proj}(q),q))
        \end{split}
    \end{gather}
    \begin{equation}\label{eq: lb-flow}
       \geq F_{c_t}(\hat{f}_t)
    \end{equation}
    \begin{equation}
        \text{ s.t. } \forall \rho: 0<\rho\leq 1,  \hat{f}_t(p,p_{proj}(p))+  \hat{f}_t(p_{proj}(q),q)= \hat{f}_t(p,q)
    \end{equation}
    \end{subequations}
    Where the inequality of Equation \ref{eq: lb-flow} comes from the following inequality on the projection distances. %Since the projection distances were given by 
    \begin{subequations}
    \begin{equation}
         \|p-p_{proj}(p)\|_2 \hat{f}_t(p,p_{proj}(p))+ \|q-p_{proj}(q)\|_2 \hat{f}_t(p_{proj}(q),q) 
    \end{equation}
    \begin{equation}
        \geq ( (1+\epsilon)n \text{diam}(A_t\cup B_t)) \hat{f}_t(p,q) \geq \text{diam}(A_t\cup B_t))\hat{f}_t(p,q) \geq c_t(p,q)\hat{f}_t(p,q)
    \end{equation}
    \end{subequations}
    Let 
    \begin{equation}
    \begin{split}
    G_{t}\triangleq (A_t \cup B_t, (A_t \cup B_t )  \times (A_t \cup B_t), \hat{f}_t\mid_{G_t},c_t\mid_{(A_t \cup B_t )  \times (A_t \cup B_t)},\mu_t\mid_{A_t \cup B_t} ) 
    \end{split}
    \end{equation} 
    We know that $f_t$ does not involve flow to or from the diagonal.
    The flow ${f}_t$ also cannot be improved with flow to or from the diagonal. Furthermore, since $f_t$ is optimal on $G_t$, it must be that $f_t$ is the optimal solution for $\hat{W}_1$. 

    $(\Leftarrow): $

    We know that any $\hat{f}_t$ cannot be improved with any flow to or from the diagonal. Thus letting the flow $\hat{f}_t\mid_{G_t}=\hat{f}_t$ and optimizing $F_{c_t}(\hat{f}_t)$, we get that:
    \begin{equation}
        f_t((u+dd,v+dd))\triangleq f^*(u,v)
    \end{equation} 
    is an optimal solution.
\end{proof}
 \textbf{Approximate Nearest Neighbor Bound:}
 \label{sec: NNproblemandbound}
 Define the following problem using the solution to Problem \ref{prob: wasserstein}.
 
 \begin{problem}
 \label{prob: NN}
 Given PDs $A_1,\ldots,A_n$ and a query PD $B$,
    find the nearest neighbor (NN) $A^*= \mathrm{argmin}_{A_i\in\{A_1,\ldots,A_n\}}W_1(B,A_i).$
    \end{problem} 
    
    We obtain the following bound on the approximate NN factor of our algorithm, where a $c$-approximate nearest neighbor $A^{*}$ to query PD $B$ among $A_1...A_n$ means that $W_1(A^*,B) \leq c \cdot \min_{i} (W_1(A_{i},B))$.
    \begin{theorem}
    \label{theorem: approx-NN}
    %Let $4+\frac{8}{\eps}= s \geq 12$. The nearest neighbor of a PD determined by the distances computed by Algorithm \ref{alg: wasserstein-PDoptFlow} with sparsity parameter $s$ is a $\frac{(1+\eps)^2}{1-\eps}$-approximate nearest neighbor w.r.t. the $W_1$ distance.
    
    Let $4+\frac{8}{\eps}= s \geq 12$. The nearest neighbor of PD $B$ among PDs $A_1,...A_n$ as computed by {\sc PDoptFlow} at sparsity parameter $s$ is a $\frac{(1+\eps)^2}{1-\eps}$-approximate nearest neighbor in the $W_1$-distance.
    
    \end{theorem}
    %Proof is in Appendix.
    \begin{proof}
    %The proof is motivated by the approximation bound of \cite{backurs2020scalable}.
    
    For a given $s$, define $\eps= \frac{8}{s-4}$ and an appropriate $\delta$ as in Proposition \ref{prop: delta-condensation}. Let $A'$ be the nearest neighbor according to {\sc PDoptFlow} at sparsity parameter $s$ and $B$ be the query PD. Let $f^*_{A',B}$ be the optimal flow between $A'$ and $B$ and let $f^*_{A'^{\delta},B^{\delta}}$ be the optimal flow on the pertrubed $\delta$-grid and let $f^s_{A',B}$ be the optimal flow between them on the sparsified graph with parameter $s$. Let $X$ be the union of all PDs of interest such as $A_1...A_n$ and $B$. Let $X^{\delta}$ be the perturbed grid obtained by snapping $X$. Let {\sc PDoptFlow}$_s$ denote the value of the optimal flow computed by {\sc PDoptFlow} for sparsity parameter $s$. We have that: 
    
    $W_1(A',B) = \sum_{(x,y) \in X\times X} f^*_{A',B} \cdot \|x-y\|_2$
    
    $\leq (\frac{1}{1-\frac{8}{s-4}}) \cdot \sum_{(x',y') \in X^{\delta} \times X^{\delta}} f^*_{A'^{\delta},B^{\delta}} \cdot \|x'-y'\|_2$ (lower bound from Proposition \ref{prop: delta-condensation})
    
    $\leq (\frac{1}{1-\frac{8}{s-4}}) \cdot \sum_{(x',y') \in \WS_s(\hat{A'}^{\delta},\hat{B}^{\delta})} f^s_{A'^{\delta},B^{\delta}} \cdot \|x'-y'\|_2$ (optimality of $f^*_{A'^{\delta},B^{\delta}}$)
    
    $= (\frac{1}{1-\frac{8}{s-4}}) \cdot$ {\sc PDoptFlow}$_s({A'}^{\delta}, B^{ \delta})$
    
    %$= (\frac{1}{1-\frac{8}{s-4}}) \cdot \WS_s^{PD}({\mu'}^{\delta}, \nu^{ \delta})$ (optimality of $f^*_{\mu'^{\delta},\nu^{\delta}}$)
    
    $\leq (\frac{1}{1-\frac{8}{s-4}}) \cdot \sum_{(x',y') \in \WS_s(\hat{A^*}^{\delta},\hat{B}^{\delta})} f^s_{{A^*}^{\delta},B^{\delta}} \cdot \|x'-y'\|_2$ (optimality of $A'$ w.r.t. {\sc PDoptFlow}$_s$)
    
    = $(\frac{1}{1-\frac{8}{s-4}}) \cdot$ {\sc PDoptFlow}$_s({A^*}^{\delta},B^{\delta})$
    
    %$= (\frac{1}{1-\frac{8}{s-4}}) \cdot \WS^{PD}_s({\mu^*}^{\delta},\nu^{\delta})$ 
    
    $\leq (\frac{1}{1-\frac{8}{s-4}}) \cdot (1+\frac{8}{s-4}) \cdot (1+\frac{4}{s}+\frac{4}{s-2}) \cdot W_1(A^*,B)$ (by Corollary \ref{corollary: approximationbound})
    
    $ \leq \frac{(1+\eps)^2}{1-\eps} \cdot W_1(A^*,B)$ (if $4+\frac{8}{\eps}=s \geq 12$ and by Corollary \ref{corollary: approximationbound})
    \end{proof}

    This bound matches with our experiments described in Section \ref{sec: 1-NN} which show the high NN prediction accuracy of {\sc PDoptFlow}.

\section{Experiments} \label{sec: experiments}
\begin{figure*}[t]
\begin{minipage}{.33\textwidth}
  \centering
  \includegraphics[width=0.6\linewidth]{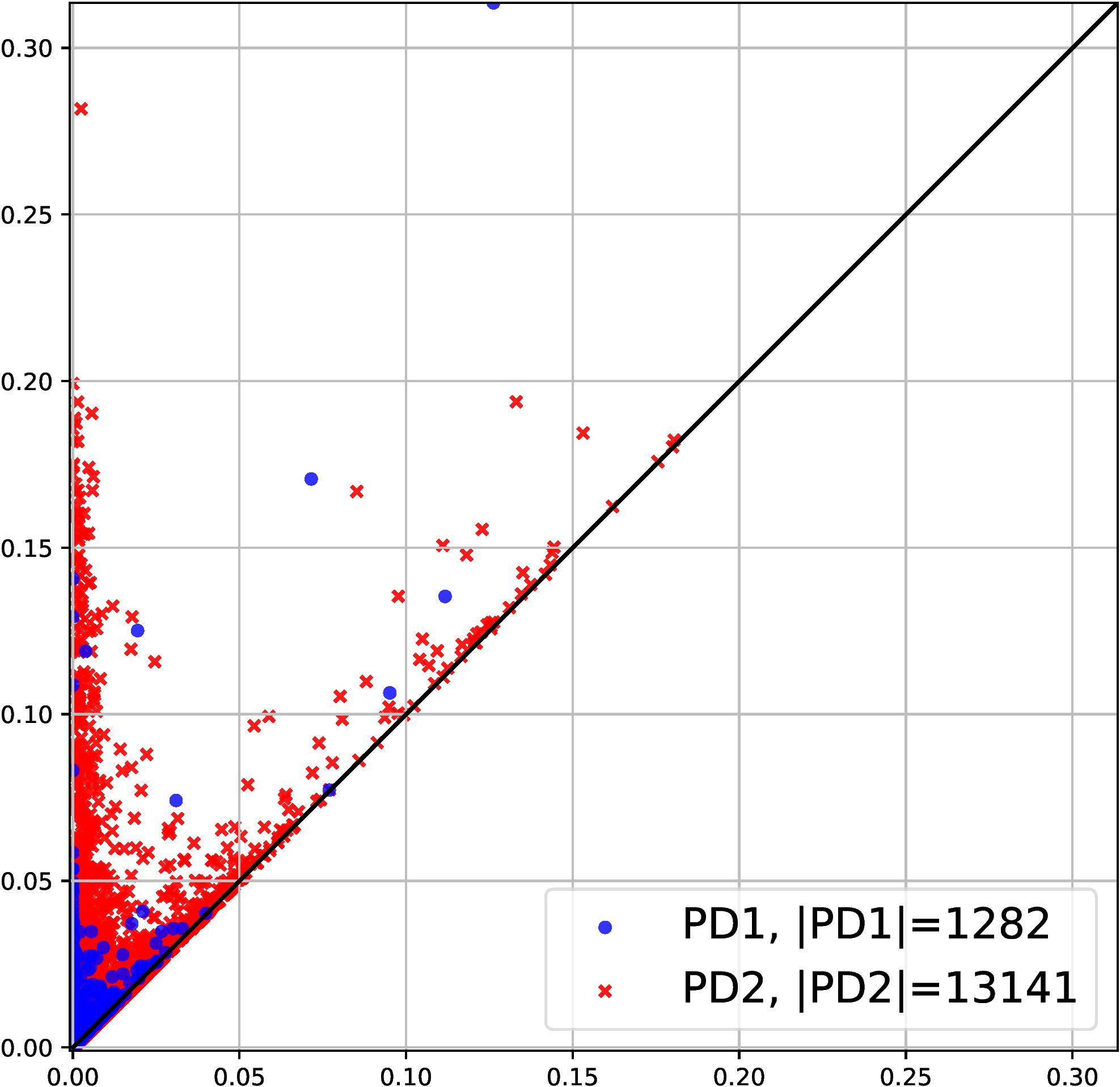}
  
  \subcaption{ PD1: {\sf Athens}, PD2: {\sf Beijing}}
  \label{fig: BeijingxAthens}
\end{minipage}%
\begin{minipage}{.33\textwidth}
  \centering
  \includegraphics[width=0.6\linewidth]{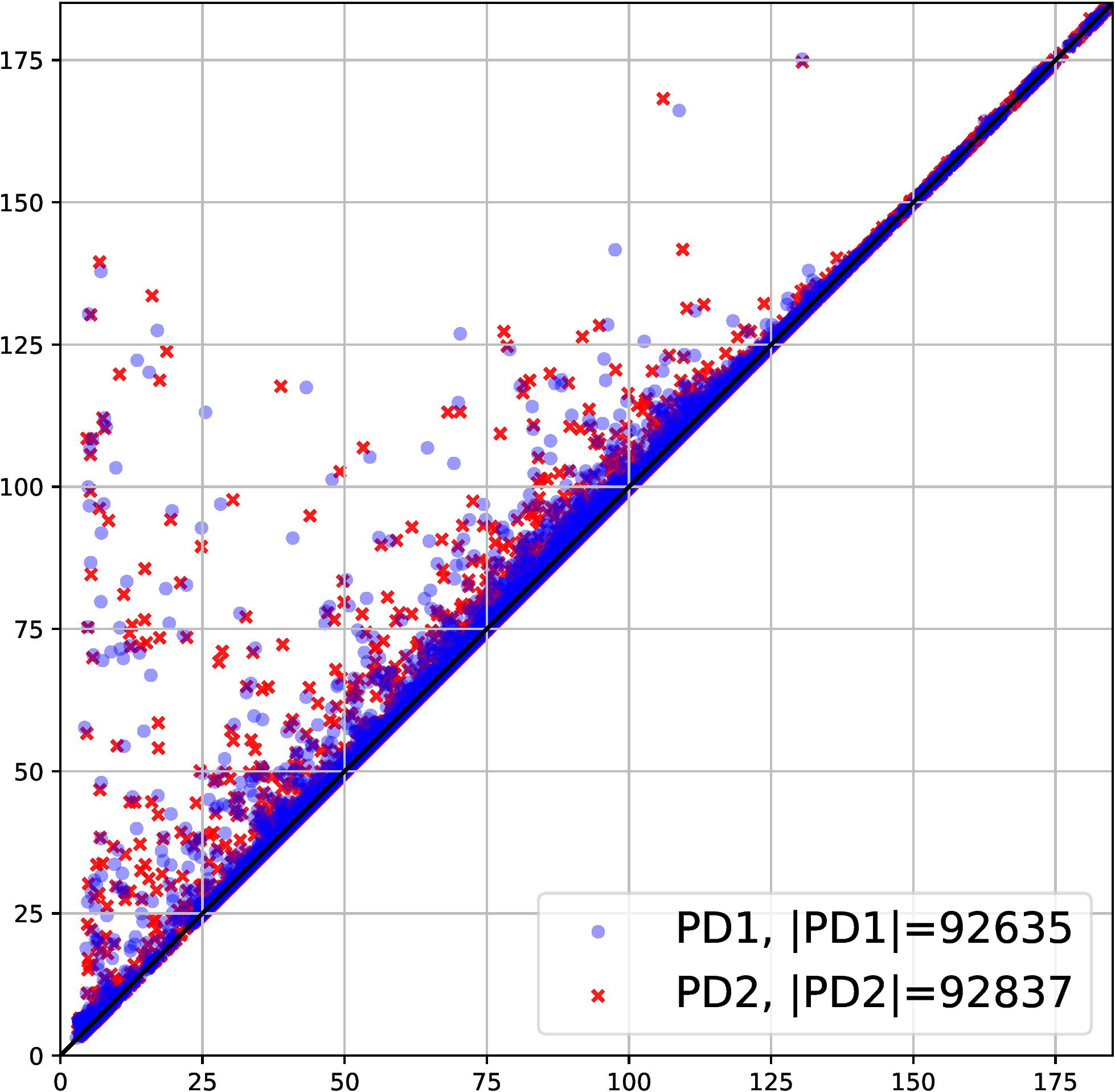}
  
  \subcaption{PD1: {\sf MRI750}, PD2: {\sf MRI751}}
  \label{fig: MRI750xMRI751-p}
  
\end{minipage}
\begin{minipage}{.33\textwidth}
  \centering
  \includegraphics[width=0.6\linewidth]{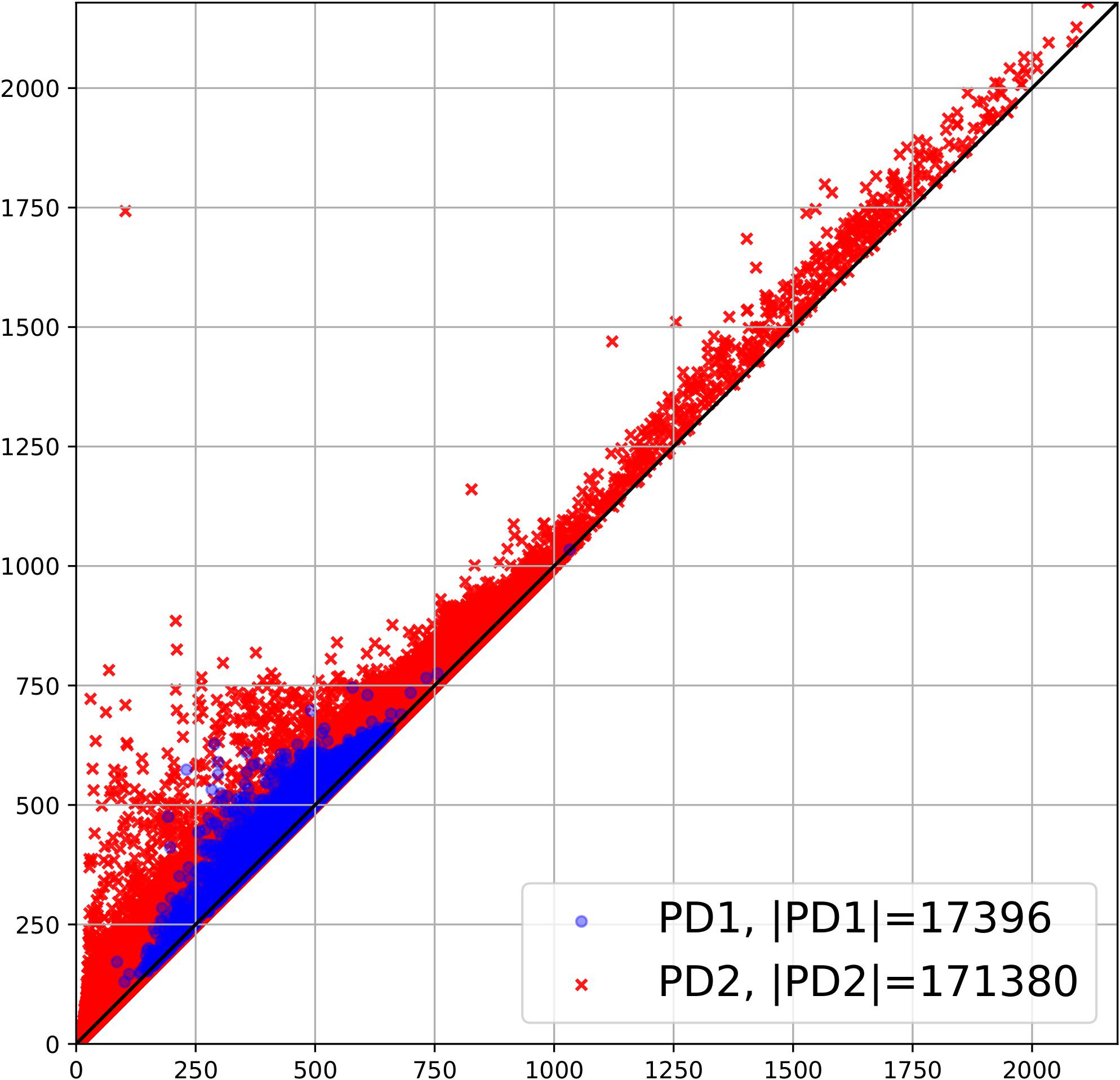}
  
  \subcaption{PD1: {\sf brain}, PD2: {\sf heart}}
  \label{fig: heartxbrain}
  \end{minipage}
\caption{Some of the persistence diagrams; PD1 is in blue and PD2 is in red.}
\label{fig: PDs}
\end{figure*}
All experiments are performed on a high performance computing platform \cite{Pitzer2018}. The node we use is equipped with an NVIDIA Tesla V100 GPU with 32 GB of memory. % that has 5120 FP32 cores and 2560 FP64 cores for single- and double-precision floating-point computation. The GPU has 32 GB High Bandwith Memory 2 (HBM2) that can provide up to 900 GB/s memory access bandwidth. 
The node also has a dual Intel Xeon 8268 with a total of 48 cores where 300 GB of CPU DRAM is used for computing. Table~\ref{table: datasets} describes the persistence diagrams data
we used for all experiments.
    
\begin{table} [h]
\centering
\begin{tabular}{ |p{2.0cm}||p{2.0cm}|p{2.0cm}|p{5.0cm}|p{2.0cm}| }

 \hline
 \multicolumn{5}{|c|}{Datasets} \\
 \hline
 Name& Multiset Card.&  Unique Points & Type of Filtration& Orig. Data\\
 \hline
 {\sf Athens}  & 1281    &1226&   H0 lower star & csv image\\
 {\sf Beijing} &   13141  & 13046   &H0 lower star& csv image\\
 {\sf Brain} &17396&  17291& H1 low. star cubical& 3d vti file\\
 {\sf Heart}    &171380 & 171335&  H1 low. star cubical& 3d vti file\\
 {\sf MRI750} & 92635  &  92635 & H0 low. star pertb. & jpg img.\\
 {\sf MRI751} & 92837  & 92837 &H0 low. star pertb. & jpg img.\\
 {\sf rips1}&  31811  &  31811& H1 Rips&pnt. cloud\\
 {\sf rips2}& 38225  &38225& H1 Rips&pnt. cloud\\
 \hline
 Name & Avg. Card.& Avg. Card. & Type of Filtration & Orig. Data\\
 \hline
 {\sf reddit}& 278.55 & 278.55 & lower/upper star &graphs \\

 \hline
 \end{tabular}
 \caption{Datasets used for all experiments.}
 \label{table: datasets}
 \end{table} 
 
The {\sf Athens} and {\sf Beijing} (producing pair {\sf AB}) are real-world images taken from the public repository of \cite{dey2018graph}. {\sf MRI750} and {\sf MRI751} (producing pair {\sf mri}) are adjacent axial slices of a high resolution 100 micron brain MRI scan taken from the data used in \cite{edlow20197}. The images are saved as csv and jpeg files, respectively. The H0 barcodes of the lower star filtration are computed using ripser.py \cite{tralie2018ripser}. The MRI scans are perturbed by a small pixel value to remove any pixel symmetry from natural images. The {\sf brain} and {\sf heart} (producing pair {\sf bh}) 3d models are vti \cite{ahrens2005paraview} files converted from raw data and then converted to a bitmap cubical complex. The {\sf brain} and {\sf heart} raw data are from \cite{souza2018open} and \cite{andreopoulos2008efficient}.
 The H1 barcodes of the lower star filtration of the bitmap cubical complex are computed with GUDHI \cite{maria2014gudhi}. Datasets {\sf rips1} and {\sf rips2} (producing pair {\sf rips}) consist of 7000 randomly sampled points from a normal distribution on a 5000 dimensional hypercube of seeds 1 and 2 respectively from the numpy.random module \cite{harris2020array}. The Rips barcodes \cite{bauer2019ripser} for H1 are computed by {\sc Ripser++} \cite{zhang2020gpu-proceedings}. The reddit dataset is taken from~\cite{chen2021approximation} and is made up of 200 PDs built from the extended persistence of graphs from the {\sf reddit} dataset with node degrees as filtration height values.

 \iffalse
 The {\sf Athens} and {\sf Beijing} are real-world images taken from the public repository of \cite{dey2018graph}. {\sf MRI750} and {\sf MRI751} are adjacent axial slices of a high resolution 100 micron brain MRI scan taken from the data used in \cite{edlow20197}. The images are saved as csv and jpeg files, respectively. The H0 barcodes of the lower star filtration are computed using ripser.py \cite{tralie2018ripser}. The MRI scans are perturbed by a small pixel value to remove any pixel symmetry from natural images. The {\sf brain} and {\sf heart} 3d models are vti \cite{ahrens2005paraview} files converted from raw data and then converted to a bitmap cubical complex. The {\sf brain} and {\sf heart} raw data are from \cite{souza2018open} and \cite{andreopoulos2008efficient}.
 The H1 barcodes of the lower star filtration of the bitmap cubical complex are computed with GUDHI \cite{maria2014gudhi}. Datasets {\sf rips1} and {\sf rips2} consists of 7000 randomly sampled points from a normal distribution on a 5000 dimensional hypersphere of seeds 1 and 2 respectively from the numpy.random module \cite{harris2020array}. The Rips barcodes \cite{bauer2019ripser} for H1 are computed by Ripser++ \cite{zhang2020gpu-proceedings}. %In our experiments, only PDs from similar sources are compared to analyze performance, e.g. MRI750 is compared with MRI751 but not rips1, see Figure~\ref{fig: PDs}.
 
% \subsection{Experiments for 1-Wasserstein Distance}\label{sec: wasserstein-experiments}
\fi
%We perform extensive experiments with our algorithm for Problem $\ref{prob: wasserstein}$. 
The input to our algorithm contains the parameter $s$ with which we determine a $\delta$ for $\delta$-condensation and construct an $s$-WSPD. The larger the $s$ is, the smaller the average supply of each node in the transshipment network and the denser the network becomes since it has $O(s^2 n)$ number of arcs for $n$ points. Since there is a quadratic dependence on $s$, it is best to use $s\in (0,18]$ on a conventional laptop for memory capacity reasons. Figure~\ref{fig: sxrelerror-Wasserstein} shows the empirical dependence of the relative error $\eps'$ w.r.t. the parameter $s$.
To calculate a tighter theoretical bound $\eps$ on the relative error than Corollary \ref{corollary: approximationbound}, one can solve for $s$ from the expression $1+\eps=(1+\frac{4}{s}+\frac{4}{s-2})\cdot(1+ (\frac{8}{s-4}))$
%Assuming $\eps\leq 1$, from the expression $s=4+\frac{8}{\eps}$ and knowing that we have a $1+\eps'=(1+ \eps)^2$ approximation upper bound, we can solve for s in the equation $\eps'= \frac{2\cdot 8}{s-4}+{(\frac{8}{s-4})}^2$ and obtain $s= \frac{4\cdot (\eps'+2\sqrt{\eps'+1}+2)}{\eps'}$.
%Actually, for $s \leq 4$, we actually cannot guarantee a $(1+\eps)$-spanner. 
%PDoptFlow has fewer guarantees for $\eps>1$ or equivalently $s<12$. However, Algorithm \ref{alg: wasserstein-PDoptFlow} works with $s<12$. For $s<12$, $\eps'$ is set at $1$ for $\delta$-condensation and the WSPD is not guaranteed to form a spanner for $s\leq2$. 

In practice the algorithm performs very well in both time and relative error with $s < 12$, see Figure \ref{fig: sxrelerror-Wasserstein} and Table \ref{table: 1-wasserstein-stats}%(see Appendix~\ref{sec: datasetdesc} for the dataset descriptions)
. Compared to {\sc flowtree} \cite{chen2021approximation}, {\sc PDoptFlow} is surprisingly not that much slower for $n \sim 100K$ and $s\leq 1$ (a very low sparsity factor) and has a smaller relative error. %It is surprising that PDoptFlow can even be for $s\leq 1$ may be due to the near tree-like network structure. 
Since the {\sc flowtree} algorithm only needs one pass through the tree, it is very efficient. On the other hand, our algorithm depends on the cycle structure of the sparsified transshipment network. 
%Furthermore, the relative error of \cite{chen2021approximation} is fixed and cannot be controlled with a parameter. 
The relative error of {\sc PDoptFlow} may heavily depend on the amount of $\delta$-condensation; see {\sf bh}, for example.% for an example of aggressive $\delta$ condensation, see the brainheart comparison. 

The $\delta$-condensation can significantly change the number of nodes in the transshipment network. From the graph $G(A,B)$, the number of nodes in $W_{40}^{PD}(A^\delta,B^\delta)$ can drop by $90\%$, $82\%$, $70\%$, and $2\%$ for the {\sf bh}, {\sf AB}, {\sf mri}, and {\sf rips} comparisons respectively. The great variability is, we presume, determined by the clustering of points in the PDs when the $n$ is not dominating the pooling operation. 

We discussed in Section \ref{sec: deltacond-assumption} that as long as the clustering has some spread to it, the $\delta$-condensation is effective. However, in our experiments, we noticed that having dense points is effective. We presume that this is because the number of points is acting insignificant in scale relative to the distances for the solution. We would expect that for example in the \textsc{rips} dataset, due to the curse of dimensionality, the range of pairwise distances for a random point cloud in high dimensions (5000) is much greater than the distribution of $2^8$ pixel values of a natural image that $\delta$-condensation would be effective. However, we notice that there is almost no clustering of filtration values for the \textsc{rips} dataset, see Table \ref{table: datasets} and Table \ref{table: 1-wasserstein-stats-condense}. Perhaps 60K points is too few to make a difference. In cases like these, $\delta$-condensation is not effective and just a spanner instead so that one may get a tighter theoretical approximation bound. More condensation results in higher empirical relative errors and less computing time. Since only the theoretical relative error is known before execution, we compare times for a given theoretical relative error bound as in Table \ref{table: 1-wasserstein-results}.

\begin{figure*}[t]
\centering
		\includegraphics[width=1.0\columnwidth]{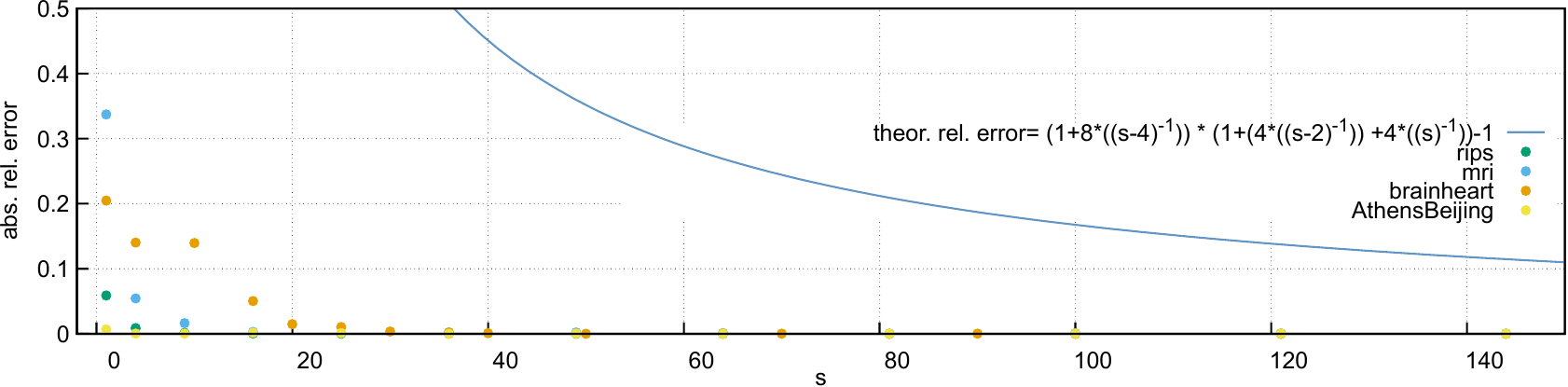}
		\caption{Convergence of {\sc PDoptFlow} for $W_1$-distance against the parameter $s$. }
		\label{fig: sxrelerror-Wasserstein}
		\centering
	\end{figure*}

    \begin{table}[h] 
    \centering
\begin{tabular}{ |p{2.0cm}|||p{3.0cm}|p{3.0cm}||p{3.0cm}|p{3.0cm}|}
 \hline
 \multicolumn{5}{|c|}{$W_1$ Empirical Errors for a given Theoretical Error } \\
 \hline
 
 PD data sets&Emp. Err. $s=40$ (Ours) & Emp. Err. $\eps=0.5$({\sc hera}) &Emp. Err. $s=93$ (Ours)& Emp. Err. $\eps=0.2$({\sc hera}) \\
 
 \hline
 %\\
 
 {\sf bh} & 0.00093 &0.00028 &  0.00014
& 0.000280 \\
 {\sf AB} & 0.00043& 0.00101& 8.6e-5 & 0.000233 \\
 {\sf mri} &0.00224 & 0.00373 &0.00077& 0.001315\\
 {\sf rips} &0.00011&0.00689& 3.4e-5&0.001770 \\
 \hline
 \end{tabular}
 \caption{Empirical relative error of {\sc PDoptFlow} and {\sc hera}.}
 \label{table: 1-wasserstein-stats}
 \end{table}
 
 \begin{table}[h] 
\begin{tabular}{ |p{3.0cm}||p{6.0cm}|p{6.0cm}|}
 \hline
 \multicolumn{3}{|c|}{$W_1$-Distance Computation Stats. for a Guaranteed Rel. Error Bound} \\
 \hline
 
 PD data sets& \%node drop, (\#nodes, \#arcs) for  $W_{40}^{PD}(A^\delta,B^\delta)$ & \%node drop, (\#nodes, \#arcs) for $W_{90}^{PD}(A^\delta,B^\delta)$ \\
 
 \hline
 %\\

 {\sf bh} & 90\%,(18K,22M) &
86\%,(25K,92M)\\
 {\sf AB} & 82\%,(2.5K,1.4M) &70\%,(4.3K,6.1M) \\
 {\sf mri} & 70\%,(55K,57M) & 67\%,(60K,188M)\\
 {\sf rips} & 2\%,(68K,133M)&0.3\%,(69K,468M) \\
 \hline
 \end{tabular}
 \caption{$\delta$-condensation statistics. K: $\times 10^3$, M: $\times 10^6$.}
 \label{table: 1-wasserstein-stats-condense}
 \end{table}

    \subsection{Nearest Neighbor Search Experiments}
    \label{sec: 1-NN}

    We perform experiments in regard to Problem \ref{prob: NN}. NN search is an important problem in machine learning \cite{backurs2020scalable,chen2009similarity}, content based image retrieval ~\cite{lew2006content}, in high performance computing \cite{xiao2016parallel, nene1997simple} and recommender systems~\cite{roumani2007finding}.
    We use the dataset given in \cite{chen2021approximation} which consists of 200 PDs coming from graphs generated by the {\sf reddit} dataset. %We split the dataset into a query set of 100 PDs and search for their NN amongst the remaining 100 PDs.
    Having established ground truth with the guaranteed $0.01$ approximation of {\sc hera}, we proceed to find the nearest neighbor for a given query PD. Following \cite{backurs2020scalable}, we consider various approximations to the $W_1$-distance. We experiment with $6$ different approximations: the Word Centroid Distance ({\sc WCD}), {\sc RWMD} \cite{kusner2015word}, {\sc quadtree}, {\sc flowtree} \cite{chen2021approximation}, {\sc PDoptFlow} at $s=1$ and {\sc PDoptFlow} at $s= 18$ for a guaranteed $2.3$ factor approximation. The {\sc WCD} lower bound is achieved with the observation in \cite{chen2021approximation}. Table \ref{table: 1-NN-accuracy-time} shows the prediction accuracies and timings of all approximation algorithms considered on the {\sf reddit} dataset. Sinkhorn or dense network simplex are not considered in our experiments because they require $O(n^2)$ memory. This is infeasible for large PDs in general.
    
    Although {\sc PDoptFlow} is fast for the error that it can achieve, the computational time to use {\sc PDoptFlow} for all comparisions is still too costly, however. This suggests combining the 7 considered algorithms to achieve high performance at the best prediction accuracy. One way of combining algorithms is through pipelining, which we discuss next.
        \begin{figure}[h]
\includegraphics[width=0.4\columnwidth]{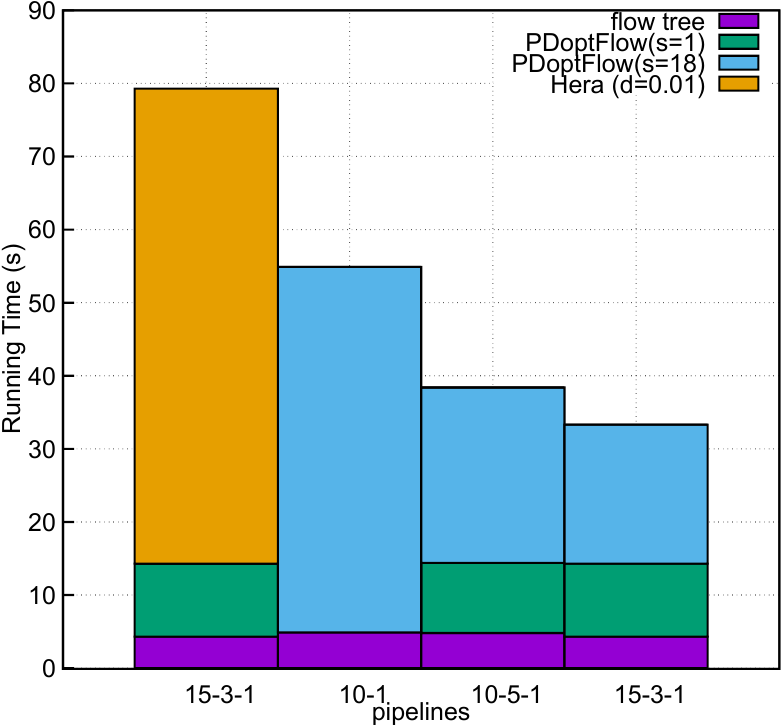}
\centering
		\caption{Pipelines for computing NN.}% $c$-$1$ means that for a given query PD, the closest $c$ PDs are kept. Then, amongst these $c$ PDs, the nearest PD is output. The overall pipeline must be correct $\geq 90\%$ of the time.} %The algorithm is empirically linear in the nusphere of arcs $m =s^2 \cdot n$. }
		\label{fig: pipelines}
		\centering
	\end{figure}
	
    \textbf{Pipelining Approximation Algorithms}:
    Following \cite{backurs2020scalable} and using a distance to compute a set of candidate nearest neighbors, we pipeline these algorithms in increasing order
    of their accuracy to find the $1$-NN with at least $90\%$ accuracy. A  pipeline of $k$ algorithms is written as $c_1-c_2-\cdots -c_k$ where $c_i$ is the number of output candidates %(non-pruned points) 
    of the ith algorithm in the pipeline. 
    
    Since {\sc flowtree} achieves a better accuracy than {\sc RWMD} and {\sc WCD} in less time, we can eliminate {\sc WCD} and {\sc RWMD} from any pipeline experiment. This is illustrated by {\sc WCD} and {\sc RWMD} not being on the Pareto frontier in Figure \ref{fig: pareto-boundary}. The {\sc quadtree} algorithm is not worth placing into the pipeline since its accuracy is too low; it prunes the NN as a potential output PD too early. It also can only save on {\sc flowtree}'s time, which is not the bottleneck of the pipeline. In fact, the last stage of computation, which can only be achieved with a high accuracy algorithm such as {\sc hera} or {\sc PDoptFlow}, always forms the bottleneck to computing the NN. %If hera is used as the final stage of computation, $82\%$ of the time is spent on only $3\%$ of the PDs. This drops to $57\%$ of the time spent if PDoptFlow($s=18$) replaces Hera.  
    
    Figure \ref{fig: pipelines} shows four pipelines involving {\sc flowtree}, {\sc PDoptFlow} and {\sc hera}. The 15-3-1 pipeline consisting of {\sc Flowtree} then {\sc PDoptFlow(s=1)} and then {\sc PDoptFlow(s=18)} was found to be the best in performance through grid search. Three other pipelines computed in the grid search are shown. Each pipeline computes 100 queries with at least $90\%$ accuracy for a random split of the {\sf reddit} dataset. We measure the total amount of time it takes to compute all 100 queries. For the pipeline 15-3-1 with {\sc hera} replacing {\sc PDoptFlow(s=18)}, {\sc hera} takes 65 seconds on 3 queries, while {\sc PDoptFlow} takes 19 seconds on 3 queries. We find that $82\%$ of the time is spent on only $3\%$ of the PDs for {\sc hera}, while $57\%$ of the time is spent if {\sc PDoptFlow}($s=18$) replaces {\sc hera}. We notice that {\sc flowtree} is able to eliminate a large number of candidate PDs in a very short amount of time though it is not able to complete the task of finding the NN due to its low prediction accuracy. {\sc PDoptFlow}($s=1$) surprisingly achieves very good times and prediction accuracies without an approximation bound. %Most of the time is still spent on the final stage where essentially no errors can be made in predicting the NN of the remaining candidates.  

	 \section{Conclusion}
We propose a new implementation for computing the $W_1$-distances
between persistence diagrams that provides a $1+O(\eps)$ approximation. We achieve a considerable speedup for a given guaranteed relative error in computation by two algorithmic and implementation design choices. First, we exploit geometric structures effectively via $\delta$-condensation and $s$-WSPD, which sparsify the nodes and arcs, respectively, when comparing PDs. Second, we exploit parallelism in our methods with an implementation in GPU and multicore. Finally, we establish the effectiveness of the proposed approaches in practice by extensive experiments. Our software {\sc PDoptFlow} can achieve an order of magnitude speedup over other existing software for a given theoretical relative error. Furthermore, {\sc PDoptFlow} overcomes the computational bottleneck to finding the NN amongst PDs and guarantees an $O(1)$ approximate nearest neighbor. One merit of our algorithm is its applicability beyond comparing persistence diagrams. The algorithm is in fact applicable to an unbalanced optimal transport problem on $\mathbb{R}^2$ upon viewing $\bar{b}$ and $\bar{a}$ as creator/destructor and reassigning the diagonal arc distances to the creation/destruction costs.

\bibliography{ms}

\appendix

\clearpage

   \begin{figure*}[h]
    \centering
		\includegraphics[width=0.9 \columnwidth]{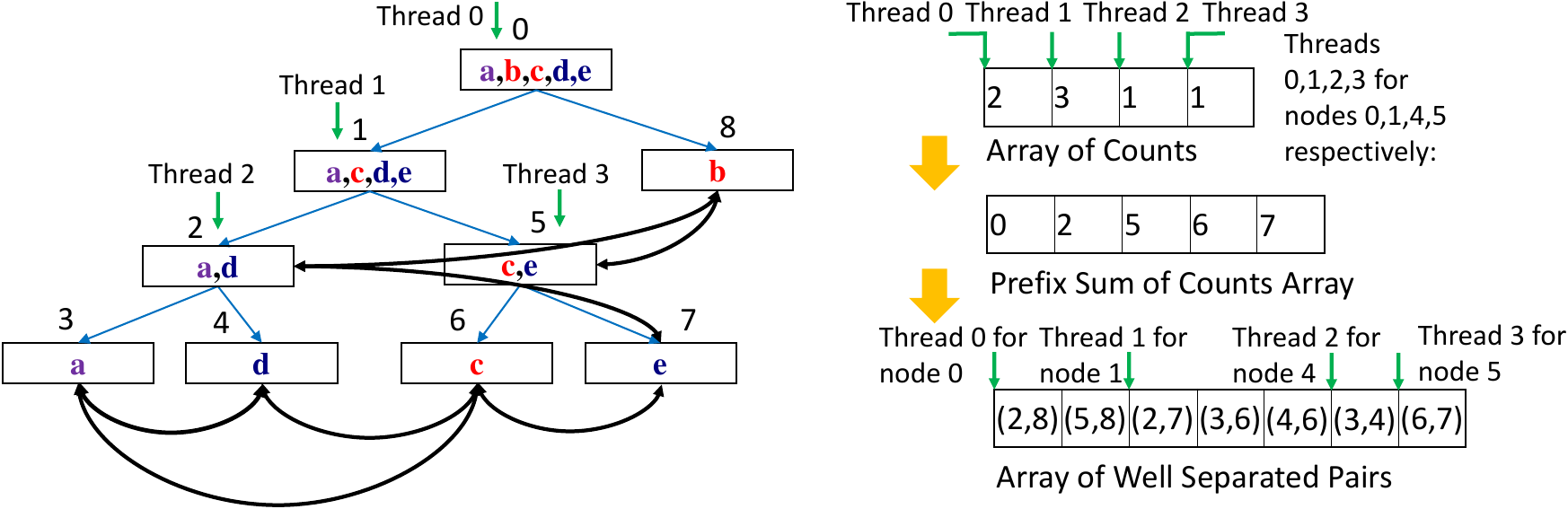}
		\caption{constructing WSPD in parallel for array from the split tree}
		\label{fig: prefixsum-wspd}
		\centering
	\end{figure*}
\section{Appendix}

Here we present the datasets, algorithms, finer implementation details, more experiments, and discussions that are not presented in the main body of the paper due to the space limit.

\subsection{More Algorithmic Details:}
Here we present the algorithmic and implementation details that are omitted in the main context of the paper.% due to space constraints. 

\subsubsection{{\sc WCD}:}

We implement the {\sc WCD} using the Observation in \cite{chen2021approximation} that $OT(A\cup \tilde{B}, B\cup\tilde{A})\leq 2 W_1(A,B)$, where $OT(A\cup \tilde{B}, B\cup\tilde{A})$ is the classical optimal transport distance between $A\cup \tilde{B}$ and $B\cup\tilde{A}$, the sum of distances of their optimal matching, is a $2$ approximation to $W_1(A,B)$. Since {\sc WCD}$(A\cup \tilde{B},B\cup \tilde{A})\leq \text{OT}(A\cup \tilde{B},B\cup \tilde{A})$, we get $\frac{1}{2}$ {\sc WCD}$ \leq W_1(A,B)$. 
    
    {\sc flowtree} is faster than {\sc WCD} on {\sf reddit} due to the small scale of the PDs in that dataset. However asymptotically {\sc WCD} is much faster on very large datasets since it can be implemented as a $O(\log n)$ depth sum-reduction of coordinates on GPU, similar to {\sc quadtree}.

\subsubsection{WSPD and Spanner Construction:}

\iffalse
		\begin{algorithm}[H]
\label{alg: wspd-construction}
\SetAlgoLined
\KwIn{$\pmb{T}$ a split tree, WSPD parameter $s$}
\KwOut{$s$-$WSPD$ represented by wspd-ptn-pairs as an array}

        counts $\gets \{ 0...0\}$ \Comment{allocate O(n) elements}
        
			\For{node $n \in \pmb{T}$ in parallel}{ 
			count-WSPD$(tid(n),$n.left.n,right,$s$,counts) \Comment{$tid(\cdot )$ is injective}
			}
			
			offsets $\gets$ prefix-sum(counts)
	        
	        $L$= offsets[n] \Comment{offsets[n]=sum(counts)}
			
			wspd-ptn-pairs $\gets \{ ...\}$ \Comment{allocate $L$ elements for wspd-ptn-pairs: $O(s^2 n)$}
			
			\For{node $n \in \pmb{T}$ in parallel}{
			construct-WSPD(tid(n),n.left,n.right,s,offsets, wspd-ptn-pairs)
			}

\caption{Construct $s$-WSPD-biarcs in parallel }
\end{algorithm}
\fi
	Here we present our simplified parallel algorithm for WSPD construction used in our implementation. %A faster parallel \sout{complexity} \simon{(it actually does have lower complexity)} algorithm is presented in \cite{wang2021fast}.\tamal{Then, why are we not adopting the better algorithm?}}\simon{(we present our simplified algorithm in this section. The other algorithm is not ours.)} 
	The purpose of traversing the split tree twice is to parallelize writing out the WSPD, the bottleneck to constructing a WSPD. Although the WSPD is linear in $n$, the number of nodes of $\WS_s^{PD}(A^{\delta},B^{\delta})$, in practice the size of the WSPD is several orders of magnitude larger than $n$. Thus writing out the WSPD requires a large amount of data movement. Algorithm \ref{alg: wspd-count} first finds the number of pairs written out by a thread rooted at some node in the split tree. The computation of counts is in parallel and is mostly arithmetic. Once the counts are accumulated, a prefix sum of the counts is computed and written out to an offsets array. The offsets are then used as starting memory addresses to write out the WSPD pairs for each thread in parallel.

	Figure \ref{fig: prefixsum-wspd} illustrates the parallel computation of the WSPD. The prefix sum is computed over the counts determined by each thread. There is a thread per internal node.
	
	In our implementation, we do not actually keep track of the point subsets for each node of the split tree. Instead, we keep track of a single point in each point subset $P \subset \hat{A} \cup \hat{B}$ as well as a bounding box of $P$. This constructs the non-diagonal arcs of $\WS_s^{PD}(A,B)$ with minimal data.

\iffalse
		\begin{algorithm}[H]
\label{alg: wspd-count}
\SetAlgoLined

\KwIn{$tid$: unique thread identifier of the root of the computation down the split tree; nodes $u$ and $v$ in the split tree; $s$: the parameter for the WSPD; counts: a global readable and writable array}
	\KwOut{counts: array of counts, counts$[tid]$= number of pairs each thread will find}
	
            \If{$u$ is $s$-well separated from $v$}{%Comment{two single points is well separated}
            %\State $wspd \gets wspd \cup \{u.point,v.point\}$
            counts$[tid]$++
            
            \Return\Comment{keep track of the number of well separated pairs associated with $tid$}
            }
            \uIf{max\_length(bounding\_box(u))>max\_length(bounding\_box(v))}
            {
            $\text{count-WSPD}(tid,u.\text{left}, v,s,$counts)
            
            $\text{count-WSPD}(tid,u.\text{right}, v,s,$counts)
            }
            \Else
            {
            $\text{count-WSPD}(tid,u, v.\text{left},s,$counts)
            
            $\text{count-WSPD}(tid,u, v.\text{right},s,$counts) 
            }
            \caption{compute WSPD thread counts for offsets}
\end{algorithm}
\fi

\begin{algorithm}[h]

\begin{algorithmic}[1]

\Require{$\pmb{T}$ a split tree, WSPD parameter $s$}

\Ensure{$s$-$WSPD$ represented by wspd-ptn-pairs as an array}

\State counts $\gets \{ 0...0\}$ \Comment{allocate O(n) elements}
        
			\For{node $w \in \pmb{T}$ in parallel}
			\State count-WSPD$(tid(w),$w.left,.w.right,$s$,counts) %\Comment{$tid(\cdot )$ is injective}
			\EndFor

			\State offsets $\gets$ prefix-sum(counts)
	        
	        \State $L$= offsets[w] \Comment{offsets[w]=sum(counts)}
			
			\State wspd-ptn-pairs $\gets \{ ...\}$ \Comment{allocate $L$ elements for wspd-ptn-pairs: $O(s^2 n)$ memory}
			
			\For{node $w \in \pmb{T}$ in parallel}
			\State construct-WSPD(tid(w),w.left,w.right,s,offsets, wspd-ptn-pairs)
			\EndFor
\end{algorithmic}
\caption{Construct $s$-WSPD-biarcs in parallel}
\label{alg: wspd-construction}
\end{algorithm}

\begin{algorithm}[h]

\begin{algorithmic}[1]

	\Function{count-WSPD}{}
	\Require{$tid$: thread id; 
%unique thread identifier of the root of the computation down the split tree; 
nodes $u$ and $v$ in the split tree; $s$: WSPD parameter; counts: the number of recursive calls made by each thread; 
%a global readable and writable array
}
\Ensure{counts: array of counts, counts$[tid]$= number of pairs each thread will find}
            \If{$u$ is $s$-well separated from $v$}%Comment{two single points is well separated}
            %\State $wspd \gets wspd \cup \{u.point,v.point\}$
            \State counts$[tid]$++
            
            \Return\Comment{keep track of the number of well separated pairs associated with $tid$}
            \EndIf
            \If{$\text{max\_len(BndingBx(u))}>\text{max\_len(BndingBx(v))}$}
            \State $\text{count-WSPD}(tid,u.\text{left}, v,s,$counts)
            
            \State $\text{count-WSPD}(tid,u.\text{right}, v,s,$counts)
            \Else
            \State $\text{count-WSPD}(tid,u, v.\text{left},s,$counts)
            
            \State $\text{count-WSPD}(tid,u, v.\text{right},s,$counts) 
            \EndIf
            \EndFunction
            \end{algorithmic}
            \caption{Compute WSPD thread counts for offsets}
            \label{alg: wspd-count}
\end{algorithm}

\begin{algorithm}[H]

\begin{algorithmic}[1]

	\Function{construct-WSPD}{}
	
	\Require{$tid$: %unique thread id of the root of the computation down the split tree; 
thread id; nodes $u$ and $v$ in the split tree; $s$: WSPD parameter; offsets: %a prefix sum of the counts of number of pairs written out by each thread, also the memory addresses to write out pairs; 
wspd: a writable array of point pairs;
}
\Ensure{$s$-WSPD with representatives of point pairs as an array}
	
	\If{$u$ is $s$-well separated from $v$}
            %\Comment{two single points are always well separated}
            %\Comment{points $u.\text{point}$ and $v.\text{point}$ are arbitrary point representatives of $u$ and $v$}
	
	\State $wspd[$offsets$[tid]$++] $\gets (u.\text{point},v.\text{point})$ \Comment{all threads write in parallel}
            
            \Return
            \EndIf
	
	\If{max\_len(BndingBx($u$))$>$max\_len(BndingBx)}
	
	    \State $\text{construct-WSPD}(tid,u.\text{left}, v,s,$wspd)
        
        \State $\text{construct-WSPD}(tid,u.\text{right}, v,s,$wspd)
        \Else
        
        \State $\text{construct-WSPD}(tid,u, v.\text{left},s,$wspd)
            
        \State $\text{construct-WSPD}(tid,u, v.\text{right},s,$wspd)
	\EndIf
	\EndFunction
	\end{algorithmic}
\caption{Write out WSPD from offsets}
\label{alg: wspd-write}
\end{algorithm}

\iffalse
\begin{algorithm}[H]
\label{alg: wspd-write}
\SetAlgoLined

\KwIn{$tid$: unique thread identifier of the root of the computation down the split tree; nodes $u$ and $v$ in the split tree; $s$: the parameter for the WSPD; offsets: a prefix sum of the counts of number of pairs written out by each thread, also being the memory addresses to write out pairs; wspd: a writable array of point pairs;}

	\KwOut{$s$-WSPD with representatives of point pairs as an array}
	
	\If{$u$ is $s$-well separated from $v$}{
            %\Comment{two single points are always well separated}
            \Comment{points $u.\text{point}$ and $v.\text{point}$ are arbitrary point representatives of $u$ and $v$}
	
	$wspd[$offsets$[tid]$++] $\gets (u.\text{point},v.\text{point})$ \Comment{all threads write in parallel}
            \Return
            }
	
	\uIf{max\_length(bounding\_box($u$))>max\_length(bounding\_box($v$))}{
	
	    $\text{construct-WSPD}(tid,u.\text{left}, v,s,$wspd)
        
        $\text{construct-WSPD}(tid,u.\text{right}, v,s,$wspd)
        }\Else{
        
        $\text{construct-WSPD}(tid,u, v.\text{left},s,$wspd)
            
        $\text{construct-WSPD}(tid,u, v.\text{right},s,$wspd)
	}
\caption{write out WSPD from offsets}
\end{algorithm}
\fi

Algorithm \ref{alg: W1-form-diagonal-arcs} shows how to write out the diagonal arcs for $\WS_{s}^{PD}(A^{\delta},B^{\delta})$. On line 2 it states that there is a parallelization by prefix sum on arc counts. This computation is similar to the algorithm for WSPD construction. The number of arcs per point is kept track of. A prefix sum is computed after this and the diagonal arcs are written out per point.

\begin{algorithm}
\begin{algorithmic}[1]
\For{point $p \in \pmb{P}= \hat{A}_{\delta} \cup \hat{B}_{\delta}$ parallelized by prefix sum on count of arcs incident on each $p$}
            \If{p is from $\hat{A}_{\delta}$}
            diag-arcs $\gets$ diag-arcs $\cup \{p,p_{proj}\}$
            \EndIf
            \If{p is from $\hat{B}_{\delta}$}
            diag-arcs $\gets$ diag-arcs $\cup \{p_{proj},p\}$
            \EndIf
    \EndFor
\end{algorithmic}
\caption{Form diagonal arcs}
\label{alg: W1-form-diagonal-arcs}

\end{algorithm}
\iffalse
\begin{algorithm}[H]
		\label{alg: W1-form-diagonal-arcs}
			\KwIn{$\pmb{P}$ point set}
			\KwOut{ diag-arcs, a set of diagonal arcs
			}
            diag-arcs $\gets \emptyset$
            
            \For{point $p \in \pmb{P}$ parallelized by prefix sum on count of arcs incident on each $p$}{
            \If{p is from A}{
            diag-arcs $\gets$ diag-arcs $\cup \{p,p_{proj}\}$
            }
            \If{p is from B}{
            diag-arcs $\gets$ diag-arcs $\cup \{p_{proj},p\}$
            }
            }
    \caption{form-diagonal-arcs}
\end{algorithm}
	\fi

\begin{figure*}[h]
		\includegraphics[width=1.0\columnwidth]{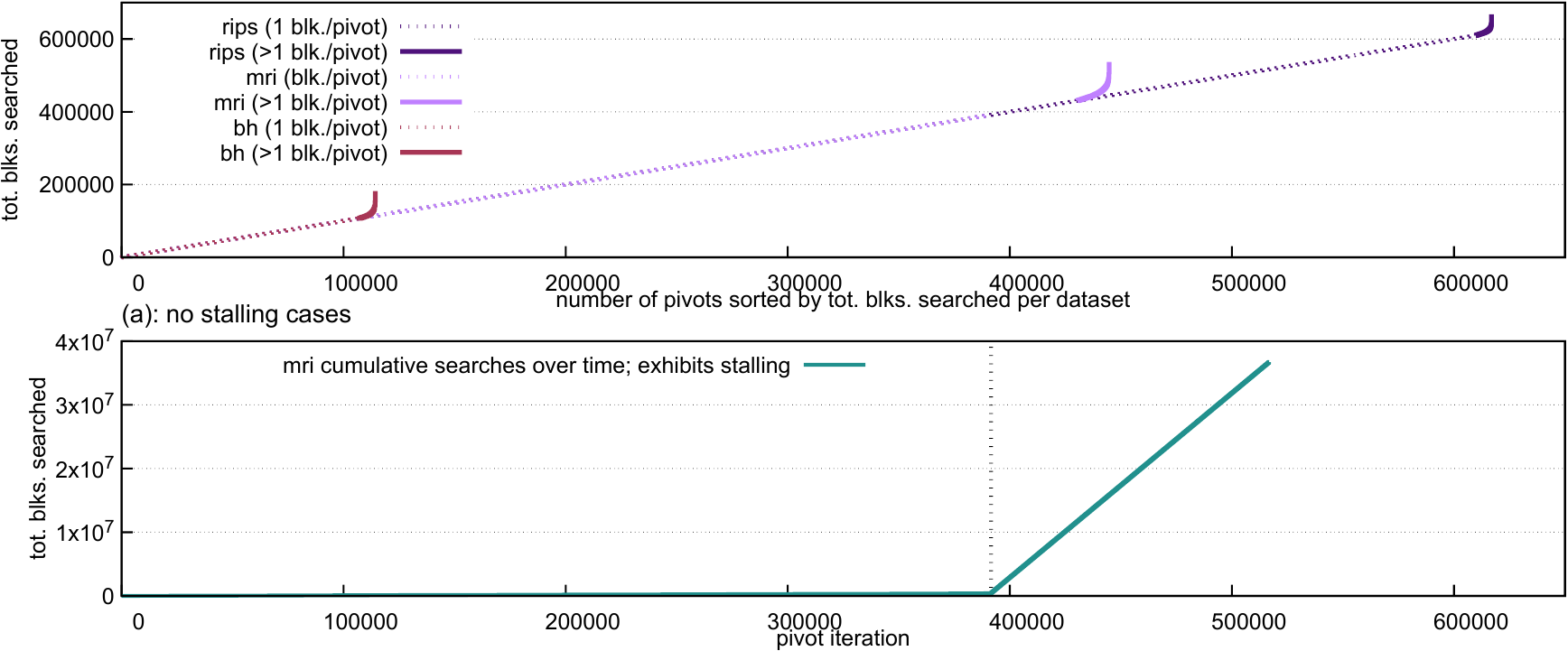}
		\centering
		\caption{(a) Plot of no stalling case of the cumulative distribution of blocks searched for {\sf rips}, {\sf mri} and {\sf brain-heart} datasets. (b) Plot of a stalling case for the {\sf mri} dataset. block size= $\sqrt{m}$}
		\label{fig: sorted-cumulativedist-blocks}
		\centering
	\end{figure*}
\subsubsection{Representing the Transshipment Network:}
The data structure used to represent the transshipment network significantly affects the performance of network simplex algorithm. Since most of the time of computation is spent on the network simplex algorithm and not the network construction stage, the network data structure is designed to be constructed to be as efficient for arc reading and updating as possible. A so-called static graph representation~\cite{dezsHo2011lemon}, essentially a compressed sparse row (CSR) \cite{eisenstat1984new} format matrix, is used to represent the transshipment network. Thus in order to build a CSR matrix, we must sort the arcs $(u,v)$ first by first node followed by second node in case of ties. This sorting can be over several millions of arcs, see Table \ref{table: 1-wasserstein-results} column 2. For example, for rips at $\eps' \leq 0.2$, 468 million arcs must be sorted. ($\eps'$ is the guaranteed relative error bound). For a sequential $O(m \log m)$ algorithm, this would form a bottleneck to the entire algorithm before network simplex, making the algorithm $\Omega(m \log (m))$. Thus we sort the arcs on GPU using the standard parallel merge sorting algorithm~\cite{blelloch2010parallel,cole1988parallel} and achieving a parallel depth complexity of $O(\log m)$.

\subsection{Computational Behavior of Network Simplex (BSP in practice):}
\label{sec: N.S. behavior}
Refer to Section \ref{sec: experiments} and Table \ref{table: datasets} for dataset information. Figure \ref{fig: sorted-cumulativedist-blocks} shows two very different computational patterns of the block search pivot based {\sc NtSmplx} algorithm. Figure \ref{fig: sorted-cumulativedist-blocks}(a) shows the vast majority of cases when there is no stalling. We show the cumulative distribution of blocks searched for the {\sf rips}, {\sf mri} and {\sf brain-heart} datasets at s=20, 49 and 150 respectively. The block sizes are set to the square root of the number of arcs; the block sizes are 6539, 9134 and 8922 respectively. Notice that 98.9\%, 96.8\% and 93.8\% of the pivots involve only a single block being searched, and account for 91.4\%, 80.2\% and 58.8\% of the total blocks searched. Although the pivots are sorted per dataset by the number of blocks searched, the cumulative distribution depending on the pivots computed over execution is almost identical. Figure \ref{fig: sorted-cumulativedist-blocks}(b) shows the relatively rare but severe case of stalling for the {\sf mri} dataset at s=36, stopped after 10 minutes. % On the plot, corresponding to the sharp increase from a slope 1 curve, 
Stalling begins at the 391559th arc found. %In our experiments, all stalling exhibits the behavior of, in addition to repeated degenerate pivots, repeated searches over many blocks per pivot search. This results in exponentially many pivot searches involving more than one block.  

\begin{figure*}[h]
		\includegraphics[width=1.0
		\columnwidth]{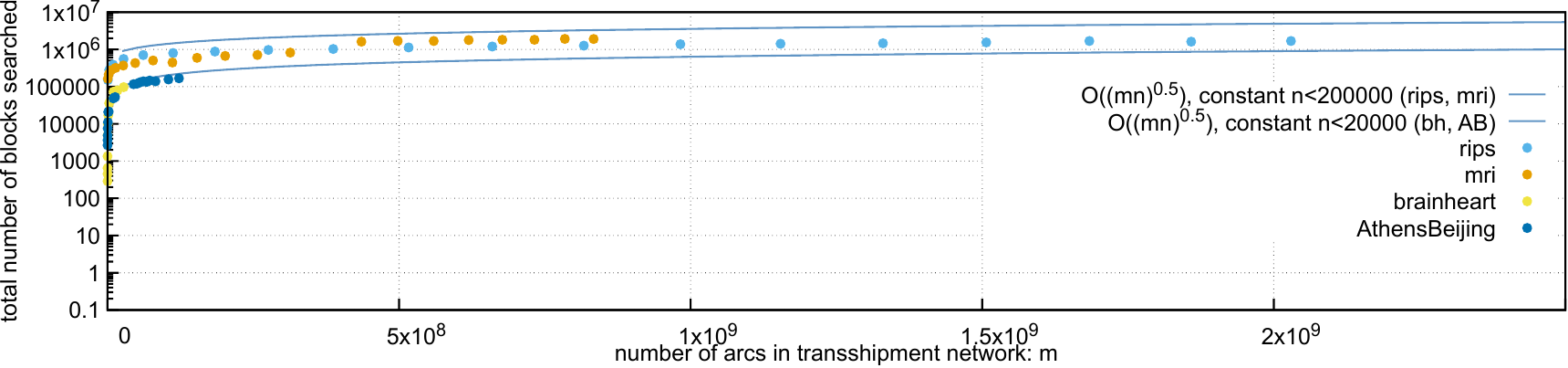}
		\centering
		\caption{Plot of the total number of block pivot searches depending on the number of arcs.}
		\label{fig: arcsXtotalblockssearched}
		\centering
	\end{figure*}

Furthermore, we have noticed empirically that repeated tie breaking of reduced costs during pivot searching results in a tendency to stall. In fact, most implementations simply repeatedly choose the smallest indexed arc for tie breaking. After applying lattice snapping by $\pi_{\delta}$, symmetry is introduced into the pairwise relationships and thus results in many equivalent costs on arcs and subsequent reduced costs. This is why we introduce a small perturbation to the snapped points in order to break this symmetry. This results in much less stalling in practice.    

\subsubsection{Parallelizing Network Simplex Algorithm:}
\label{sec: parallel-NS}
Network simplex is a core algorithm used for many computations, especially exact optimal transport. This introduces a natural question: can we directly parallelize some known network simplex pivot search strategies and gain a performance improvement? We attempted to implement parallel pivot search strategies such as a $O(\log m)$-depth parallel min reduction over all reduced costs on either GPU or multicore, such as in \cite{ploskas2014gpu}. These approaches did not improve performance over a sequential block pivot search strategy. There was speedup over Dantzig's pivot strategy, where all admissible arcs are checked, however. For GPU, there is an issue of device to host and host to device memory copy. These IO operations dominate the pivot searching phase and are several of orders of magnitude slower than a single block searched sequentially from our experiments. Recall that most searches result in a single block by Figure \ref{fig: sorted-cumulativedist-blocks}. For multicore, there is an issue of thread scheduling which provides too much overhead. In general, it is very difficult to surpass the performance of a sequential search over a single block when the block size can fit in the lower level cache due to the two aforementioned issues. For example, in our experiments the cache size is 28160KB, which should hold $6 \cdot B \cdot 8$ bytes for $B= \sqrt{m}$, the block size, and $m< 3 \times 10^{11}$ where 6 denotes the 6 arrays needed to be accessed to compute the reduced cost and 8 is the number of bytes in a double. This bound on $m$, the number of arcs, should hold for almost all pairs of conceivable input persistence diagrams and $s>0$ in practice. This does not preclude, however the possibility of efficient parallel pivoting strategies completely since stalling still exists for the sequential block search algorithm.

\subsubsection{Empirical Complexity:} 
\label{sec: empirical complexity}

\begin{figure*}[h]
\includegraphics[width=1.0\columnwidth]{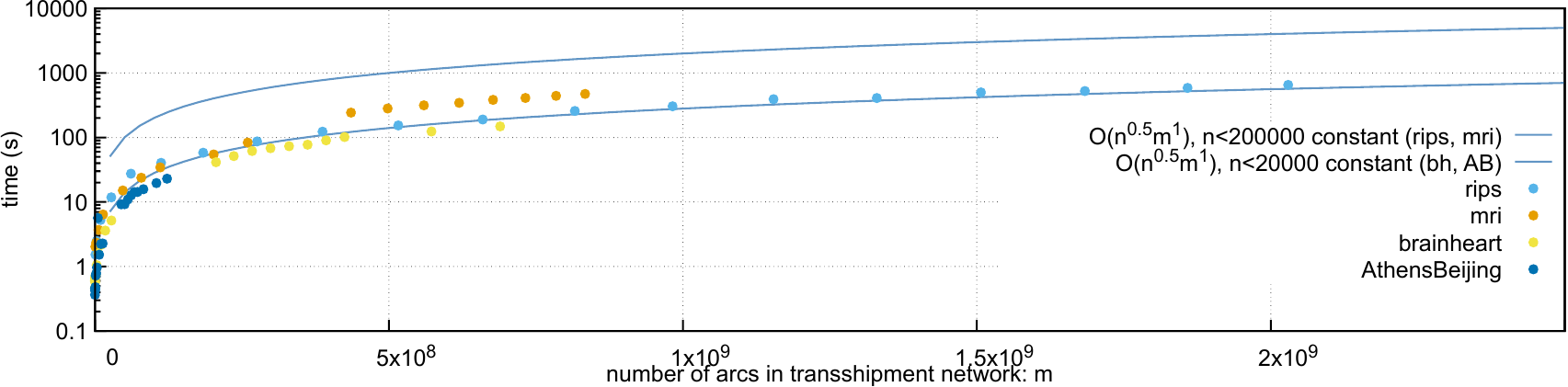}
\centering
		\caption{Plot of the empirical time (log  scale) depending on the number of arcs of the sparsified transshipment network for each dataset. $n$ is the number of nodes.} %The algorithm is empirically linear in the nusphere of arcs $m =s^2 \cdot n$. }
		\label{fig: arcsXtime}
		\centering 
	\end{figure*}
	
	For each of the datasets from Table \ref{table: datasets}, our experiments illustrated in Figure \ref{fig: arcsXtime}, show that for varying $s$ and fixed $n$, our overall approach runs empirically in $O(\sqrt{n}m)$ time, where $m=s^2 n$ with $s$ the WSPD parameter and $n$ the number of nodes in the sparsified transshipment network.

\iffalse	
\begin{figure*}[h]
\includegraphics[width=1.3\columnwidth]{figures/plots/nXtime-crop.pdf}
\centering
		\caption{Plot of the empirical time (log scale) depending on the number of nodes $n$ of sparsified transshipment network of synthetic Gaussian PDs.} %The algorithm is empirically linear in the nusphere of arcs $m =s^2 \cdot n$. }
		\label{fig: nXtime-appendix}
		\centering
	\end{figure*}
	\fi
	
	Here we explain in more detail the experiment illustrated in Figure \ref{fig: nXtime}. We determine the empirical complexity with respect to the number of points on a synthetic Gaussian dataset. These are not real persistence diagrams and are made up of points randomly distributed on the plane above the diagonal. The points follow a Gaussian distribution. For fixed $s\leq40$, as a function of $n$ our algorithm empirically is upper bounded by $O(s^2n^{1.5})$. This is determined through upper bounding the least squares curve fitting. %The upper dependency of $O(n^{1.5})$ appears to optimistically match the  known theoretical complexity of $\tilde{O}(m+n^{1.5})$ of~\cite{brand2021minimum} for min-cost flow, the determining factor of time for these sparsified transshipment networks. 
	Since we are still solving a linear program, it should not be expected that the empirical complexity can be truly linear, except perhaps under certain dataset conditions. The proximity of points, for certain real persistence diagrams, for example, could be exploited more by $\delta$-condensation. We notice that the empirical complexity is better, the smaller the $s$, including for $s\leq 40$. This is why in Section \ref{sec: experiments} {\sc PDoptFlow} for $s=1$ performs so much faster than {\sc PDoptFlow} for $s=18$. %This is not as prominent for a synthetic Gaussian dataset. In fact, $\delta$ condensation only achieves a drop of 1\% for the synthetic Gaussian for s=40 and n=100000.

\subsubsection{Stopping Criterion:}
\label{sec: stoppingcriterion}
Due to the rareness of stalling for given $s$ in practice, our stopping criterion is designed to justify the empirical time bound. If the block size is $\sqrt{m}$, the computation goes like $O(s^2n^{1.5})$, and each iteration within a block search determines the time, $O(\frac{s^2n^{1.5}}{s\sqrt{n}})= O(\sqrt{mn})$ blocks is an upper bound on the number of searched blocks when there is no stalling. Figure \ref{fig: arcsXtotalblockssearched} illustrates this relationship amongst $m$, $n$ and the time. Thus the stopping criterion is set to $C \sqrt{mn}+b$. In practice, $C$ may simply be set to $0$ and $b$ set to a large number however it has been empirically found that the stopping citerion goes like $\sqrt{mn}$ blocks for a large number of the various types of real persistence diagrams such as those generated by the persistence algorithm on lower star filtrations induced by images and rips filtrations on random point clouds, to name the types from the experiments.

\subsubsection{Bounds on Min-Cost Flow:}
 %The 1-Wasserstein distance between PDs is not the same as the Earth Mover's Distance (EMD)\cite{rubner2000earth} between point sets with a ground metric and thus we cannot directly apply methods from geometric optimal transport or geometric matching theory e.g. \cite{agarwal2004near, agarwal2019faster, fox2019near,khesin2019preconditioning} or metric embedding methods such as in \cite{dong2019scalable, indyk2003fast}. Instead,
 The $W_1$-distance between PDs is a special case of the unbalanced optimal transport (OT) problem as formulated in \cite{lacombe2018large,sato2020fast}. %and can also be formulated in terms of optimal transport with arbitrary edge costs. 
 %In \cite{sato2020fast} it is proven that under the strong exponential time hypothesis, the unbalanced OT problem between point sets using the $l_1$ ground metric cannot be solved exactly in subquadratic time. This suggests that in order to reach subquadratic theoretical complexity, we should approximate.
 Solving such a problem exactly using min-cost flow is known to take cubic complexity~\cite{orlin1993faster} in the number of points. However, affording cubic complexity is usually infeasible in practice and thus we seek a subcubic solution.

 There are several approaches to approximating the  distance between PDs with $n$ total points. 
 %\simon{sato is wrong in that he assumes a constant distance to the diagonal}
 %\cite{sato2020fast} develops a $O(n \cdot \log^2(n))$ time algorithm that can reduce the problem to computing the unbalanced optimal transport distance on the tree metric for a fixed $\log(\Delta)$ approximation, where $\Delta$ is the aspect ratio. 
 In \cite{chen2021approximation}, a $\log A$ approximation is developed, where, $A$ is the aspect ratio, adapting the work of \cite{indyk2003fast} and \cite{backurs2020scalable} for persistence diagrams. In \cite{kerber2017geometry}, the auction matching algorithm performs a $(1+\eps)$ approximation, also lowering complexity by introducing geometry into the computation. Geometry lowers a linear search over $O(n)$ points for nearest neighbors to $O(\sqrt{n})$ via kd-tree. This does not lower the theoretical bound below $O(n^{2.5})$, however. Our approach lowers complexity by introducing a geometric spanner \cite{cabello2005matching}, using a linear number of arcs between points. 
 
  Min-cost flow algorithms can have theoretically very low complexity. The input to min-cost flow is a transshipment network and its output is the minimum cost flow value. Let $m$ be the number of arcs in the transshipment network and $n$ its number of nodes. It was shown that min-cost flow can be found exactly in $\tilde{O}(m+n^{1.5})$ complexity in \cite{brand2021minimum}, by network simplex in $\tilde{O}(n^2)$ complexity, in parallel in $\tilde{O}(\sqrt{m})$ and approximated on undirected graphs in \cite{bernstein2021deterministic} in $\tilde{O}(m^{1+o(1)})$ complexity.
 
 Since the number of nodes and arcs of the transshipment network depend directly on the points and pairwise distances respectively, an implication of using a geometric spanner for $(1+O(\eps))$ approximation is that the complexity becomes theoretically subcubic and requiring $O(h(\frac{1}{\epsilon})n)$ memory, $h$ a constant degree polynomial. In fact this bound is actually achieved in practice. We show the empirical complexity is actually similar to $O(s^2n^{1.5})$ as shown in Figure \ref{fig: arcsXtime} and Figure \ref{fig: nXtime} but only for small $s$. %This is the motivation behind our algorithm design, however our actual implementation  
 
 %Min-cost flow is typically computed with the network simplex algorithm, which exhibits good performance in practice however has cubic worst case complexity.   

%MAKE NOTE IN THE EXPERIMENTS THAT THETA GRAPHS ARE INFERIOR TO WSPD IN CONSTRUCTION TIME (HOURS vs seconds)

\end{document}